\documentclass[a4paper,UKenglish]{lipics}
\usepackage[utf8]{inputenc}
\usepackage{amsmath}

\usepackage{multirow}
\usepackage{amssymb}
\usepackage{stmaryrd}
\usepackage{booktabs}
\usepackage{xspace}
\usepackage{placeins}
\usepackage{enumerate}
\usepackage{verbatim}
\usepackage{enumerate}
\usepackage{mathtools}
\usepackage{url}
\usepackage{thm-restate}
\usepackage{color}
\usepackage{wrapfig}
\usepackage{bbm}
\usepackage{thm-restate}

\DeclareMathAlphabet{\mathpzc}{OT1}{pzc}{m}{it}

\usepackage{graphicx}
\usepackage{xcolor,pgf,tikz}
\usetikzlibrary{arrows,automata,decorations.pathmorphing,snakes}

\title{On Equivalence and Uniformisation Problems for Finite Transducers}
\author[1]{Emmanuel Filiot}
\author[1]{Ismaël Jecker}
\author[2]{Christof Löding}
\author[2]{Sarah Winter}
\affil[1]{Universit\'{e} Libre de Bruxelles} \affil[2]{RWTH Aachen}

\authorrunning{E.\, Filiot and I.\, Jecker and C.\, Löding and S.\, Winter} %mandatory. First: Use abbreviated first/middle names. Second (only in severe cases): Use first author plus 'et. al.'

\Copyright{Emmanuel Filiot and Ismaël Jecker and Christof Löding and Sarah Winter}%mandatory, please use full first names. LIPIcs license is "CC-BY";  http://creativecommons.org/licenses/by/3.0/

% mandatory: Please choose ACM 1998 classifications from http://www.acm.org/about/class/ccs98-html . E.g., cite as "F.1.1 Models of Computation". 
% mandatory: Please provide 1-5 keywords
\serieslogo{}%please provide filename (without suffix)
\volumeinfo%(easychair interface)
  {Billy Editor and Bill Editors}% editors
  {2}% number of editors: 1, 2, ....
  {Conference title on which this volume is based on}% event
  {1}% volume
  {1}% issue
  {1}% starting page number
\EventShortName{}
\DOI{10.4230/LIPIcss.xxx.yyy.p}% to be completed by the volume editor
\newtheorem{proposition}{\textbf{Proposition}}
\newtheorem{claim}{\textbf{Claim}}

\newcommand\trans[2]{\ensuremath{#1\mid #2}}

\newcommand\aut{A}
\newcommand\tra{T}
\newcommand\uni{U}
\newcommand\sync{\mathbb{S}}
\newcommand\dsync{\mathbb{D}}

\newcommand\inp{\mathbbmtt{i}}
\newcommand\outp{\mathbbmtt{o}}
\newcommand\dom{\text{dom}}
\newcommand\id{\text{id}}

\newcommand\inr{\text{in}}
\newcommand\outr{\text{out}}

\newcommand\del{\text{delay}}
\newcommand\lag{\text{lag}}

\newcommand{\lang}[1]{\mathcal{L}_{#1}}
\newcommand{\rel}[1]{\mathcal{R}_{#1}}

\newcommand\ina{\Sigma_\inp}
\newcommand\outa{\Sigma_\outp}
\newcommand\inout{{\inp\outp}}
\newcommand\inouta{\Sigma_\inout}

\def\stackrel#1#2{\mathrel{\mathop{#2}\limits^{#1}}}
\newcommand{\nat}{\mathbb{N}}
\newcommand{\Qin}{Q^{\inp}}
\newcommand{\Qout}{Q^{\outp}}
\newcommand{\Lin}{L^{\inp}}
\newcommand{\Lout}{L^{\outp}}
\newcommand{\Vin}{V^{\inp}}
\newcommand{\Vout}{V^{\outp}}
\newcommand{\Sigmaend}{\Sigma_{\dashv}}

\newcommand{\game}{\mathcal{G}}
\newcommand{\reduce}[1]{\inp(#1)}

\newcommand{\Pump}[2]{\mathit{Pump}_{\ge #1}(#2)}

\newcommand{\choice}[2]{\textsf{Ch}_{#1}^{#2}}
\newcommand{\prop}[1]{\textbf{\textsf{P}$_{#1}$}}
\newcommand{\propH}[1]{\textbf{\textsf{R}$_{#1}$}}
\newcommand{\outru}[1]{\textsf{out}_{#1}}
\newcommand{\outf}[1]{\textsf{fin}_{#1}}
\newcommand{\heap}[2]{\textsf{H}_{#1}^{#2}}
\newcommand{\heapf}[1]{\textsf{H}{(#1)}}

\newcommand{\oldbound}{N_{\tra}'}
\newcommand{\newbound}{N_{\tra}}

\begin{document}

\maketitle

\begin{abstract}
    Transductions are binary relations of finite words. For
    rational transductions, i.e., transductions defined by finite
    transducers, the inclusion, equivalence and sequential
    uniformisation problems are known to be undecidable. In this
    paper, we investigate stronger variants of inclusion, equivalence
    and sequential uniformisation, based on a general notion of
    transducer resynchronisation, and show their decidability. We also
    investigate the classes of finite-valued rational transductions and
    deterministic rational transductions, which are  known to have a decidable
     equivalence problem. We show that sequential
    uniformisation is also decidable for them.
\end{abstract}

\section{Introduction}

Transductions generalise finite word languages to binary relations of
finite words. The notion of rationality for languages, and its
correspondence with finite automata, has been extended to transductions
and finite automata over pairs of words, called \emph{finite transducers}
\cite{berstel2009}. In this paper, we study decision problems for
finite transducers and prove new decidability results.

\vspace{1mm}
\noindent\textbf{Finite transducers} (Finite) transducers are
nondeterministic finite
automata whose transitions are labelled by pairs of words. The
(rational) transduction $\rel{\tra}$ defined by a transducer $\tra$ consists of all the pairs of
words $(u,v)$ obtained by concatenating the pairs occurring on
transitions of its successful computations. In this paper, we follow a
dynamic vision of transducers, as a machine that processes input words $u$
and produces output words $v$. Therefore, we may speak of the domain
of a transduction, as the language of input words that admit at least
one output word.

\vspace{1mm}\noindent\textbf{Equivalence problem} Unlike finite
automata, finite transducers have undecidable inclusion and equivalence problems
\cite{DBLP:journals/jacm/Griffiths68,FischerR68}, even restricted to unary
alphabets \cite{SICOMP::Ibarra1978}. The largest known classes with
decidable equivalence problem are those of finite-valued transducers
and deterministic transducers. A transducer is finite-valued if it
produces at most $k$ outputs per input, for a bound $k$
that only depends on the transducer. It is decidable whether a
transducer is $k$-valued for a given $k$ \cite{GurIba83}, and whether there
exists $k$ such that it is $k$-valued \cite{DBLP:journals/acta/Weber89}. Any finite-valued
transducer is known to be (effectively) equivalent to a finite union of unambiguous 
transducers \cite{DBLP:conf/mfcs/Weber88}, and thus to a finitely ambiguous transducer. 
Equivalence of $k$-ambiguous transducers was shown to be
decidable in \cite{GurIba83}, and equivalence of $k$-valued
transducers was first shown to be decidable in
\cite{CulKar86c,DBLP:conf/mfcs/Weber88}. Other algorithms with better
complexities appeared later, and the best known algorithm runs in
exponential time, for a fixed $k$ \cite{DBLP:conf/dlt/Souza08}.  

A transducer is deterministic if the transitions are deterministic in
the classical sense, and furthermore each state  processes either only
input symbols or only output symbols. The class of deterministic
rational transductions is also referred to as DRat, and it strictly
extends the class of synchronous rational transductions (also called
automatic relations), see e.g.\ \cite{CartonCG06} for an overview of
these sub-classes of rational transductions. As opposed to the
class of finite-valued transducers, it is undecidable whether a
transduction is equivalent to a deterministic transduction
\cite{FischerR68}. However, the equivalence problem for DRat is known
to be decidable \cite{Bird73}, even in polynomial time
\cite{FriedmanG82}. This makes this class an interesting candidate for
further investigations of decision problems.

\vspace{1mm}
\noindent\textbf{Uniformisation problem} Two classes of interest are
the rational and sequential functions, which are respectively defined
by $1$-valued transducers and sequential transducers. The latter read 
input words in a deterministic manner, and therefore produce a unique output word for each
input. There are rational functions that are not sequential, but it is
decidable in \textsc{PTime} whether a transducer defines a sequential
function \cite{DBLP:journals/iandc/WeberK95}. Since rational transductions do not define, in general,
functions, an interesting question is whether a unique output word can
be picked for each input word of a rational transduction $R$, in a
regular way, thus defining a function $f\subseteq R$ with the same
domain as $R$. Such a function $f$ is called a \emph{uniformiser} of $R$. 
It is known that any rational transduction admits a rational uniformiser
\cite{journals/iandc/Kobayashi69,Eilenberg:1974:ALM:540337} and, in
the case of DRat, even a \emph{lexicographic} uniformiser that picks the
smallest output words according to a lexicographic order, making the
uniformiser only depend on the transduction \cite{lncs194*300,PelSak99}.
In this paper, we are interested in sequential uniformisers. Even rational functions do not admit
sequential uniformisers in general, and therefore this gives rise to a
decision problem: Given a finite transducer, does it admit a
sequential uniformiser? It is worth noting that even if any rational
transduction $R$ can be uniformised by a rational uniformiser $U$, the
sequential uniformisability of $R$ does not imply, in general, that
\emph{any} of the uniformisers $U$ is equivalent to a sequential transducer. As a matter of fact, it
is known that the sequential uniformisation problem is undecidable for
rational transductions~\cite{CarayolL14}.

The sequential uniformisation problem echoes a similar problem introduced by
Church, the \emph{synthesis problem}, which currently receives a
lot of attention from the computer-aided verification community in the
context of open reactive systems (see
\cite{KupfermanPV06,FiliotJR11,BloemJPPS12} for some work on this
subject from the last decade). This problem asks whether given a logical
specification of a system, there exists an implementation that
satisfies it. In this context, reactive systems are non-terminating
systems that react to some unpredictable environment stimuli in a
\emph{synchronised} fashion: for each environment input, they produce
an output in a deterministic manner, such that the specification is met in the limit. 
Their executions are modelled by infinite words over a
product alphabet, and the interaction with the environment makes 
game theory a powerful tool in this context. A seminal result due to B\"uchi and Landweber shows
that the synthesis problem is decidable for MSO specifications
\cite{BuLa69} (see \cite{Thomas08} for a modern presentation and an
overview). 

Restricted to finite words, the sequential uniformisation problem generalises the synthesis
problem to an asynchronous setting: the transduction $R$ is the specification,
while the sequential uniformiser $f$ is the implementation.
%, and both are allowed to produce several outputs per input, or none. 

% Unlike finite automata, important problems are undecidable for finite
% transducers. The inclusion and equivalence problems asks whether given
% two transducers $\tra_1$ and $\tra_2$, the transductions they define
% are included and equal respectively. These problems are undecidable
% \cite{DBLP:journals/jacm/Griffiths68}. Large subclasses of rational
% transductions have then been introduced to recover decidability of
% inclusion and equivalence. The largest known class of 

\vspace{1mm}
\noindent\textbf{Resynchronisers} One of the main difficulty of
transducers is that two equivalent transducers may produce their
outputs very differently: One transducer may go fast and be ahead of
the other. By tagging symbols with two colours (for input and output),
transductions can be seen as languages, called \emph{synchronisation
languages}. It is known by
Nivat's theorem that rational transductions are synchronised by
regular languages \cite{Niv-transductions-lc}, and any transducer defines a
regular synchronisation language. Other correspondences between classes of
synchronisation languages and classes of rational transductions have
been established in \cite{conf/stacs/FigueiraL14}.
However in general, there is an infinite number of synchronisation
languages for a single transduction, making problems such as equivalence and sequential
uniformisation undecidable. To overcome this difficulty, Bojanczyk has
introduced \emph{transductions with origin information}, which amounts
to add the synchronisation information into the semantics of
transducers, via an origin function mapping output positions their
originating input positions \cite{DBLP:journals/corr/Bojanczyk13}. The main result of
\cite{DBLP:journals/corr/Bojanczyk13} is a
machine-independent characterisation of transductions (with origin
information) defined by two-way transducers. With respect to the
equivalence problem, considering the origin information makes the
problem easy: two transducers define the same transduction with
same origin mappings if they have the same synchronisation language. 
In this paper, we generalise this idea and propose decision problems modulo resynchronisation. A
\emph{resynchroniser} $\sync$ is a transduction, mapping a synchronisation
language to another one. Then, we consider related equivalence and
sequential uniformiser problems: for instance, given two transducers, are their
synchronisation languages equal modulo $\sync$? For the identity
resynchroniser, it is the same as origin-equivalence.

\vspace{1mm}
\noindent\textbf{Contributions} As a first contribution, we show that
inclusion, equivalence and sequential uniformisation are decidable
modulo \emph{rational} resynchronisers. For equivalence, it easily
reduces to an automata equivalence problem. For sequential
uniformisation, it boils down to solving a two-player safety game. 
We then consider a particular class of resynchronisers, the
\emph{$k$-delay resynchronisers}, that can apply a fixed delay $k$ 
to a synchronisation language, where the delay is a measure of how
ahead an output word is from another one \cite{BealCPS03}. The $k$-delay
resynchroniser is rational for each $k$, which implies the decidability of the
corresponding decision problem. Interestingly, we show that for the
class of real-time transducers (reading at least one input symbol in
each transition),
$k$-delay resynchronisers encompass all the power of rational
synchronisers with respect to the decision problems considered in this
paper. 

Our second main contribution is to show that equivalence and
sequential uniformisation modulo $k$-delay resynchronisers are
complete for finite-valued transducers. Given two finite-valued
transducers, if they are equivalent, then some $k$ can be computed 
such that they are $k$-delay equivalent. This yields another, delay-based, proof of
the decidability of finite-valued transducer equivalence. We show a
similar result for sequential uniformisation, by a pumping argument
based on an analysis of the idempotent elements in the transition
monoid of finitely-ambiguous transducers. This implies a new result: The decidability of
sequential uniformisation for finite-valued transducers.

Finally, our third main contribution is a decidability proof for the
sequential uniformisation problem for deterministic rational
transductions, extending a corresponding result for automatic
relations from  \cite{CarayolL14}.

\vspace{1mm}
\noindent\textbf{Structure of the paper} In Section~\ref{sec:prelims},
we introduce automata, transducers and decision problems for them. In
Section~\ref{sec:synch}, we define the notion of resynchronisers for
transductions and study their associated decision problems. We also
introduce the particular class of bounded delay resynchronisers. In
Section \ref{sec:fvalued}, we study the class of finite-valued
rational transductions and prove decidability of their sequential
uniformisation. Finally in Section \ref{sec:drat}, we prove
decidability of sequential uniformisation for deterministic rational
transductions. Due to lack of space, proofs are only sketched in the
paper. All full proofs can be found in the appendix section. 

%%% Local Variables:
%%% mode: latex
%%% TeX-master: t
%%% End:

\section{Automata and Transducers}\label{sec:prelims}

Let $\mathbb{N}$ denote the set of non-negative integers
$\{0,1,\dots\}$, and for every $n \in \mathbb{N}$, let $[n]$ denote the set $\{1,\ldots,n\}$.
Given a finite set $A$, let $|A|$ denote its cardinality.

 % of $A$ and let $\mathcal{P}(A)$ be set of subsets of $A$.
% Given two sets $A$ and $B$, let $\mathcal{F}(A,B)$ denote the set of
% functions from $A$ to $B$.

\vspace{1mm}
\noindent\textbf{Languages and Transductions of Words} An alphabet
$\Sigma$ is a finite set of symbols. The elements of the free monoid
$\Sigma^*$ are called \emph{words} over $\Sigma$. The length of a word $w$ is the number of its symbols. It is 
written $|w|$. The empty word (of length $0$) is denoted by $\epsilon$, and $\Sigma^+ = \Sigma^* \setminus\{\epsilon\}$ 
%For all $i\in\{1,\dots,|u|\}$, we denote by 
%$w[i]$ the $i$-th letter of $w$, and for all $j\in\{1,\dots,|u|\}$ such
%that $j\geq i$, by $w[i{:}j]$ the factor $w[i]w[i+1]\dots w[j]$. By
%$w[{:}i]$ (resp. $w[i{:}]$) we denote the prefix up to (and including)
%position $i$ (resp. the suffix from (and including) position $i$). We
%set $w[{:}0] =\epsilon$. 
%For $\Sigma'\subseteq \Sigma$, we define the morphism
%$\pi_{\Sigma'}$ from $\Sigma^*$ to $\Sigma'^*$ by
%$\pi_{\Sigma'}(\sigma) = \sigma$ if $\sigma\in \Sigma'$ and by
%$\pi_{\Sigma'}(\sigma) = \epsilon$ otherwise.
%For instance, if $\Sigma = \{a,b,c\}$ and $\Sigma' = \{a,b\}$, we have
%$\pi_{\Sigma'}(acbcb) = abb$ and $\pi_{\Sigma'}(ccc) = \epsilon$. 
The set $\Sigma^*$ can be partially ordered by the word prefix relation
$\preceq$.%, i.e., $u\preceq v$ if $u$ is a prefix of $v$. 

We denote by $\Sigma^{-1}$ the set of
symbols $\sigma^{-1}$ for all $\sigma\in\Sigma$. Any word $u\in
(\Sigma\cup \Sigma^{-1})^*$ can be reduced into a unique irreducible word $\overline{u}$ by the equations
$\sigma\sigma^{-1} = \sigma^{-1}\sigma = \epsilon$ for all
$\sigma\in\Sigma$. Let $G_\Sigma$ be the set of irreducible words over
$\Sigma\cup \Sigma^{-1}$. The set $G_\Sigma$ equipped with concatenation $u.v =
\overline{uv}$ is a group, called the free group over $\Sigma$. We denote by $u^{-1}$ the inverse
of $u$. E.g. $(a^{-1}bc)^{-1} = c^{-1}b^{-1}a$. For $u\in G_\Sigma$,
we denote by $|u|$ its number of symbols. E.g., $|a^{-1}b^{-1}| = 2$,
$|a^{-1}bc^{-1}| = 3$

 A \emph{language} $L$ over $\Sigma$ is a subset of
$\Sigma^*$. A \emph{transduction} $R$ over $\Sigma$ is a subset of 
$\Sigma^*\times \Sigma^*$. The \emph{domain} of $R$ is the set 
$\dom(R) = \{ u\ |\ \exists v\in\Sigma^*\cdot (u,v)\in R\}$. 
For a word $u\in\Sigma^*$, we denote by $R(u)$ the set $\{v\ |\
(u,v)\in R\}$, and extend this notation to languages $L$ by $R(L) =
\bigcup_{u\in L} R(u)$. When $R$ is a function, we simply write $R(u) = v$
instead of $R(u)=\{v\}$. Finally, we denote by $\id_{\Sigma^*}$ the
identity relation on $\Sigma^*$.

\vspace{1mm}
\noindent\textbf{Automata} A (finite state) \emph{automaton} over a finite alphabet $\Sigma$ is a tuple $\aut  = (Q,I,F,\Delta)$, where
$Q$ is the finite set of states,
$I \subseteq Q$ is the set of initial states,
$F \subseteq Q$ is the set of final states,
and $\Delta \subseteq Q \times \Sigma^* \times Q$ is the finite transition relation.
Given a transition $(q,w,q') \in \Delta$, $q$ is called its source, $q'$ its target, and $w$ its label.
An automaton is called \emph{deterministic} if each of its transition
is labelled by a single letter, and it admits no pair of transitions
that have same source, same label, and different targets. 

A \emph{run} of $A$ on a word $u\in\Sigma^*$ from state $q$ to
state $p$ is either a single state $q\in Q$ if $u=\epsilon$ and $q=p$, or a word $r = (q_1,u_1,p_1)(q_2,u_2,p_2)\dots (q_n,u_n,p_n)\in\Delta^+$
if $u\in\Sigma^+$, where $u = u_1\dots u_n$, $q_1 = q$ and $p_n = p$, and for all
$i\in\{1,\dots,n-1\}$, $p_i = q_{i+1}$. We write
$q_1\xrightarrow{u}_A p_n$ (or simply $q_1\xrightarrow{u} p_n$) if
such a run exists.  A run $r$ from a state $q$ to a state $p$ is \emph{accepting} if
$q$ is initial and $p$ is final. The \emph{language} recognised by $\aut$, written $\lang{\aut}$, is
the set of words $w \in \Sigma^*$ such that there exists an accepting
run of $A$ on $w$. If $B$ is an automaton, we write $A\subseteq B$
(resp. $A\equiv B$) whenever $\lang{A}\subseteq
\lang{B}$ (resp. $\lang{A} = \lang{B}$). 

\vspace{1mm}
\noindent\textbf{Transducers}
A (finite state) \emph{transducer} over a finite alphabet $\Sigma$ is a tuple $\tra  = (Q,I,F,\Delta,f)$, where 
$Q$ is the finite set of states,
$I \subseteq Q$ the set of initial states,
$F \subseteq Q$ the set of final states,
$\Delta \subseteq Q \times \Sigma^* \times \Sigma^* \times Q$ the transition relation,
and $f : F {\rightarrow} \Sigma^*$ the final output function.

As for automata, a \emph{run} of a transducer is either a single state
or a sequence of transitions. The \emph{input}  (resp. \emph{output}) of a run $r =
(q_1,u_1,v_1,p_1)\dots (q_n,u_n,v_n,p_n)\in\Delta^*$ is $\inr(r) = u_1\dots u_n$
(resp. $\outr(r)=v_1\dots v_n$). If $r$ is reduced to a single state, its input
and output are both $\epsilon$. We say that $r$ is a run of $\tra$ on
$u_1\dots u_n$. We write $q\xrightarrow{u|v} p$ to
mean that there exists a run on input $u\in\Sigma^*$ whose output is
$v\in\Sigma^*$. In particular, $q\xrightarrow{\epsilon|\epsilon} q$
for all $q\in Q$. The notion of accepting run of automata carries over
to transducers. The \emph{transduction} recognised by $\tra$, written $\rel{\tra}$ is the set
of pairs $(u,vf(p)) \in \Sigma^* \times \Sigma^*$ such that there
exists an accepting run of $\tra$ on $u$ from a state $q$ to a state
$p$ whose output is $v$.  We define $\dom(\tra)$ as
$\dom(\rel{\tra})$. The class of \emph{rational transductions} is the class of relations definable by
finite state transducers.

The \emph{input automaton} of $\tra$ is the automaton $\aut =
(Q,I,F,\Delta')$ over the alphabet $\Sigma$, where $\Delta' = \{
(q,u,q') | (q,u,v,q') \in \Delta \}$. A transducer is called \emph{real time} if each of its transition is labelled by a pair $(a,v)$, where $a \in \Sigma$ and $v \in \Sigma^*$.
A transducer is called \emph{sequential} if its input automaton is
deterministic\footnote{Our model of sequential transducers was originally called
  subsequential transducers in the literature. We follow the
  terminology of \cite{DBLP:journals/tcs/LombardyS06}, where it is discussed.}. Sequential transducers define \emph{sequential
  transductions}. 
A transducer is \emph{trim} if all its accessible
states are co-accessible, i.e. for all $q\in Q$, $q_0\in I$, $u,v\in\Sigma^*$,
if $q_0\xrightarrow{u\mid v}
q$, then there exist $u',v'\in\Sigma^*$ and $q_f\in F$ such that
$q\xrightarrow{u'\mid v'} q_f$. 

\vspace{2mm}
\noindent\textbf{Decision Problems for Transducers} Let $\tra_1,\tra_2$ be two transducers over an alphabet $\Sigma$. We
write $\tra_1\subseteq \tra_2$ whenever $\rel{\tra_1}\subseteq
\rel{\tra_2}$. The \emph{inclusion
problem} asks, given $\tra_1,\tra_2$, whether $\tra_1\subseteq
\tra_2$. Similarly, we define the equivalence problem by asking
whether $\rel{\tra_1} = \rel{\tra_2}$, denoted $\tra_1\equiv \tra_2$. 
Let $\tra$ be a transducer over an alphabet $\Sigma$.
A uniformiser of $\tra$ is a transducer $\uni$ such that
$\uni\subseteq \tra$ and $\dom(\uni) = \dom(\tra)$.
We sometimes write seq-uniformiser for sequential uniformiser. 
The \emph{sequential uniformisation problem} (seq-uniformisation
problem) asks, given a transducer $\tra$ over
$\Sigma$, whether $\tra$ admits a seq-uniformiser.

\begin{center}
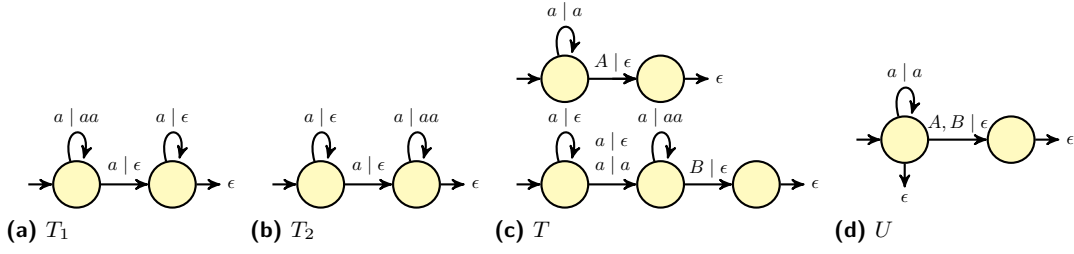
\begin{figure}[t]
\vspace{-5mm}
\begin{center}
\subfloat[$\tra_1$]{
\begin{tikzpicture}[->,>=stealth',auto,node distance=3cm,thick,scale=0.7,every node/.style={scale=0.7}]
  \tikzstyle{every state}=[fill=yellow!30,text=black]
  % \tikzstyle{every edge}=[draw=black,text=red]
  \tikzstyle{initial}=[initial by arrow, initial where=left, initial text=]
  \tikzstyle{accepting}=[accepting by arrow, accepting where=right, accepting text=$\epsilon$]

  \node[initial,state] (A) at(0,0) {};
  \node[accepting,state] at(1.8,0) (B)  {};
  \path (A) edge node {$\begin{array}{l}\trans{a}{\epsilon}\end{array}$} (B);
  \path (A) edge [loop above] node {\trans{a}{aa}} (A);
  \path (B) edge [loop above] node {\trans{a}{\epsilon}} (B);
\end{tikzpicture}
}
\subfloat[$\tra_2$]{
\begin{tikzpicture}[->,>=stealth',auto,node distance=3cm,thick,scale=0.7,every node/.style={scale=0.7}]
  \tikzstyle{every state}=[fill=yellow!30,text=black]
  % \tikzstyle{every edge}=[draw=black,text=red]
  \tikzstyle{initial}=[initial by arrow, initial where=left, initial text=]
  \tikzstyle{accepting}=[accepting by arrow, accepting where=right, accepting text=$\epsilon$]

  \node[initial,state] (A) at(0,0) {};
  \node[accepting,state] at(1.8,0) (B)  {};
  \path (A) edge node {$\begin{array}{l}\trans{a}{\epsilon}\end{array}$} (B);
  \path (A) edge [loop above] node {\trans{a}{\epsilon}} (A);
  \path (B) edge [loop above] node {\trans{a}{aa}} (B);
\end{tikzpicture}
}%
\subfloat[$\tra$]{
\begin{tikzpicture}[->,>=stealth',auto,node distance=3cm,thick,scale=0.7,every node/.style={scale=0.7}]
  \tikzstyle{every state}=[fill=yellow!30,text=black]
  % \tikzstyle{every edge}=[draw=black,text=red]
  \tikzstyle{initial}=[initial by arrow, initial where=left, initial text=]
  \tikzstyle{accepting}=[accepting by arrow, accepting where=right, accepting text=$\epsilon$]

  \node[initial,state] (A) at(0,1) {};
  \node[initial,accepting,state] (B) at(1.8,1) {};
  \node[initial,state] (A2) at(0,-1) {};
  \node[state] (B2) at(1.8,-1) {};
  \node[accepting,state] (C2) at(3.6,-1) {};
  \path (B2) edge node {\trans{B}{\epsilon}} (C2);
  \path (A) edge node {\trans{A}{\epsilon}} (B);
  \path (A2) edge node {$\begin{array}{l}\trans{a}{\epsilon}\\ \trans{a}{a}\end{array}$} (B2);
  \path (A) edge [loop above] node {\trans{a}{a}} (A);
  \path (A2) edge [loop above] node {\trans{a}{\epsilon}} (A2);
  \path (B2) edge [loop above] node {\trans{a}{aa}} (B2);
\end{tikzpicture}
}
\subfloat[$U$]{
\begin{tikzpicture}[->,>=stealth',auto,node distance=3cm,thick,scale=0.7,every node/.style={scale=0.7}]
  \tikzstyle{every state}=[fill=yellow!30,text=black]
  % \tikzstyle{every edge}=[draw=black,text=red]
  \tikzstyle{initial}=[initial by arrow, initial where=left, initial text=]
  \tikzstyle{accepting}=[accepting by arrow, accepting where=right, accepting text=$\epsilon$]
  \tikzstyle{acceptingdown}=[accepting by arrow, accepting where=below, accepting text=$\epsilon$]

  \node[initial,acceptingdown,state] (A) at(0,1) {};
  \node[accepting,state] at(2,1) (B)  {};
  \path (A) edge node {\trans{A,B}{\epsilon}} (B);
  \path (A) edge [loop above] node {\trans{a}{a}} (A);
\end{tikzpicture}
}
\caption{Transducers such that $\tra_1\equiv \tra_2$ and $\tra$ is
  seq-uniformisable by $U$.}
\label{fig:notk} 
\vspace{-4mm}
\end{center}
\end{figure}
\end{center}

\vspace{-5mm}
\begin{example}\label{ex:trans}
    Consider the transducers $\tra_1$ and $\tra_2$ of
    Fig.~\ref{fig:notk} over the alphabet $\{a\}$. They both define
    the transduction $\{ (a^n, a^{2i})\mid n\geq 1,0 \le i \le n{-}1\}$
    and are therefore equivalent. The transducer $\tra$ is over the
    alphabet $\{a,A,B\}$ and defines the transduction $\{ (a^nA, a^n)\mid n\geq 0\}\cup \{ (a^nB, a^i)\mid
    n\geq 1,\ 0 \le i \le 2n{-}1\}$. It is uniformisable by the
    sequential transducer $U$ with $\rel{U} = \{ (a^n\alpha,
    a^n)\mid n\geq 0,\alpha\in\{A,B\}\}$. 
\end{example}

\begin{theorem}[\cite{DBLP:journals/jacm/Griffiths68,CarayolL14}]\label{thm:undecunif}
\hspace{-1mm}The inclusion, equivalence and sequential uniformisation problems for rational transductions are undecidable.
\end{theorem}

% For subclasses of transducers, decidability has been recovered. The
% largest known class with decidable inclusion and equivalence is the
% class of \emph{finite-valued transducers}
% \cite{GurIba83,CulKar86c,DBLP:conf/mfcs/Weber88,DBLP:conf/dlt/Souza08}, i.e. transducers $\tra$
% such that there exists $k\in\mathbb{N}$ and for all
% $u\in\dom(\rel{\tra})$, $|\rel{\tra}(u)|\leq k$. It is also known
% that any rational relation is effectively uniformisable by a rational
% function \cite{journals/iandc/Kobayashi69}.

%%% Local Variables:
%%% mode: latex
%%% TeX-master: t
%%% End:

\section{Decision Problems Modulo Resynchronisers}\label{sec:synch}

A pair $(u,v)\in \Sigma^*\times \Sigma^*$ can be represented by a 
coloured word over $\Sigma\times \{\inp,\outp\}$,
where the colours indicate whether a symbol in $\Sigma$ is an input or an
output symbol. Such a coloured word is called a synchronisation of
$(u,v)$. 
More generally, any language over the alphabet $\Sigma\times \{\inp,\outp\}$ represents
s transduction $R\subseteq \Sigma^*\times \Sigma^*$, and is called a synchronisation language for $R$.
This way of representing transductions is analysed in \cite{conf/stacs/FigueiraL14}. 
What we call a resynchroniser below, is a transduction of synchronisations, that is, 
over words in $(\Sigma\times \{\inp,\outp\})^*$, that preserves
the represented pairs. In this section, we study stronger notion of
inclusion, equivalence and sequential uniformisation, parametrised by such a
resynchroniser. We show their decidability for rational resynchronisers
and introduce the class of bounded delay resynchronisers, and show that it
has appealing properties.

\vspace{1mm}
\noindent\textbf{Synchronisations and resynchronisers} 
Given an alphabet $\Sigma$, we let $\inouta = \Sigma\times
\{\inp,\outp\}$, $\ina = \Sigma\times
\{\inp\}$ and $\outa = \Sigma\times
\{\outp\}$. For $c\in \{\inp,\outp\}$, we write $\sigma^c$
instead of $(\sigma,c)$. The colouring $c$ can be seen as a morphism
$.^c : \Sigma^*\rightarrow \Sigma^*_c$ and we write $u^c$ its
application on a word $u\in\Sigma^*$. Conversely, for
$c\in\{\inp,\outp\}$, we define two
morphisms $\pi_c : \inouta^*\rightarrow \Sigma$ that extract the input
and output words, by $\pi_c(\sigma^c) = \sigma$ and $\pi_c(\sigma^d) =
\epsilon$, for all $\sigma\in\Sigma$, and $d\neq c$. %
% The \emph{input and output morphisms} of the words over $ (\inouta)^*$ are
% defined respectively as the monoid morphisms $\pi_\inp$ and
% $\pi_\outp$ generated by $\pi_\inp(\intag{\sigma}) = \sigma$,
% $\pi_\inp(\outtag{\sigma}) = \epsilon$, $\pi_\outp(\outtag{\sigma}) = \sigma$,
% $\pi_\outp(\outtag{\sigma}) = \epsilon$ for all $\sigma\in\Sigma$. 
Two words $u,v\in(\inouta)^*$ are said to be \emph{equivalent},
denoted by $u\sim_\inout v$, if 
$\pi_c(u)=\pi_c(v)$ for all $c\in\{\inp,\outp\}$.
For example, $a^\inp b^\inp a^\outp$ and $a^\inp a^\outp b^\inp$ are equivalent, and both are synchronisations of $(ab,a)$.  
Any language $L\subseteq \inouta^*$ defines a transduction
over $\Sigma$ defined by $\rel{L}\ =\ \{
(\pi_{\inp}(w),\pi_{\outp}(w))\in \Sigma^*\times \Sigma^*\mid w\in
L\}$, and $L$ is called a $\emph{synchronisation}$ of a transduction
$R\subseteq \Sigma^*\times \Sigma^*$ if $\rel{L} = R$. %
% Note that two different languages may define
% the same transduction. For instance, let $\Sigma = \{a,b\}$.
% Then $L_1 = \{ a^{\inp}b^{\outp} \}$ and $L_2 = \{
% b^{\outp}a^{\inp} \}$ are two different languages, but 
% $\rel{L_1} = \rel{L_2} = \{ (a,b)\}$. However if
% $L_1\subseteq L_2$, then $\rel{L_1}\subseteq \rel{L_2}$.
 % In this
% paper, we define several notions to compare languages $L_1,L_2$ such
% that $R(L_1) = R(L_2)$ or $R(L_1)\subseteq R(L_2)$. 
% \begin{proposition}
    % For all languages $L_1,L_2$ over $\inouta$, if $L_1\subseteq L_2$
    % then $\rel{L_1}\subseteq \rel{L_2}$.
% \end{proposition}
We also say that $L$
synchronises $R$. Note that two different languages may
synchronise the same transduction.

Mapping a synchronisation to
another one is done through the notion of resynchroniser. A \emph{resynchroniser} is a transduction $\sync\subseteq
\inouta^*\times \inouta^*$, such that $(i)$
$\id_{\inouta^*}\subseteq \sync$ and $(ii)$ for all
$(w,w')\in\sync$, it holds $w\sim_\inout w'$.
For instance, the identity relation $\id_{(\inouta)^*}$ is a
resynchroniser that we shall denote by $\mathbb{I}$, as well as the relation $\mathbb{U}_{\inouta}  = \{
(w,w')\in \inouta\mid w\sim_\inout w'\}$, called the \emph{universal resynchroniser} 
over $\inouta$. We write $\mathbb{U}$ instead of
$\mathbb{U}_{\inouta}$ when it is clear from the context.
Note that for any resynchroniser $\sync$, we have
$\id_{\inouta^*}\subseteq \sync \subseteq \mathbb{U}$. %
% Observe that condition $(1)$ implies that $\dom(\sync) =
 % (\inouta)^*$ and $L\subseteq \sync(L)$.
% Note that there is in general infinitely many synchronisations for a
% transduction.
The properties $(i)$ and $(ii)$ of resynchronisers are chosen such that
they preserve the represented transductions, as stated in the proposition below.
% \begin{proposition}
% Let $L$ be a synchronisation for a transduction $R$ and $\sync$ a
% resynchroniser. Then $\sync(L)$ synchronises $R$, i.e. $R = \rel{L} =
% \rel{\sync(L)}$.
% \end{proposition}
\begin{proposition}
For all $L\subseteq \inouta^*$ and all resynchronisers $\sync\subseteq
\inouta^*\times \inouta^*$, $\rel{L} = \rel{\sync(L)}$.
\end{proposition}
Classes of synchronisation languages and their correspondence with the classes of rational
relations they synchronise have been studied in \cite{conf/stacs/FigueiraL14}. 
We can formulate in this framework a result known as Nivat's
theorem \cite{Niv-transductions-lc} as follows.
\begin{theorem}{\cite{Niv-transductions-lc}}\label{thm:nivat}
    A transduction $R$ is rational iff it is synchronised by a regular language.
\end{theorem}
 
Any transducer $\tra  = (Q,I,F,\Delta,f)$ naturally defines a regular
synchronisation for $\rel{\tra}$ by its \emph{underlying automaton}, which is
the automaton obtained by concatenating the pairs of input and output words 
on the transitions and marking them with the respective symbol from $\{\inp,\outp\}$.
Formally, it is the 
automaton $A = (Q \cup \{ q_\dashv\},I, \{ q_\dashv\}, \Delta')$ over
$\inouta$, where $\Delta' = \{(q,v^{\inp}w^{\outp},q') | (q,v,w,q')
\in \Delta \} \cup \{(q,f(q)^{\outp},q_\dashv)|q \in F \}$. The \emph{language recognised by $\tra$} is the language 
recognised by its underlying automaton, denoted by $\mathcal{L}_{\tra}$, i.e. $\lang{\tra} =
\lang{A}$. Obviously, $\lang{\tra}$ is a synchronisation for the
relation $\rel{\tra}$, i.e. $\rel{\lang{\tra}} = \rel{\tra}$.

%%% Local Variables:
%%% mode: latex
%%% TeX-master: t
%%% End:

\vspace{2mm}
\noindent\textbf{Decision problems for transducers modulo
  resynchronisers} Let $\Sigma$ be an alphabet, $\sync$ be a resynchroniser over
$(\inouta)^*$, and $\tra_1,\tra_2$ be two transducers over $\Sigma$.  
We say that $\tra_1$ is \emph{included in $\tra_2$ modulo $\sync$} (or
\emph{$\sync$-included}), denoted by $\tra_1\subseteq_\sync \tra_2$, if 
$\lang{\tra_1} \subseteq \sync(\lang{\tra_2})$. We say that $\tra_1$ is \emph{equivalent to
  $\tra_2$ modulo $\sync$} (or \emph{$\sync$-equivalent}), denoted by
$\tra_1 \equiv_\sync \tra_2$, if $\tra_1 \subseteq_\sync \tra_2$ and $\tra_2\subseteq_\sync
\tra_1$. For a fixed synchroniser $\sync$, the \emph{$\sync$-inclusion
(resp. $\sync$-equivalence) problem} asks, given two transducers $\tra_1, \tra_2$ over $\Sigma$, whether 
$\tra_1\subseteq_\sync \tra_2$ (resp. $\tra_1 \equiv_\sync
\tra_2$). We say that $\tra_1$ is \emph{sequentially $\sync$-uniformisable} if it
admits a sequential uniformiser $U$ such that $U\subseteq_\sync \tra_1$, and in that case $U$ is called a
sequential $\sync$-uniformiser of $\tra_1$ (seq-$\sync$-uniformiser
for short). The \emph{sequential $\sync$-uniformisation
  problem} asks whether a given transducer is
seq-$\sync$-uniformisable.

It should be clear from the definition that $\mathbb{I}$-inclusion
implies $\sync$-inclusion for any resynchroniser $\sync$, which in
turn implies $\mathbb{U}$-inclusion. As a matter of fact, it is easy
to see that $\mathbb{U}$-inclusion is equivalent to classical
inclusion. The same remarks can be made for equivalence and
sequential uniformisation, and therefore, as a consequence of Theorem~\ref{thm:undecunif}, we get:

\begin{theorem}
    The $\mathbb{U}$-inclusion, $\mathbb{U}$-equivalence, sequential
    $\mathbb{U}$-uniformisation problems for rational transductions are undecidable. 
\end{theorem}

%%% transducer Variables:
%%% mode: latex
%%% TeX-master: "decision.tedx"
%%% End:

\vspace{.5mm}
\noindent\textbf{Decision problems for transducers modulo rational
  resynchronisers} The $\mathbb{U}$-decision problems
are undecidable, this raises the question whether there is an interesting
class of resynchronisers for which we can recover decidability. 
It turns out that $\mathbb{U}$ is not rational. In
contrast, we show that, as long as $\sync$ is rational,
the $\sync$-decision problems are reducible to the
$\mathbb{I}$-decision problems, which in turn can be solved by reduction to
decidable problems of automata and two-player games.

% % We shall begin by exposing the regular language inclusion problem, and the reduction from the $\mathbb{I}$-inclusion problem into it.

% % Given two automata $\aut_1$ and $\aut_2$, the inclusion problem asks whether $\lang{\aut_1}$ is included into $\lang{\aut_2}$.

% \begin{theorem}[reference needed]
% The regular language inclusion problem is decidable.
% \end{theorem}

\begin{restatable}{proposition}{decisionIproblems}\label{prop:Iinclusion}\label{prop:Iuniformisation}
The $\mathbb{I}$-inclusion and $\mathbb{I}$-equivalence problems are
\textsc{PSpace-complete}. The sequential $\mathbb{I}$-uniformisation problem is 
\textsc{ExpTime-complete}. 
\end{restatable}

\begin{proof}
First, note that $\tra_1\subseteq_{\mathbb{I}}\tra_2$ iff
$\lang{\tra_1}\subseteq \lang{\tra_2}$ iff $A_1\subseteq A_2$, where
$A_1,A_2$ are the underlying automata of $\tra_1,\tra_2$
respectively. Automata inclusion and equivalence problems are
\textsc{PSpace-complete}, and they easily reduce (by putting $\epsilon$
outputs) to $\mathbb{I}$-inclusion
and $\mathbb{I}$-equivalence.

To get \textsc{ExpTime} membership of seq-$\mathbb{I}$-uniformisation, for
a transducer $\tra$, we construct a two-player safety game $G_\tra =
(V = V_{\textsf{In}}\uplus V_{\textsf{Out}}, v_0, E)$ between an
adversary (Player \textsf{In}) who picks input symbols and controls
positions in $V_{\textsf{In}}$, and a protagonist
(Player \textsf{Out}) who picks sequences of output symbols and
controls positions in $V_{\textsf{Out}}$. Wlog we assume that $\tra$
has no final output function, by adding an endmarker $\dashv$ to words of
its domain. Let $\aut = (Q,q_0,F,\delta)$ be a complete DFA equivalent to the the underlying
automaton of $\tra$ (whose size is at most
exponential  in the size of $\tra$). Player positions have three
components: a residual language\footnote{A residual of a language $L$
  over some alphabet $\Sigma$ is a language $u^{-1}L = \{ v\mid uv\in L\}$ for $u\in \Sigma^*$.} of $\dom(\tra)$ that controls 
the possible continuations of the input word chosen so
far by  Player \textsf{In}, a state of $\aut$ and a round $r \in
\{\textsf{In},\textsf{Out}\}$. Let $\mathcal{D} = \{ u^{-1}\dom(\tra)\mid u\in\Sigma^*\}$
be the set of residuals of $\dom(\tra)$ (for example, represented by the states of 
the minimal DFA for $\dom(\tra)$, which can be computed in exponential time in the size of $\tra$). 
Then, $V_{\textsf{In}} = \mathcal{D}\times Q\times
\{\textsf{In}\}$ and $V_{\textsf{Out}} = \mathcal{D}\times Q\times
\{\textsf{Out}\}$. The initial position is $v_0  = (\dom(\tra), q_0,
\textsf{In})$ and the edge relation $E$ as follows: from a position 
$(D,q,\textsf{In})$, there are outgoing edges to all states
$(\sigma^{-1}D,\delta(q,\sigma^\inp), \textsf{Out})$ for all
$\sigma\in\Sigma$. From a position $(D,q,\textsf{Out})$, Player
$\textsf{Out}$ can pick any state $q'\in Q$ such that there exists a
sequence $v\in\outa^*$ such that $q\xrightarrow{v}_{\aut} q'$, and in
that case an outgoing edge to $(D, q', \textsf{In})$ is added to $E$. 
The \emph{unsafe positions} for Player \textsf{Out} are all positions
$(D,q,\textsf{In})$ such that $\epsilon\in D$ and $q\not\in F$: At
such positions, Player $\textsf{In}$ could choose to terminate the sequence
of input symbols (while staying in $\dom(T)$ since $\epsilon\in D$)
and the sequence of output symbols chosen by Player \textsf{Out},
mixed with the input symbols chosen by Player \textsf{In}, does not
belong to $L(\aut)$ (because $q\not\in F$). It can be shown that
Player \textsf{Out} has a strategy to avoid the unsafe positions in
$G_\tra$ iff there exists a seq-$\mathbb{I}$-uniformiser of $\tra$
(finiteness of the seq-$\mathbb{I}$-uniformiser comes from the fact that
safety games are memoryless determined). Since safety games can be solved in polynomial time and $G_\tra$ has exponential
size, we get the result. The results on safety games that we use here can be found, e.g., in \cite{GradelTW02}.

For the \textsc{ExpTime} lower bound, we note that in our formalism we
can model the synchronous uniformisation (or synthesis) problem, as
considered in \cite{PnuRos:89} for infinite words, by taking
synchronisations that strictly alternate between input and output. It
seems to be common knowledge in the synthesis community that the
synchronous uniformisation problem is \textsc{ExpTime-complete} if the
relation is given by a nondeterministic automaton. However, we were
not able to find a reference for this result. We thus give a reduction
from the acceptance problem for alternating \textsc{Pspace} Turing
machines in the appendix.
\end{proof}

For all transducers $\tra$ and synchronisers $\sync$,
$\sync(\lang{\tra})$ is a regular synchronisation language and by
Nivat's theorem (Theorem~\ref{thm:nivat}), there exists a transducer 
$\tra^\sync$ such that 
$\lang{\tra^\sync} = \sync(\lang{\tra})$. It implies that the
seq-$\sync$-uniformisation of $\tra$ reduces to the
seq-$\mathbb{I}$-uniformisation of $\tra^\sync$. Similar arguments apply
for inclusion and equivalence and, as a consequence of
Proposition~\ref{prop:Iuniformisation}, we obtain:

\begin{restatable}{theorem}{decisionSproblems}\label{thm:rat}
Let $\sync$ be a rational resynchroniser, given as a transducer.
The $\sync$-inclusion and $\sync$-equivalence problems are
\textsc{PSpace-complete}. The
sequential $\sync$-uniformisation problem is
\textsc{ExpTime-complete}. 
\end{restatable}

% We write $L_1\equiv L_2$ if $R(L_1) = R(L_2)$, and
% $L_1\subseteqq L_2$ if $R(L_1)\subseteq R(L_2)$. Given
% two binary relations $\triangle_1,\triangle_2$ over languages
% such that $\triangle_1 \subseteq\triangle_2$, we say that 
% $\triangle_1$ is finer than $\triangle_2$ and that
% $\triangle_2$ is coarser than $\triangle_1$. In this paper,
% we study language relations that are coarser than language
% equality (resp. inclusion) and finer than the transduction
% equivalence $\equiv$ (resp. the transduction inclusion
% $\subseteqq$). 

%%% Local Variables:
%%% mode: latex
%%% TeX-master: t
%%% End:

%%% Local Variables:
%%% mode: latex
%%% TeX-master: t
%%% End:

\vspace{.5mm}
\noindent\textbf{Bounded delay resynchronisers}\label{sec:delay} 
The notion of \emph{delay} between outputs of transducers is a
powerful way of comparing transducers, which has been used, for
instance, to characterise sequential functions
\cite{berstel2009}. Intuitively, the delay between two runs on the
same input is a parameter that measures how a run is ahead of the
other, and the lag is the maximal delay over prefixes of the two
runs. We adapt the notion of delay and lag to coloured words and
define delay resynchronisers as resynchronisers that apply a fixed
delay to words in $\inouta^*$ (our notion of lag is not related to the one from \cite{conf/stacs/FigueiraL14}). 
Our results show that delay resynchronisers form
a fundamental class of resynchronisers.

% Since bounded delay resynchronisers are rational, their associated
% decision problems are decidable. We give complexity bounds for
% uniformisation modulo bounded delay. We also prove that inclusion,
% equivalence and uniformisation modulo bounded delay are not complete,
% in the sense that two transducers may be equivalent but not bounded
% delay equivalent. 

% Intuitively, for any two equivalent words $u,v\in(\inouta)^*$, the
% delay between $u$ and $v$ is a parameter that measure, along 

% A particular and powerful way of comparing how two transducers produce their
% outputs on the same input is via the notion of \emph{delay} between
% outputs. In this section, we study stronger versions of equivalence,
% inclusion and uniformisation problem modulo bounded
% delay. Intuitively, for instance for the equivalence problem, we
% require that the two transducers that are compared produce their
% output in a similar manner, where ``similar'' is formalised here via
% a notion of bounded delay. 

% We first define formally the notion of delay, and show that the
% new decision problems are equivalent to decision problems modulo
% rational synchronisers. In other words, applying a bounded delay to 
% a sequence in $(\inouta)^*$ can be done by a rational synchroniser, and therefore decidability of the
% bounded delay decision problems follows from the
% results of Section \ref{sec:}. However, it turns out that  bounded
% delay synchronisers is an important class of rational synchronisers: for large
% classes of transducers, they provide a complete procedures to test of
% equivalence and inclusion. 

The \emph{delay} between two words $u$ and $v$ over an alphabet
$\Sigma$ is the element
from the free group $G_\Sigma$ defined by $\del(u,v) = u^{-1} v$. 
E.g., $\del(ab,acd) = b^{-1}cd$. Note that $\del(u,v)\in \Sigma^*$ iff
$u\preceq v$, and $\del(u,v)\in (\Sigma^{-1})^*$ iff $v\preceq u$. 
The \emph{lag mapping} $\lag : (\inouta)^* \times (\inouta)^* \rightarrow
\mathbb{N} \cup \{+\infty\}$ gives the maximal length of the delay between the output part of
two words in $(\inouta)^*$ that have the same input. It is the metric defined by
$\lag(u,v)=+\infty$ if $\pi_\inp(u)\neq \pi_\inp(v)$. If $\pi_\inp(u)
= \pi_\inp(v)$, then $u$ and $v$ can be decomposed into 
$u = u_0a_1u_1 \ldots u_{n-1}a_nu_n$ and $v = v_0a_1v_1 \ldots
a_{n-1}b_nv_n$, such that $a_1, \ldots, a_n \in \ina$, $u_0,v_0,
\ldots, u_n,v_n \in (\outa)^*$. Then $\lag(u,v) = \max_{0 \leq i \leq
  n} |\del(u_0 \ldots u_i,v_0 \ldots v_i)|$. %
% Note that if $\lag(u,v)<+\infty$, then $u$ and $v$ are equivalent, and
% therefore for all $0\leq i\leq n$, $\del(u_0 \ldots u_i,v_0 \ldots v_i)\in\Sigma^*\cup
% (\Sigma^{-1})^*$. 
As an example, for $n\geq 1$, take $u_n = a^\inp a^\outp (a^\inp)^n$ and $v_n =
(a^\inp)^n a^\inp a^\outp$. Then for all $n\geq 1$, $\lag(u_n,v_n){=}
1$. Note that the occurrence of $a^\outp$ in $u_n$ is arbitrary
far from that of $a^\outp$ in $v_n$. 

We now define the $k$-delay resynchroniser
$\mathbb{D}_k$. Intuitively, it can shift output symbols of a word $u$
to the left or to the right, as long as the lag between $u$ and the
new word obtained this way is bounded by $k$. Formally, the
\emph{$k$-delay resynchroniser} is defined by $\dsync_k\ =\ \{
(u,v)\in (\inouta)^2\mid u\sim_{\inout}v\wedge \lag(u,v)\leq
k\}$. We define the $k$-inclusion, $k$-equivalence and
sequential $k$-uniformisation problems as the corresponding $\dsync_k$-decision
problems, and write $\subseteq_k$ and $\equiv_k$ instead of 
$\subseteq_{\dsync_k}$ and $\equiv_{\dsync_k}$
respectively. We also say that a transduction is seq-$k$-uniformisable
if it is seq-$\dsync_k$-uniformisable. An important property of $\dsync_k$ is:
\begin{restatable}{proposition}{dsyncrational}\label{prop:dsynckrat}
    For all $k\geq 0$, $\dsync_k$ is rational. 
\end{restatable}
As a direct consequence of the latter proposition and
Theorem~\ref{thm:rat}, the $k$-delay decision problems are all
decidable. We can be more precise:
\begin{restatable}{theorem}{decisionkdelay}
    For all $k\geq 0$, the $k$-inclusion, $k$-equivalence
    and sequential $k$-uniformisation problems are decidable and
    \textsc{ExpSpace-hard} if $k$ is part of the input. If $k$ is
    fixed, then the $k$-inclusion and $k$-equivalence problems are
    \textsc{PSpace-complete}, and the sequential $k$-uniformisation problem is
    \textsc{ExpTime-complete}.
\end{restatable}

% \paragraph{Bounded delay problems for transducers} 
% Given $k \in \mathbb{N}$ and two transducers $\tra_1$ and $\tra_2$, we
% say that $\tra_1$ is $k$-delay included in $\tra_2$, denoted by
% $\tra_1 \subseteq_k \tra_2$, if for every word $w \in \lang{\tra_1}$
% there exists a word $w' \in \lang{\tra_2}$ such that $w$ and $w'$ are
% equivalent, and $\lag(w,w') \leq
% k$. Similarly,  $\tra_1$ and $\tra_2$ are $k$-delay equivalent, denoted by
% $\tra_1 \equiv_k \tra_2$, if $\tra_1 \subseteq_k \tra_2$ and $\tra_2
% \subseteq_k \tra_1$. Finally, we say that $\tra_1$ is $k$-delay uniformisable if it admits a uniformiser $\uni$ such that $\uni \subseteq_k \tra_1$.
% With $k$ being part of the input, we can naturally define the
% $k$-delay inclusion, equivalence and uniformisation problems. 

% We first show that $k$-delay inclusion, equivalence and uniformisation
% are equivalent to inclusion, equivalence and uniformisation modulo
% rational synchronisers. 

% \begin{proposition}\label{prop:equivsynchro}
% Let $k \in \mathbb{N}$, and let $\tra_1,\tra_2,\tra$ be transducers. There exists a rational synchroniser $\dsync_k$
% such that 
% \begin{itemize}
    % \item $\tra_1$ and $\tra_2$ are $k$-delay equivalent iff $\tra_1$
      % and $\tra_2$ are equivalent modulo $\dsync_k$,
    % \item $\tra_1$ is $k$-delay included in $\tra_2$ iff $\tra_1$ is
      % included in $\tra_2$ modulo $\dsync_k$,
    % \item $\tra$ is $k$-delay uniformisable iff $\tra$ is
      % uniformisable modulo $\dsync_k$. 
% \end{itemize}
% Moreover, $\dsync_k$ can be defined by a transducer with $O(|\Sigma|^k)$
% states. 
% \end{proposition}

Even if inclusion is undecidable while $k$-inclusion is decidable, it
could be the case that inclusion reduces to $k$-inclusion, for
some $k$ that cannot be computed. We show that that it is not the
case, by using the transducers of Fig.~\ref{fig:notk}.

% : there are transducers $\tra_1,\tra_2$ such that $\tra_1\subseteq
% \tra_2$ but $\tra_1\not\subseteq_k\tra_2$ for any $k$. We prove
% similar results for equivalence and uniformisation. In the following
% proposition, $\subseteq_k$ and $\equiv_k$ are seen as binary relations
% % on transducers over a fixed unary alphabet $\Sigma$, and we denote by $\text{Unif}_k$ the class of transducers that are $k$-delay
% uniformisable, and by $\text{Unif}$ the class of transducers that are
% uniformisable.

\begin{restatable}{proposition}{propapprox}\label{prop:notcomplete}
    There exist transducers $\tra_1,\tra_2,\tra$ such that
    $\tra_1\equiv \tra_2$, $\tra_1\subseteq \tra_2$ and $\tra$ is
    seq-uniformisable, but for all $k\geq 0$,
    $\tra_1\not\equiv_k\tra_2$, $\tra_1\not\subseteq_k\tra_2$, and $\tra$ is not seq-$k$-uniformisable.
    % $(\bigcup_{k=0}^{+\infty} \subseteq_k)\ \subsetneq\ (\subseteq)$, 
    % $(\bigcup_{k=0}^{+\infty} \equiv_k)\ \subsetneq\ (\equiv)$,
    % $(\bigcup_{k=0}^{+\infty}\text{Unif}_k)\ \subsetneq\ \text{Unif}$.     
\end{restatable}

\begin{proof}
    Consider $\tra_1,\tra_2$ of Fig.\ref{fig:notk} and pairs of the
    form $(a^{n+1},a^{2n})\in\rel{\tra_1}=\rel{\tra_2}$. They both
    accept these pairs but $\tra_2$ will be arbitrarily late compared
    to $\tra_1$.  Consider now the transducer $T$ and its sequential uniformiser
    $U$. On inputs $a^nB$, $U$ will be arbitrarily ahead of $T$, and
    one can show that is the case for \emph{any} seq-uniformiser of $T$.
\end{proof}

Finally, we show that for real-time transductions, $k$-delay resynchronisers subsume any rational
resynchroniser $\sync$, in the sense that $\sync$-inclusion implies 
$k$-inclusion, for some $k$ that depends on $\sync$. Similar
results hold for equivalence and sequential uniformisation. The idea is that a
rational synchroniser cannot advance or delay the production of
outputs arbitrarily far away with a finite set of states. 

\begin{restatable}{theorem}{thmsyncdel}\label{thm:thmsyncdel}
    Let $\sync$ be a rational synchroniser (given by a transducer). 
    Let $\tra_1,\tra_2,\tra$ be real-time transducers. 
    There exists a computable integer $k\in\mathbb{N}$ such that:
   $(i)$ if $\tra_1\subseteq_\sync \tra_2$, then
        $\tra_1\subseteq_k \tra_2$, $(ii)$  if $\tra_1\equiv_\sync \tra_2$, then
        $\tra_1\equiv_k \tra_2$, and $(iii)$ if $\tra$ is seq-$\sync$-uniformisable, then $\tra$ is
        seq-$k$-uniformisable.
\end{restatable}

\noindent One cannot drop the real-time assumption in the latter
theorem. Indeed consider the following transducers
$\tra_1,\tra_2,\sync$, for which $\tra_1{\equiv_\sync} \tra_2$ but
$\tra_1, \tra_2$ are not $k$-equivalent for any $k$:
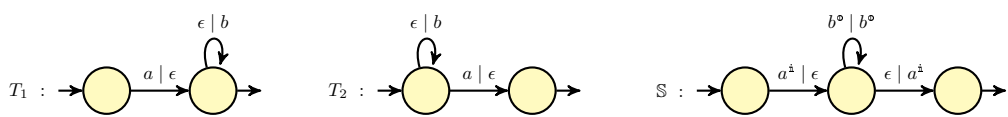
\vspace{-1mm}
\begin{center}
\begin{tikzpicture}[->,>=stealth',auto,node
    distance=1cm,thick,scale=0.7,every node/.style={scale=0.7}]
  \tikzstyle{every state}=[fill=yellow!30,text=black]
  % \tikzstyle{every edge}=[draw=black,text=red]
  \tikzstyle{initial}=[initial by arrow, initial where=left, initial text=]
  \tikzstyle{accepting}=[accepting by arrow, accepting where=right, accepting text=]
  \tikzstyle{acceptingdown}=[accepting by arrow, accepting where=below, accepting text=]

  \tikzstyle{initialT1}=[initial by arrow, initial where=left, initial
  text={$\tra_1\ :\ $}]
  \tikzstyle{initialT2}=[initial by arrow, initial where=left, initial
  text=$\tra_2\ :\ $]
  \tikzstyle{initialS}=[initial by arrow, initial where=left, initial
  text=$\sync\ :\ $]

  \node[initialT1,state] (A) at(0,0) {};
  \node[accepting,state] at(2,0) (B)  {};
  \path (A) edge node {\trans{a}{\epsilon}} (B);
  \path (B) edge [loop above] node {\trans{\epsilon}{b}} (B);

  \node[initialT2,state] (C) at(6,0) {};
  \node[accepting,state] at(8,0) (D)  {};
  \path (C) edge node {\trans{a}{\epsilon}} (D);
  \path (C) edge [loop above] node {\trans{\epsilon}{b}} (C);

  \node[initialS,state] (E) at(12,0) {};
  \node[state] at(14,0) (F)  {};
  \node[accepting,state] at(16,0) (G)  {};
  \path (E) edge node {\trans{a^\inp}{\epsilon}} (F);
  \path (F) edge [loop above] node {\trans{b^\outp}{b^\outp}} (F);
  \path (F) edge node {\trans{\epsilon}{a^\inp}} (G);
\end{tikzpicture}
\end{center}
\vspace{-1mm}

%%% Local Variables:
%%% mode: latex
%%% TeX-master: t
%%% End:

\section{Finite-valued transducers}\label{sec:fvalued}

Let $m\in\mathbb{N}$. We remind the reader that a transducer $\tra$ is called
\textit{$m$-valued} if each input has at most $m$ outputs, i.e. for
all $w\in\dom(\tra)$, $|\rel{\tra}(w)|\leq m$. It is finite-valued if
it is $m$-valued for some $m$. Finite-valuedness is decidable
\cite{DBLP:journals/acta/Weber89}. We prove that for the class of
finite-valued transducers, $k$-inclusion and sequential $k$-uniformisation
are complete. This yields, for finite-valued transducers, an alternative proof of the decidability of
the inclusion problem, and a new result: The decidability of
sequential uniformisation.

Let $m$ be a natural number. An automaton $A$ (resp. transducer
$\tra$) is called $m$-ambiguous if it is real-time\footnote{For
  simplicity reasons, we put real-timeness in the definition, but it
  is known to be wlog.}, and for any word $w\in
\lang{A}$ (resp. $w\in \dom(\tra)$), there exist at most $m$
accepting runs of $A$ (resp. $\tra$) on $w$.
An automaton (transducer) is called finitely ambiguous if there exists
$m \in \mathbb{N}$ such that it is $m$-ambiguous, and unambiguous if
it is $1$-ambiguous. Our proofs uses the following known decomposition
initially due to Weber:
\begin{theorem}\label{trans_fv}
\cite{DBLP:conf/mfcs/Weber88,journals/mst/SakarovitchS10} Any finite-valued transducer $\tra$ is
(effectively) equivalent to a union of unambiguous transducers. 
\end{theorem}

%% caution: it would not be correct to say that T is n-ambiguous if
%% its input automaton is. It would be correct if we impose that for
%% any two transitions (q,a,w1,p), (q,a,w2,p') implies p \neq p'

We first prove that, for the class of finitely ambiguous transducers,
inclusion and equivalence reduces to $k$-inclusion and $k$-equivalence
for some computable $k$. We state this result for inclusion, which immediately
implies it for equivalence.
% Proposition \ref{prop:notcomplete} does not hold for inclusion and equivalence.

\begin{restatable}{theorem}{thmfinitelyambiguous}\label{lem:finitelyambiguous}
Let $\tra_1$ and $\tra_2$ be two real-time transducers such that $\tra_2$ is $m$-ambiguous. 
Then there exists a computable integer $k$ such that $\tra_1\subseteq
\tra_2 \implies \tra_1\subseteq_k \tra_2.$ Moreover, $k$ can be chosen to be exponential in the size of $\tra_2$
and linear in the size of $\tra_1$.
\end{restatable}

\begin{proof}[Sketch of Proof]

We show the result by choosing $k$ high enough with respect to the size of $\tra_1$ and $\tra_2$, assuming that $\tra_1$ is not $k$-included
    in $\tra_2$, and exhibiting a contradiction.
If $\tra_1$ is not $k$-included
    in $\tra_2$, there exists a word $w\in \lang{\tra_1}$ such that for all words
    $w'\in \lang{\tra_2}$ equivalent to $w$, $\lag(w,w')>k$.
    Such a word is called a \textit{witness}.
    Given a witness $v$, let $N_{\tra_2}(v)$ denote the number of words $v'\in\lang{\tra_2}$
    equivalent to $v$.
We prove that for each witness $v$, if $N_{\tra_2}(v) > 0$, there exists a witness
      $t\in\lang{\tra_1}$ such that $N_{\tra_2}(t)<N_{\tra_2}(v)$,
      which by induction yields a contradiction. To do so, we follow three steps.

$(i)$ We exhibit synchronised loops between the run $\rho$ of $\tra_1$ recognising $v$, and the set $R_2$ containing all the runs of $\tra_2$ recognising words equivalent to $v$.

$(ii)$ Given a run $\rho' \in R_2$, as the delay between $\rho$ and $\rho'$ is greater than $k$ at some point, we can chose one of those loops such that the delay between $\rho$ and $\rho'$ grows along the loop.

$(iii)$ By iterating the loop a sufficient number of times, we
generate a new witness $t$ such that $N_{T_2}(t)<N_{T_2}(v)$.

Then, as $\tra_2$ is $m$-ambiguous, if there exists a witness $w$, $N_{\tra_2}(w) \leq m$, and, by applying the preceding remark inductively, we expose a witness $w_0$ such that $N_{\tra_2}(w_0) = 0$.
In other words,
$(\pi_{\inp}(w_0),\pi_{\outp}(w_0))\in\rel{\tra_1}\cap
\overline{\rel{\tra_2}}$, which contradicts the fact that $\rel{\tra_1}\subseteq
\rel{\tra_2}$.
\end{proof}

Since $k$-inclusion is decidable by Theorem \ref{thm:rat},
Theorem \ref{lem:finitelyambiguous} implies that the inclusion and
equivalence problems are decidable for finitely ambiguous
transducers. From the decomposition of Theorem \ref{trans_fv}, we
obtain an alternative proof of the decidability of equivalence of
finite-valued transducers, which was proved in \cite{CulKar86c,DBLP:conf/mfcs/Weber88}. 

\begin{corollary}[\cite{CulKar86c,DBLP:conf/mfcs/Weber88} Alternative proof]
The inclusion and equivalence problems for finite-valued transducers are decidable.
\end{corollary}

We now prove the two corresponding results for the sequential uniformisation problem.

\begin{restatable}{theorem}{thmunifvalued}\label{f_a_det}
Let $\tra$ be a real-time trim transducer given as a finite union of unambiguous transducers.
Then there exists a computable integer $N_T$ such that if $\tra$ is
sequentially uniformisable, then it is sequentially $N_T$-uniformisable.
\end{restatable}

\begin{proof}[Sketch of proof]
If $\tra$ is seq-uniformisable, then there exists a sequential uniformiser $\uni$
    of $\tra$ such that
    $\dom(U) = \dom(T)$ and $U\subseteq T$. The latter inclusion implies, by Theorem
    \ref{lem:finitelyambiguous}, that there exists an integer $k$ such
    that $\uni\subseteq_k T$, and so $\uni$ is a seq-$k$-uniformiser of $\tra$.
    However, $k$ depends on the number of states of the hypothetical
    uniformiser $\uni$. We show how to construct, by simulating the
    behaviour of $U$, another  seq-$N_T$-uniformiser $U'$, where $N_T$ only depends on
    $T$ and can be computed.

    More precisely, we define a function $\rho : \Sigma^*\rightarrow
    \Sigma^*$ and define $U'$ such that on any input $w$, it simulates $U$ on 
    input $\rho(w)$. The function $\rho$ iterates some well-chosen
    subwords of $w$  to blow up the delay between the outputs of the
    runs of $T$ on $\rho(w)$. On input $\rho(w)$, any seq-$k$-uniformiser of
    $T$, and $U$ in particular, is forced to make choices between
    possible outputs of $T$ on $\rho(w)$, in
    order to decrease the delay. % 
 % in order to blow up the delay between $U$ and $T$ on
    % input $\rho(w)$ (up to $k$ as $U$ is a seq-$k$-uniformiser of
    % $T$). This blow-up only happens when on $w$, the delay between
    % $U$ and $T$ is close to $N$.
 The main idea is that if, by making
    some good choice of output, 
    $U$ is able to react to a threat of exceeding delay $k$ on
    $\rho(w)$, then by doing the
    same choice on $w$, $U'$ can also react to a threat of exceeding
    delay $N_T$.

    We identify several key properties that $\rho$ must satisfy, in
    order to be able to construct $U'$. For instance, we require that
    $\rho(w)$ is a prefix of $\rho(wa)$ for all $w\in\Sigma^*$,
    $a\in\Sigma$, but also some property relating the delays between
    $U$ and $T$ on  input $w$ and on input $\rho(w)$.

    The challenging part of the proof is to prove the existence of $N_T$ and
    $\rho$. It is based on a study of the structural properties of
    the transition monoid of finitely ambiguous transducers (a monoid
    that captures the state behaviour of automata and transducers),
    and the effect of its elements on the delays. In particular,
    subwords of $w$ that are iterated to define $\rho(w)$ correspond to idempotent
    elements in the transition monoid of $T$, and the bound $N_T$ is
    obtained by an application of Ramsey's theorem.
\end{proof}

Since by Theorem~\ref{thm:rat} every finite-valued transducer is
effectively equivalent to a finitely ambiguous transducer, the
sequential uniformisation problem for finite-valued transducers reduces to sequential
$N$-uniformisation, for computable integers $N$. Hence by Theorem
\ref{thm:rat} and the fact that any transducer can be trimmed in
polynomial time, we get decidability of sequential uniformisation of
finite-valued transducers, one of the main results of this paper. 

% This proves that for every transducer $\tra$ given as a union of
% unambiguous transducers, the uniformisation problem reduces to the
% $N$-delay uniformisation problem for a computable $N$, which is
% decidable, by Theorem \ref{thm:rat}.
% Once again, as any finite-valued transducer can be transformed into an
% equivalent union of unambiguous transducers, based on Theorem
% \ref{trans_fv}, we have the following corollary.

\begin{corollary}\label{fv_unif}
The sequential uniformisation problem for finite-valued
transducers is decidable.
\end{corollary}

%%% Local Variables:
%%% mode: latex
%%% TeX-master: t
%%% End:

\vspace{-1em}
\section{Deterministic Rational Transductions} \label{sec:drat}

In this section we consider another subclass of rational
transductions, namely the deterministic rational transductions,
denoted by DRat. This class is defined in terms of specific
deterministic transducers and some problems that are undecidable for
general rational transductions are decidable in the case of DRat. For
example, the equivalence problem is decidable \cite{Bird73} (while
inclusion is easily seen to be undecidable \cite{FischerR68}), and
whether a given relation in DRat is recognisable \cite{CartonCG06} is
also decidable.
We obtain here another decidability result, namely that the sequential
uniformisation problem is decidable for deterministic transducers. 

For the definition of deterministic rational transducers, we work with
endmarkers (this is the common way of doing it, see also \cite{DBLP:books/daglib/0023547}). The determinism includes the standard definition of unique
successor states for each symbol and additionally a deterministic
choice between input and output. This is enforced by a partition of
the state space into states processing input symbols and states
processing output symbols.

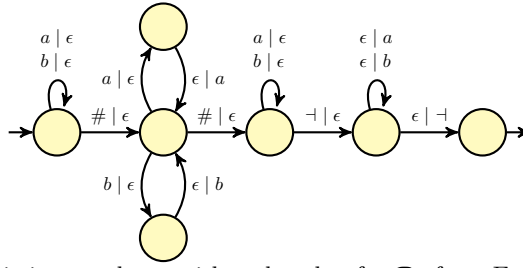
\begin{figure}[t]
\begin{center}
\begin{tikzpicture}[->,>=stealth',auto,node
    distance=1cm,thick,scale=0.7,every node/.style={scale=0.7}]
  \tikzstyle{every state}=[fill=yellow!30,text=black]
  % \tikzstyle{every edge}=[draw=black,text=red]
  \tikzstyle{initial}=[initial by arrow, initial where=left, initial text=]
  \tikzstyle{accepting}=[accepting by arrow, accepting where=right, accepting text=]
  \tikzstyle{acceptingdown}=[accepting by arrow, accepting where=below, accepting text=]

  \tikzstyle{initialT1}=[initial by arrow, initial where=left, initial
  text=]

  \node[initialT1,state] (A) at (0,0) {};
  \node[state]           (B) at (2,0) {};
  \node[state]           (C) at (4,0) {};
  \node[state]           (D) at (6,0) {};
  \node[state,accepting] (E) at (8,0) {};
  \node[state]           (F) at (2,2) {};
  \node[state]           (G) at (2,-2){};
  
  \path (A) edge [loop above] node {$\begin{array}{l}\trans{a}{\epsilon}\\\trans{b}{\epsilon}\end{array}$} (A);
  \path (A) edge node {\trans{\#}{\epsilon}} (B);
  \path (B) edge node {\trans{\#}{\epsilon}} (C);
  \path (B) edge [bend left] node {\trans{a}{\epsilon}} (F);
  \path (B) edge [bend right,swap] node {\trans{b}{\epsilon}} (G);
  \path (C) edge [loop above] node {$\begin{array}{l}\trans{a}{\epsilon}\\\trans{b}{\epsilon}\end{array}$} (C);
  \path (C) edge node {\trans{\dashv\ }{\epsilon}} (D);
  \path (D) edge [loop above] node {$\begin{array}{l}\trans{\epsilon}{a}\\\trans{\epsilon}{b}\end{array}$} (D);
  \path (D) edge node {\trans{\epsilon }{\ \dashv}} (E);
  \path (F) edge [bend left] node {\trans{\epsilon}{a}} (B);
  \path (G) edge [bend right,swap] node {\trans{\epsilon}{b}} (B);
\end{tikzpicture}
\vspace{-2em}
\end{center}
\caption{A deterministic transducer with endmarker for $\mathcal R_1$ from Ex.~\ref{ex:drat}.}
\vspace{-1em}
\label{fig:drat} 
\end{figure}

\begin{example}\label{ex:drat}
 The transduction $\mathcal R_1 = \{ (u\#v\#w,vx) \mid u,v,w,x \in \{a,b\}^*\}$ is in DRat since it is recognised by the deterministic transducer depicted in Fig.~\ref{fig:drat} over the alphabet 
 $\{a,b,\#\}$ with endmarker $\dashv$. Note that for each state the outgoing transitions
 either all have $\epsilon$ as  output component, or all have
 $\epsilon$ as input component. In the formal definition, this
 is captured by the partition into input and output states. 
\end{example}

Let $\Sigma$
be an alphabet and $\dashv$
a fresh symbol used as endmarker. We let
$\Sigmaend := \Sigma \cup \{\dashv\}$.
A deterministic transducer over the alphabet $\Sigma$
with endmarker $\dashv$
is of the form $\tra = (\Qin,\Qout,F,q_0,\delta)$
with a set $\Qin$
of input states, a set $\Qout$
of output states (we write $Q$
for the union of these two sets), a unique initial state $q_0$,
a transition function
$\delta: Q \times \Sigmaend \rightarrow Q$,
and a set $F \subseteq Q$
of accepting states. In the presence of endmarkers, the final output function is not required anymore.

For defining the semantics of such a deterministic transducer, one can transform it into a standard transducer. However, this transformation needs to take care of the endmarker only being allowed at the end of the word, which is not enforced in the definition of deterministic transducers. To avoid this, we rather define the semantics by extending the transition function to pairs of words (input and output word).
For $(u,v) \in \Sigmaend^* \times \Sigmaend^*$
and $q \in Q$,
we define
$\delta^*: Q \times \Sigmaend^* \times \Sigmaend^* \rightarrow Q
\times \Sigmaend^* \times \Sigmaend^*$ inductively as follows:
\begin{itemize}
\item If $q \in \Qin$, then $\delta^*(q,\epsilon,v) =
  (q,\epsilon,v)$ and $\delta^*(q,au,v) =
  \delta^*(\delta(q,a),u,v)$.
\item If $q \in \Qout$, then $\delta^*(q,u,\epsilon) =
  (q,u,\epsilon)$ and $\delta^*(q,u,av) = \delta^*(\delta(q,a),u,v)$.
\end{itemize}
So $\delta^*$
applies $\delta$
to the next input letter from states in $\Qin$
and to the next output letter from states in $\Qout$
as long as possible. The transduction $\rel{\tra}$
defined by $\tra$ is
\[
  \rel{\tra} = \{(u,v) \in \Sigma^* \times \Sigma^* \mid
  \delta^*(q_0,u\dashv,v\dashv) = (q,\epsilon,\epsilon) \mbox{
    with } q \in F\}.
\]
Recall from Section~\ref{sec:fvalued} that $k$-delay inclusion and
equivalence are
complete for finite-valued transducers, as stated in
Theorem~\ref{lem:finitelyambiguous}. We note that this is not the case
for DRat. 
\begin{remark}
  There are deterministic transducers $\tra_1$ and $\tra_2$ such that
  $\tra_1 \equiv \tra_2$ but there is no $k$ such that $\tra_1 \equiv_k \tra_2$.
\end{remark}
\begin{proof}
  Consider the complete relation $\Sigma^* \times \Sigma^*$,
  and let $\tra_1$
  be the deterministic transducer that first reads all input symbols
  (up to the endmarker $\dashv$),
  and then reads all output symbols. Let $\tra_2$
  be the deterministic transducer that first reads all output symbols
  and then the input symbols. Obviously,
  $\rel{\tra_1} = \rel{\tra_2} = \Sigma^* \times \Sigma^*$.
  However, the lag for the two runs of $\tra_1$,
  $\tra_2$ on a pair $(u,v)$ is $|v|$ and thus not bounded.
\end{proof}

Our main result for DRat is the following, which extends the corresponding
result for automatic relations from \cite{CarayolL14}.
\begin{restatable}{theorem}{dratunif}\label{the:drat-unif}
 The sequential uniformisation problem for deterministic transducers is decidable. 
\end{restatable}
The proof uses the game-theoretic approach, building a game between
players Input and Output. A winning strategy for player Output then
corresponds to a sequential uniformiser. The moves of the game simulate the deterministic transducer $T$ on the pairs of input and output word played by the two players in order to check whether the output indeed matches the input. However, Output might need to delay the moves to gain some lookahead on the input for making the next decisions. The main challenge in the proof is to find a way to keep the lookahead information bounded without losing too much information. It is not sufficient to simply store words of bounded length as lookahead. The information in the lookahead rather provides information on the behaviour that the lookahead word induces in $T$. Player Output can delete parts of this information to reduce the size of the lookahead.

The sequential uniformiser that is constructed from the game in the
decidability proof can be shown to have bounded delay. So we obtain
the following result, showing that sequential $k$-uniformisation is complete for deterministic transducers. 

\begin{restatable}{theorem}{dratkdelay}\label{the:drat-k-delay}
    Any sequentially uniformisable deterministic transducer is
    sequentially $k$-uniformisable for some $k\in \mathbb{N}$ that can be computed from the given transducer.
\end{restatable}

%%% Local Variables:
%%% TeX-master: "main"
%%% End:
\section{Conclusion} \label{sec:conclusion}
We have introduced the notion of resynchronisers, which are
transformations for synchronisations of transductions. The decision
problems of inclusion, equivalence, and sequential uniformisation, which
are undecidable for general rational transductions, become decidable
modulo rational resynchronisers. Furthermore, we have shown that it is
sufficient to consider $k$-delay resynchronisers in the context of
these decision problems. We have  analysed two subclasses of
transducers, finite-valued transducers and deterministic
transducers. For both classes, sequential uniformisation is decidable, and the
existence of a sequential uniformiser implies the existence of a
sequential $k$-uniformiser. Additionally, for finite-valued transducers $k$-inclusion is shown to be complete.
One interesting open question is the problem of deciding for a
transducer whether it admits a sequential $k$-uniformiser for some $k$. 

\vspace*{-5mm}
\bibliographystyle{plain}
\bibliography{biblio}

\newpage
\appendix
\noindent{\bf \Large Appendix}
\medskip 

\section{Details for Section \ref{sec:synch}}

\decisionIproblems*

\begin{proof}
\newcommand{\blank}{\ensuremath{\mbox{\protect\raisebox{.65ex}{\ensuremath{\underbracket[.5pt][1.5pt]{\hspace*{1ex}}}}}}}
It remains to show \textsc{ExpTime-hardness} of sequential $\mathbb{I}$-uniformisation.

Let $M$ be a polynomial space bounded alternating Turing Machine that solves some \textsc{ExpTime-hard} decision problem.
To show \textsc{ExpTime-hardness} of sequential $\mathbb{I}$-uniformisation, we give a polynomial time reduction from the word problem for $M$.

Given a word $w$, we construct a synchronous transducer $T$ of polynomial size in $|w|$ that is sequentially $\mathbb{I}$-uniformisable if, and only if, $M$ accepts $w$.
Let $Q_M$ denote the state set of $M$, partitioned into universal states $Q_\forall$ and existential states $Q_\exists$, and let $\Gamma_M$ denote the tape alphabet of $M$ including the blank symbol \blank.
Since a computation on $w$ uses space bounded by a polynomial $p$, we can encode a configuration of $M$ as a string of length $p(|w|)+1$ in the form $xqy\blank$, where $xy\blank \in \Gamma_M^*$ is the content of the tape of $M$ with an additional \blank-symbol at the end, $q \in Q_M$ is the current control state of $M$, and the head of $M$ is on the first position of $y$.
The additional symbol at the end was added because later on we want to mark an update error in a configuration just after the error occurs; if the errors occurs in the last position of $y$, then we need one more position to mark this.
Wlog, we can assume that every computation of $M$ on $w$ is halting (and the resulting configuration is either accepting or rejecting).

For both, input and output, we are interested in words over the alphabet $\Gamma_M \cup Q_m \cup \{\square,\$,\#\}$ of the form
$
 \# x_1\# x_2\# \dots \# x_n
$
where each $x_i$ is either a configuration of $M$, or a string over $\{\square,\$\}$ of length $p(|w|) + 1$. We say that such a word is a correct coding.

The basic idea is that $u = \# u_1\# u_2\# \dots \# u_n$ and $v = \# v_1\# v_2\# \dots \# v_n$ are both correct codings, where the next configuration is $u_{i+1}$ if the previous configuration was universal, and $v_{i+1}$ is the next configuration if the previous configuration was existential. We build the transduction such that a sequential $\mathbb{I}$-uniformiser chooses the successors of existential configurations.
The $\$$ is used as a marker that will occur at most once in the input word. The pair is rejected if this marker identifies an error in the configuration chosen on the output. Below we give the conditions for a pair to be accepted.

%We are interested in pairs $(u,v) \in (\Gamma_M \cup Q_m \cup \{\square,\$\})^* \times (\Gamma_M \cup Q_m \cup \{\square,\$\})^*$ in which the first component is of the form
%\begin{center}
% $u = \# u_1\# u_2\# \dots \# u_n$,
%\end{center}
%where each of the following holds
%\begin{itemize}
% \item Each $u_i$ is either a configuration of $M$, or a string over $\{\square,\$\}$ of length $p(|w|) + 1$.
% \item $u_1$ is the initial configuration.
% \item $u$ contains exactly one $\$$, or $u_n$ is a rejecting configuration.
%\end{itemize}
%We say a word $u$ is a correct coding if it is of this form.
%If the first component $u$ is a correct coding of the form $\# u_1\# u_2\# \dots \# u_n$, then the second component $v$ has to be a correct coding of the form $\# v_1\# v_2\# \dots \# v_n$.
%Let $c_i$ refer to the configuration encoded by $u_i$ or $v_i$.

We define $R$ to be the transduction such that a pair $(u,v) \in R$ if it satisfies the conditions 1.--3. and at least one of 4.a., 4.b., or 4.c.\ given below.
\begin{enumerate}
 \item $|u| = |v|$, and both $u = \# u_1\# u_2\# \dots \# u_n$ and $v = \# v_1\# v_2\# \dots \# v_n$ are correct codings. Furthermore,
   \begin{itemize}
   \item $u_1$ is the initial configuration of $M$ on $w$, and
   \item $u$ contains exactly one $\$$, or $u$ contains no $\$$ and in this case $u_n$ is either a rejecting configuration or not a configuration.
   \end{itemize}
 \item At least one of $u_i,v_i$ is a configuration.

       In the following we let $c_i = u_i$ if $u_i$ is a configuration, and otherwise $c_i = v_i$.  
 \item If $c_i$ is existential, then $v_{i+1}$ is a configuration.
 \item  \begin{enumerate}
 \item The input word $u$ introduces a mistake w.r.t.\ a universal configuration:
 \begin{itemize}
  \item $c_i$ is universal and $u_{i+1}$ is a not configuration, or
  \item $c_i$ is existential and $u_{i+1}$ is a configuration, or
  \item there exists $i$ such that $c_i$ is universal and $u_{i+1}$ is not a successor configuration of $c_i$.
 \end{itemize} 
 \item The input word $u$ marks a position with $\$$ that is not a mistake, i.e., there exists $i$ such that $c_i$ is existential, $u_{i+1}$ contains $\$$ at position $j$ and the positions $j-3,j-2,j-1$ in $v_{i+1}$ are correctly updated w.r.t.\ the positions $j-3,j-2,j-1$ of $c_i$% and $v_{i+1}$ is accepting.
 \item $u_n$ is an accepting configuration, or $u$ does not contain $\$$ and $v_n$ is a configuration but not rejecting.
\end{enumerate}
\end{enumerate}

We build a transducer $T$ that recognises $R$ and works synchronously, i.e., each transition is labelled by a pair of letters.
Therefore, $T$ guesses whether condition 4.a., 4.b., or 4.c., holds.
Depending on its guess, $T$ goes to $T_A$ or $T_B$ resp.\ $T_C$ in order to verify the guess as well as verify that conditions 1.--3.\ are satisfied.

First, before we describe $T_A$, $T_B$ and $T_C$, note that a transducer can easily synchronously check whether a pair $(u,v)$ satisfies conditions 1--3.
Such a transducer needs states polynomial in $|w|$ since it has to verify that the length of each $u_i$ resp.\ $v_i$ is exactly $p(|w|)+1$ and furthermore check that $u_1$ the initial configuration.
Checking the other conditions requires only a constant number of states.
The transducers $T_A$, $T_B$ and $T_C$ can be modified to do this check by a product construction with the above described transducer.

Now, we construct a transducer $T_A$ that verifies that condition 4.a.\ holds. The first two options of 4.a.\ are easy to check.
For the third possibility, a transducer guesses a position in a universal configuration and stores the letters of the three consecutive positions, advances $(p(|w|)+1)-2$ positions, and verifies that there is a mistake in the update.
The size of such a transducer is polynomial in $|w|$.

We construct a transducer $T_B$ that verifies that condition 4.b.\ holds.
$T_B$ guesses a position $j-3$ in an existential configuration $c_i$, stores the letters of $c_i$ of the three consecutive positions $j-3$, $j-2$, $j-1$, and advances $(p(|w|)+1)-2$ positions.
Then, $T_B$ has reached the position $j-3$ in $u_{i+1}$ and $v_{i+1}$ and verifies that the $M$-computation update was correct in the positions $j-3$, $j-2$ and $j-1$ in $v_{i+1}$ w.r.t.\ the stored values, and also checks that $u_{i+1}$ contains $\$$ in position $j$. % and that $c_{i+1}$ is accepting.
A transducer of size polynomial in $|w|$ suffices.

A transducer $T_C$ of constant size for condition 4.c.\ can be easily constructed, it has to guess the beginning of $u_n$ and verify that $u_n$ is an accepting configuration.

Note that $T_A$, $T_B$, and $T_C$ have a polynomial size in $|w|$, modifying the transducers to also check whether conditions 1.--3.\ hold again yields transducers of polynomial size.
Then $T$ is the union transducer of (the modified versions of) $T_A$, $T_B$ and $T_C$.

We claim that $T$ has a sequential $\mathbb{I}$-uniformiser if, and only if, $M$ accepts $w$.
To begin with, assume that $M$ accepts $w$, then a $\mathbb{I}$-uniformiser $U$ for $T$ can be constructed as follows.
The computation of $M$ on $w$ can be represented in a computation tree.
Each node of the tree is labelled with a configuration $c$ and its children are all successor configurations if $c$ is universal, and one successor configuration if $c$ is existential.
The root is labelled with the initial configuration of $M$ on $w$.
We assumed that every $M$-computation is halting.
Thus the computation tree of $M$ on $w$ is finite, and all leaves are labelled with accepting configurations.
The idea is to use this computation tree to build a synchronous finite state uniformiser for $T$ that works as follows.
At first, assume that the input word $u$ is a correct coding.
Assume the transducer has read $\#u_1\#\dots \#u_i$ so far and has produced $\#v_1\#\dots \#v_i$ such that $c_1,\dots,c_i$ are nodes along a path in the computation tree.
We further assume that $U$ as stored $c_i$ in its state. We distinguish two cases.

For the first case, assume $c_i$ is a universal configuration, then $u_{i+1}$ has to be a configuration (otherwise condition 4.b is satisfied and the uniformiser can produce anything that satisfies 1--3).
The transducer then reads $u_{i+1}$ stores its value and produces a $\square$-sequence of the same length.
If $u_{i+1}$ was not a valid  successor configuration of $c_i$, then condition 4.a.\ is satisfied.
Then the uniformiser reads the remainder of the input and just has to make sure to produce output such that conditions 1.--3.\ are satisfied.
Otherwise, if $u_{i+1}$ was a valid successor configuration of $c_i$, $U$ proceeds with the procedure we are currently describing.

For the second case, assume $c_i$ is an existential configuration.
Then $v_{i+1}$ has to be a configuration.
Since $M$ accepts $w$, the accepting computation tree contains a successor of $c_i$, which is then produced by $U$ for $c_{i+1}$ (one can show that this successor can be chosen only based on $c_i$ without knowing the whole computation tree).

If the input word $u$ does not satisfy condition 1, then $U$ simply rejects because $u$ is not in the domain of $R$.

We have to show that $U$ is indeed a sequential $\mathbb{I}$-uniformiser.
If the input $u$ is a correct coding, then it is clear that $U$ produces $v$ such that $(u,v)$ satisfies conditions 1--3.
It is left to prove that one of 4.a.--4.c.\ is satisfied. So assume that 4.a is not satisfied, which means that the input does not introduce any mistake. If the input contains a $\$$, then 4.b is satisfied because the uniformiser correctly updates the configurations. If the input does not contain $\$$, then (by condition 1) $u_n$ is either a rejecting configuration or not a configuration. The first case is not possible because we assume that the input does not introduce a mistake, and hence the simulated computation of $M$ cannot reach a rejecting configuration. Thus, $u_n$ is not a configuration. Then $v_n$ is a configuration and since it is not rejecting, 4.c is satisfied.

Conversely, assume that there is a sequential transducer $U'$ that $\mathbb{I}$-uniformises $T$. This implies that the transitions of $U'$ are labelled by pairs of letters, otherwise it can not be the case that it is an $\mathbb{I}$-uniformiser for $T$. We can use $U'$ to show that $M$ accepts $w$. We start with the input $\#u_1$ with $u_1$ the initial configuration of $M$ on $w$. If $u_1$ is universal, then we can pick an arbitrary successor configuration $u_2$ and continue extend the input to $\#u_1\#u_2$. We continue until in  $\#u_1 \cdots \#u_i$ an existential configuration $u_i$ is reached. Then we continue the input with $\square$, and $U'$ starts producing a next configuration $v_{i+1}$ (otherwise it would not be a uniformiser of $T$).  If a prefix $v'$ of $v_{i+1}$ is reached with a mistake that shows that $U'$ does not produce a successor configuration of $u_i$, then we add $\$$ to the input, and continue with $\square$ until the required length of $u_{i+1}$ is reached. The resulting pair is such that $\#u_1 \cdots \#u_{i+1}$ is in the domain of $T$ but the pair is not accepted because non of 4.a, 4.b, or 4.c is satisfied. Thus, $v_{i+1}$ must be a successor of $u_i$ because $U'$ is a uniformiser of $T$. We continue this process (picking successors of universal configurations on the input and letting $U'$ produce successors of existential configurations on the output), faithfully simulating a computation of $M$ on $w$. This computation must halt eventually. Assume that the last configuration $c_n$ is rejecting. If the second last configuration was universal, then $u_n = c_n$, and the input word is in the domain of $T$. However, none of 4.a--4.c is satisfied, contradicting the assumption that $U'$ is a uniformiser of $T$.
The other case is that the second last configuration was existential, and thus $v_n = c_n$. Then $u_n$ is not a configuration and the input word is in the domain also in this case. But again 4.a--4.c are not satisfied, contradicting the choice of $U'$ as uniformiser of $T$. 

We conclude that the above simulation always reaches an accepting configuration of $M$, no matter how we choose the successors of universal configurations. This implies that $M$ accepts $w$. 
\end{proof}

\decisionSproblems*

\begin{proof}
The lower bounds are obtained as a direction consequence of
Proposition \ref{prop:Iuniformisation}.

To establish decidability, note that since $\sync$ is rational, for every transducer $\tra$, $\sync(\lang{\tra})$ is
(effectively) regular \cite{berstel2009}. The regular language
$\sync(\lang{\tra})$ is a subset of $(\inouta)^*$ and it defines the
transduction $\rel{\sync(\lang{\tra})}$, which is clearly synchronised
by $\sync(\lang{\tra})$. Therefore by Nivat's theorem (Theorem~\ref{thm:nivat}),
$\rel{\sync(\lang{\tra})}$ is rational, and since Nivat's theorem is
effective, one can construct a transducer $\tra^\sync$ such that
$\lang{\tra^\sync} = \sync(\lang{\tra})$.

Given two transducers $\tra_1$ and $\tra_2$, we have $\tra_1
\subseteq_{\sync} \tra_2$ iff $\lang{\tra_1}\subseteq
\sync(\lang{\tra_2})$ iff $\lang{\tra_1}\subseteq \lang{\tra_2^\sync}$
iff $\tra_1\subseteq_{\mathbb{I}} \tra_2^\sync$. In other words,
$\sync$-inclusion reduces to $\mathbb{I}$-inclusion. The same
arguments applies for equivalence and sequential uniformisation, and hence we
decidability by Proposition \ref{prop:Iuniformisation}.

To get the upper-bounds, one can show that $\tra^\sync$ has a
polynomial size in the size of $\tra$ and in the size of the
transducer realising $\sync$, for all transducers $\tra$. First, $\lang{\tra}$ is defined by 
the underlying automaton of $\tra$, which has polynomial size in
$\tra$. Then, $\sync(\lang{\tra})$ is obtained first, by restricting
the domain of transducer for $\sync$ to $\lang{\tra}$, thus obtaining
a new transducer $\tra'$ (which can constructed in polynomial-time via
a product construction), and then by projecting the inputs of $\tra'$,
thus obtaining an automaton $\aut'$ that recognised
$\sync(\lang{\tra})$. Finally, it remains to turn $\aut'$ (which is
over the alphabet $\inouta^*$) into a transducer $\tra^\sync$ (over
$\Sigma$), according to Nivat's theorem. Again, this can be done in
polynomial time: first, one can make $\aut'$ to have transitions with
single letters only, by splitting transitions on words of length $n$,
$n>1$, into $n$ transitions. Then, any transition on a symbol
$\sigma^\inp$ is replaced by a transition on $\sigma\mid
\epsilon$, and any transition on $\sigma^\outp$ is replaced by a
transition on $\epsilon\mid \sigma$. This can be done in
polynomial-time. 
\end{proof}

%--------------------------------------------------
% k-delay rational
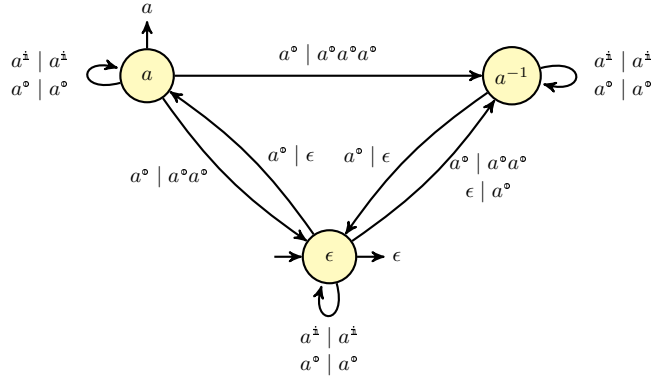
\begin{figure}[!ht]
\begin{center}
\begin{tikzpicture}[->,>=stealth',auto,node distance=3cm,thick,scale=0.8,every node/.style={scale=0.8}]
  \tikzstyle{every state}=[fill=yellow!30,text=black]
  % \tikzstyle{every edge}=[draw=black,text=red]
  \tikzstyle{initial}=[initial by arrow, initial where=left, initial text=]
  \tikzstyle{acceptinge}=[accepting by arrow, accepting where=right, accepting text=$\epsilon$]
  \tikzstyle{acceptinga}=[accepting by arrow, accepting where=above, accepting text=$a$]

  \node[initial,acceptinge,state] (A) at(0,0) {$\epsilon$};
  \node[acceptinga,state] at(-3,3) (B)  {$a$};
  \node[state] at(3,3) (C)  {$a^{-1}$};
  \path (A) edge [bend right=10] node [right] {$\begin{array}{l}\trans{a^\outp}{\epsilon}\end{array}$} (B);
  \path (A) edge [bend right=10] node [right]
  {$\begin{array}{c}\trans{a^\outp}{a^\outp a^\outp} \\ \trans{\epsilon}{a^\outp}\end{array}$} (C);
  \path (A) edge [loop below] node {$\begin{array}{l}\trans{a^\inp}{a^\inp}\\ \trans{a^\outp}{a^\outp}\end{array}$} (A);
  \path (B) edge [loop left] node {$\begin{array}{l}\trans{a^\inp}{a^\inp}\\ \trans{a^\outp}{a^\outp}\end{array}$} (B);
  \path (C) edge [loop right] node {$\begin{array}{l}\trans{a^\inp}{a^\inp}\\ \trans{a^\outp}{a^\outp}\end{array}$} (C);
  \path (B) edge [bend right=10] node [left]
  {$\begin{array}{l}\trans{a^\outp}{a^\outp a^\outp}\end{array}$} (A);
  \path (C) edge [bend right=10] node [left]
  {$\begin{array}{l}\trans{a^\outp}{\epsilon}\end{array}$} (A);

  \path (B) edge node [above]
  {$\begin{array}{l}\trans{a^\outp}{a^\outp a^\outp a^\outp}\end{array}$} (C);
\end{tikzpicture}
 \caption{\label{fig:tradelay}  $1$-Delay synchroniser for $\Sigma = \{a\}$}
\end{center}
\end{figure}

\dsyncrational*
\begin{proof}
We construct a transducer $\tra_k = (Q,I,F,\Delta,f)$ over $\inouta$ such that
$\rel{\tra_k} = \dsync_k$, by:
\begin{itemize}
    \item $Q = \{ u\in\Sigma^*\cup (\Sigma^{-1})^*\mid k\geq |u|\}$,
      $I = \{\epsilon\}$ and $F = Q\cap \Sigma^*$
    \item $f(u) = u^\outp$ for all $u\in F$
    \item $\Delta = \{ (u,x^\inp,x^\inp,u)\mid u\in
      Q,x\in\Sigma\}\cup \{(u,x^\outp,v^\outp,\del(v,ux))\mid
      x\in\Sigma\wedge u,v\in\Sigma^*\wedge u,v^{-1}u x\in Q\} \cup
      \{(\epsilon,\epsilon,v^\outp, v^{-1})\mid
      |v|\leq k\}$.
\end{itemize}
For $k = 1$ and $\Sigma = \{a\}$, the transducer $\tra_k$ is depicted
on Fig.~\ref{fig:tradelay}. 
\end{proof}
% end k-delay rational
%--------------------------------------------------

%--------------------------------------------------
% complexity of k-delay
\decisionkdelay*

%\manu{still to prove: pspace-completeness for inclusion and equivalence,
%  exptime-completeness for uniformisation}

\begin{proof}
In the case where $k$ is fixed, the upper-bounds are consequences of
Theorem~\ref{thm:rat} and the fact the rational $k$-delay
resynchroniser $\mathbb{D}_k$ can be defined by a transducer of
polynomial size (the transducer constructed in the proof of
Prop.~\ref{prop:dsynckrat} has polynomial size when $k$ is fixed). To
get the lower bounds when $k$ is fixed, it suffices to take $k=0$,
which in that case correspond to $\mathbb{I}$-decision problems, and
to apply Prop.\ref{prop:Iuniformisation}.

When $k$ is not fixed, let us show that the problems gets harder.
Let $M$ be an exponentially space bounded Turing machine that solves some \textsc{ExpSpace-hard} decision problem.
 To show \textsc{ExpSpace-hardness} of the $k$-inclusion,
 $k$-equivalence and sequential $k$-uniformisation problem, we give a polynomial time reduction from the word problem for $M$, respectively.
 
 We start with the \textsc{ExpSpace-hardness} of sequential $k$-uniformisation.
 Given a word $w$, we construct a transducer $T$ that has a sequential
 $\big((n+3)2^{n}+1\big)$-uniformiser if, and only if, $M$ rejects $w$, where $n = |w|$.
 Let $Q_M$ denote the state set of $M$ and $\Gamma_M$ denote the tape alphabet of $M$.
 Since a computation of $M$ on $w$ uses space bounded by $2^n$, we can encode a configuration of $M$ as a string of $\big((n+3)2^{n}+1\big)$ letters in the form
 \begin{center}
  $a_0\$bin(0)\$a_1\$bin(1)\$a_2\ldots \$bin(2^{n}-1)\$a_{2^n}$,
 \end{center}
 where $a_0\ldots a_{2^n} = xqy$, $xy \in \Gamma_M^*$ is the content of the tape of $M$, $q \in Q_m$ is the current control state of $M$, the head of $M$ is on the first position of $y$, and $bin(i)$ is the binary representation (using $n$ bits) of $i$ for all $i \in \{0,\ldots,2^{n}-1\}$.
 
 We define a transduction, where the core of the reduction are pairs $(u,v) \in \Sigma^* \times \Sigma^*$ in which the first component is of the form
 \begin{center}
  $c_0\#c_1\#\ldots c_i\#c_{i+1}\#\ldots c_\ell\#\kappa\#^*X$,
 \end{center}
 where each $c_i$ is a configuration of $M$, $c_0$ is the initial configuration of $M$ on $w$, $c_\ell$ is a final configuration of $M$, $\kappa$ is the string $\square\$bin(0)\$\square\$bin(1)\$\square\ldots \$bin(2^{n}-1)\$\square$ which has the length of a configuration coding, and $X \in \{A,B\}$.
 We say $u$ is a correct coding, if it is of this form.
 
 If $u$ is not a correct coding, then every word $v$ with $|v| \leq |u|$ is allowed as the second component.
 If $u$ is a correct coding and ends in $A$, then $(u,v)$ is accepted if, and only if, $v=u$.
 If $u$ is a correct coding and ends in $B$, then $v$ has to be of the form
  \begin{center}
  $c'_0\#c'_1\#\ldots c'_i\#c'_{i+1}\#\ldots c'_m\#\Sigma^*$,
 \end{center}
 such that $c_{i+1}$ is not the successor configuration of $c'_i$ for some $i \in \{0,\ldots,\ell-1\}$, and $|v| \leq |u|$.
 
 We build a transducer $T$ of polynomial size in $n$ that recognises this transduction.
 To this end, $T$ guesses at the beginning whether $u$ is a correct coding, and whether it ends in $A$ or $B$.
 Depending on its guess, $T$ goes to $T_w$ if $u$ is not a correct coding, or to $T_A$ resp.\ $T_B$ if $u$ is a correct coding and ends in $A$ resp.\ $B$.
 The transducers $T_w$, $T_A$ and $T_B$ are described below.
 
 First, we construct a real-time transducer $T_w$ that accepts pairs where $u$ is not a correct coding and $|v| \leq |u|$.
 To verify that $u$ is not a correct coding, the transducer has to check that at least one of the following mistakes occurs in $u$:
 \begin{enumerate}
  \item The word $u \notin \big(\big((\Gamma_M \cup Q_M)\$\{0,1\}^n\$\big)^*(\Gamma_M \cup Q_M)\#\big)^*(\square\$\{0,1\}^n\$)^*\square\#^*(A+B)$, or some $c_i$ does not contain exactly one state.
  \item The first configuration $c_0$ is not the initial configuration of $M$ on $w$.
  \item The last configuration $c_\ell$ is not a final configuration of $M$.
  \item The sequence of counter values in some $c_i$ or $\kappa$ is not correct, i.e.,
  \begin{enumerate}
  \item it does not start with $bin(0)$, or
  \item it does not end with $bin(2^n-1)$, or
  \item it contains more than once $bin(0)$, or
  \item it contains two successive binary counter values that are not obtained by increment.
  \end{enumerate}
 \end{enumerate}
 For these properties, we construct transducers that work synchronously, i.e., each transition is labelled by a pair of letters (or the output component is $\epsilon$ if the output has already ended).
 
 The first and the second property can be easily checked by a transducer of size linear in $n$.
 For the third property, a transducer has to guess when it reads the last configuration, then check that is does not read any accepting $M$-state and verify that no further configuration follows. This can be done by a transducer of constant size.
 Properties 4.(a)-(c) can be verified by a transducer that guesses the beginning of a configuration that contains a mistake and then verifies this.
 A transducer of size linear in $n$ suffices, since the binary counters use $n$ bits.
 Property 4.(d) requires a transducer to guess the bit of a binary counter whose update will be faulty and then to verify the guess.
 Upon reading this bit, its position and value is guessed and stored.
 Then, the transducer verifies that the guess of the position was correct by counting the remaining bits and in the process also checks and remembers whether this bit has to be flipped by testing for zeros.
 Thereafter, when the next binary counter begins, the transducer counts up to the right position and checks that the update was incorrect.
 Such a transducer needs states linear in $n$.
 
 A union of these transducers accepts pairs $(u,v)$, where $u$ is not a correct coding.
 This union-transducer can be easily modified to accept only pairs, where $|v| \leq |u|$, which results in the desired transducer $T_w$. 

 Secondly, we construct a real-time transducer $T_A$ that accepts a pair $(u,v)$ if $u$ ends in $A$ and $v = u$.
 For this purpose, $T_A$ synchronously tests whether input and output are equal and whether the last letter is $A$. Only 2 states are needed.
 
 Lastly, we construct a real-time transducer $T_B$ that accepts a pair $(u,v)$ if $u$ ends in $B$, and there is some $i$ such that $c_{i+1}$ is not the $M$-successor configuration of $c'_i$, and $|v| \leq |u|$. 
 At the beginning, $T_B$ reads the first configuration $c_0$, but does not read any part of the output.
 This then allows $T_B$ to read $c_{i+1}$ and $c'_i$ in parallel for further indices $i > 0$.
 To verify that $c_{i+1}$ is the $M$-successor of $c'_i$, $T_B$ checks that both configurations only differ and have been correctly updated at positions $j-1,j,j+1$, where the $j$th position of $c'_i$ contains an $M$-state.
 This is checked until $T_B$ has read some $c_{i+1}$ and $c'_i$ such that $c_{i+1}$ is not the successor configuration of $c'_i$.
 Then, $T_B$ reads the remainder of $u$ and $v$ in parallel until the beginning of $\kappa$ is reached.
 So far, $T_B$ has processed $|c_0|$ letters more of $u$ than of $v$, but it remains to be checked in real-time that $|v| \leq |u|$.
 Thus, $T_B$ has to catch up $|c_0|$-output letters before the input ends.
 For this purpose $\kappa$ was introduced to the correct coding of $u$.
 Since $|\kappa|=|c_0|$, it suffices that $T_B$ reads two output-letters per read input-letter while reading $\kappa$.
 This guarantees that after reading $\kappa$ the same amount of $u$ and $v$ has been processed by $T_B$.
 Subsequently, $T_B$ reads the rest in parallel and verifies that $|v| \leq |u|$ and that $u$ ends with $\#^*B$. 
 The size of such a transducer is constant.
 
 It follows that $T$ can be constructed from $M$ and $w$, and the state space of $T$ is polynomial in $n = |w|$.
 We claim that $T$ has a sequential $\big((n+3)2^{n}+1\big)$-uniformiser if, and only if, $M$ rejects $w$.
 Assume $M$ does not accept $w$, then the sequential transducer that
 synchronously realises the identity function is a sequential uniformiser of $T$.
 Let $U$ be this transducer.
 Obviously, $(u,U(u)) \in \rel{T}$ in case $u$ is not a correct coding or ends in $A$.
 Hence, we verify the case where $u$ is a correct coding and ends in $B$.
 Since $M$ rejects $w$, the configuration sequence $c_0\#c_1\#\ldots\#c_\ell$ must contain two configurations $c_i$ and $c_{i+1}$ such that $c_{i+1}$ is not the successor configuration of $c_i$.
 Since $U(u) = u$, and in particular $c'_i = c_i$, we obtain $(u,U(u)) \in \rel{T}$.
 Finally, note that $U \subseteq_k T$, where $k = \big((n+3)2^{n}+1\big)$.
 This follows from the construction of $T_B$.
 
 Conversely, assume that $M$ accepts $w$ and there is a sequential transducer $U'$ that uniformises $T$.
 We consider an input $u$ that is a correct coding and codes the accepting computation of $M$ on $w$ followed by $\#\kappa\#^j$ for some $j$ such that $|\#\kappa\#^j|$ is longer than $l_1  + l_2$, where $l_1$ is the maximal length of an output used in a transition of $U'$ and $l_2$ is the maximal length of an output of the final output function of $U'$.
 Then, after consuming the input thus far, $U'$ must have produced the accepting computation of $M$ on $w$, because the next input letter could be $A$ and end the input.
 However, if the next letter is $B$, then the output does not satisfy the condition.
 This is a contradiction, no uniformiser exists, and especially no
 seq-$k$-uniformiser exists for any $k$.
 
 Now, we show that the $k$-inclusion problem is \textsc{ExpSpace-hard}.
 Given a word $w$, let $T$ be the transducer constructed from $w$ as described above and let $U$ denote a synchronous transducer that realises the identity function over the same alphabet.
 We claim, $U \subseteq_k T$ if, and only if, $M$ rejects $w$, where $k = \big((n+3)2^{n}+1\big)$ and $n = |w|$.
 Assume $M$ rejects $w$.
 Obviously, $U \subseteq_k T$, because then $U$ is a $k$-uniformiser for $T$ as shown before.
 This implies that $U$ is $k$-included in $T$.
 Conversely, assume that $M$ accepts $w$.
 Consider an input $u$ that is a correct coding that ends in $B$ and codes the accepting computation of $M$ on $w$.
 Then $(u,u) \in \mathcal R_U$, but $(u,u) \notin \mathcal R_T$, because the output does not satisfy the condition that $c_{i+1}$ is not the successor configuration of $c'_i$ for some $i$.
 Hence, $U \not\subseteq T$ and thus $U \not\subseteq_k T$ for any $k$.
 
 Lastly, we show that the $k$-equivalence problem is \textsc{ExpSpace-hard}.
 Given a word $w$, let $T$ be constructed from $w$ as above and $U$ also be as above.
 Recall, $T$ was constructed as the union of the transducers $T_w$, $T_A$ and $T_B$.
 Let $T_1$ be the union of the transducers $T_w'$, $T_A$ and $T_B$.
 The transducer $T_w'$ works like $T_w$ and additionally tests whether input and output are equal.
 Hence, $\mathcal R_{T_1} = ((\mathcal R_{T_w} \cup \mathcal R_{T_A}) \cap \mathcal R_U) \cup \mathcal R_{T_B}$.
 Recall that $R_{T_A} \subseteq \mathcal R_U$ by construction of $T_A$.
 Furthermore, let $T_2$ be the union of $U$ and $T_B$.
 We claim that $T_1 \equiv_k T_2$ if, and only if, $M$ rejects $w$, where $k = \big((n+3)2^{n}+1\big)$ and $n = |w|$.
 
 Assume $M$ rejects $w$.
 First, consider a pair $(u,v) \in \mathcal R_{T_1}$, then clearly also $(u,v) \in \mathcal R_{T_2}$.
 Secondly, consider a pair $(u,v) \in \mathcal R_{T_2}$.
 We distinguish four cases.
 If $u \neq v$, then $(u,v) \in \mathcal R_{T_B}$.
 If $u = v$ and $u$ is not a correct coding, then $(u,v) \in \mathcal R_{T_w'}$.
 If $u = v$ and $u$ is a correct coding that ends in $A$, then $(u,u) \in \mathcal R_{T_A}$.
 If $u = v$ and $u$ is a correct coding that ends in $B$, then $(u,u) \in \mathcal R_{T_B}$: Since $M$ rejects $w$, the configuration sequence $c_0\#c_1\#\ldots\#c_\ell$ must contain two configurations $c_i$ and $c_{i+1}$ such that $c_{i+1}$ is not the successor configuration of $c_i$.
 Consequently, $(u,v) \in \mathcal R_{T_2}$ implies $(u,v) \in \mathcal R_{T_1}$.
 Altogether, $T_1 \equiv T_2$ and by construction of $T_1$ and $T_2$ follows $T_1 \equiv_k T_2$.
 
 Conversely, assume $M$ accepts $w$.
 Consider an input $u$ that is a correct coding that ends in $B$ and codes the accepting computation of $M$ on $w$.
 Then $(u,u) \in \mathcal R_U \subseteq \mathcal R_{T_2}$, but $(u,u) \notin \mathcal R_{T_1}$, because the pair neither satisfies the specification of $T_{w'}$, nor $T_A$ nor $T_B$.
 Hence $T_1 \not\equiv T_2$ and thus $T_1 \not\equiv_k T_2$ for any $k$.
\end{proof}
% end complexity of k-delay
%--------------------------------------------------

The following lemma is used in the proof of
Proposition~\ref{prop:notcomplete} below.
\begin{lemma}\label{lem:lag}
    For all $u,v{\in} (\inouta)^*,\sigma{\in}\Sigma,u',v'{\in}
    (\outa)^*$ such that $\pi_\inp(u) = \pi_\inp(v)$:
    $$
    \lag(u\sigma^\inp u',v\sigma^\inp v') = max\ (\lag(u,v),
        |\del(\pi_\outp(u)u', \pi_\outp(v)v')|)
$$
\end{lemma}

\begin{proof}
Direct by definition of the lag.
\end{proof}

\propapprox*

\begin{proof}
    % Note that showing this result for equivalence implies the result
    % for inclusion. Indeed, one can show that if
    % $(\bigcup_{k=0}^{+\infty} \subseteq_k)\ \subsetneq\
    % (\subseteq)$, then     $(\bigcup_{k=0}^{+\infty} \equiv_k)\
    % \subsetneq\ (\equiv)$. Suppose that two transducers $\tra_1$ and
    % $\tra_2$ are equivalent, then it means that $\tra_1\subseteq
    % \tra_2$ and $\tra_2\subseteq \tra_1$, and therefore by
    % assumption, there exist two natural $k_1,k_2$ such that 
    % $\tra_1\subseteq_{k_1}\tra_2$ and $\tra_2\subseteq_{k_2}
    % \tra_1$. It implies that $\tra_1\equiv_{max(k_1,k_2)} \tra_2$. 

    Consider the transducer
    $\tra_1,\tra_2$ of Fig.~\ref{fig:notk}, which are equivalent.
    We show that $\tra_1\not\subseteq_k\tra_2$ for any $k$, which
    will imply the result for equivalence as well. The transducers
    $\tra_1$ and $\tra_2$
    both
    realise the transduction $\{ (a^n, a^{2i})\mid n\geq 1,
    i=0,\dots,n{-}1\}$. Suppose that $\tra_1\equiv_k \tra_2$ for some
    $k\geq 0$. Therefore, $\tra_1\subseteq_k \tra_2$. Take the input
    word $a^{2k+2}$ and the output word $a^{2k+1}$. There are
    only one word $w_1\in \lang{\tra_1}$ and one word
    $w_2\in\lang{\tra_2}$ such that 
    $(\pi_\inp(w_1),\pi_\outp(w_1)) = (\pi_\inp(w_2),\pi_\outp(w_2)) =
    (a^{2k+2},a^{2k+1})$:
    $$
    w_1 = (a^\inp a^\outp a^\outp)^{k+1} (a^\inp)^{k+1}\qquad w_2 =  (a^\inp)^{k+1}(a^\inp a^\outp
    a^\outp)^{k+1}.$$
    However, $\lag(w_1,w_2)=2k+2>k$, which contradicts
    $\tra_1\subseteq_k \tra_2$.

    Finally, we show the result for sequential uniformisation. Consider the
    transducer $\tra$ of Fig.~\ref{fig:notk}. It realises the
    transduction $\{ (a^nA, a^n)\mid n\geq 0\}\cup \{ (a^nB, a^i)\mid
    n\geq 1,\ i=0,\dots,2n{-}1\}$. Clearly, this transduction is
    uniformisable by the sequential function $\{(a^n\alpha,a^n)\mid
    n\geq 0,\alpha\in\{A,B\}\}$, defined by the sequential
    transducer $U$ of Fig.~\ref{fig:notk}. Note that $U$ generates  an
    arbitrary large lag with $\tra$ on the family of pairs
    $((a^nB,a^n))_{n\geq 1}$. We show that \emph{any} sequential uniformiser
    generate an arbitrary large lag with $\tra$. We prove the result
    for uniformisers without output function (or equivalently, with the
    constant output function that maps accepting states to
    $\epsilon$). This is w.l.o.g. since any $k$-uniformiser of
    $\tra$ with output function can be turned into a $k$-uniformiser
    without output function (it suffices to output the content of the
    output function when reading the ending symbols $A$ and $B$, that
    are unique). 

    Let $D$ be a sequential transducer that uniformises $\tra$,
    \emph{without} output function. First, we need the following
    claim, which gives a lower a bound on the lag between the words of $\lang{T}$ 
    recognised by the upper part and lower part of $\tra$, on input of
    the form $a^n$:
    $$
    \text{\textit{Claim.} For all $n,n_1,n_2{\geq} 0$, $\alpha\in\{\epsilon,a^\outp\}$, }\lag((a^\inp
    a^\outp)^n, (a^\inp)^{n_1} a^\inp \alpha  (a^\inp a^\outp
    a^\outp)^{n_2})\geq \frac{n{-}1}{3} 
    $$

    \textit{Proof of the claim.} Let $L = \lag((a^\inp
    a^\outp)^n, (a^\inp)^{n_1} a^\inp \alpha  (a^\inp a^\outp
    a^\outp)^{n_2})$. If either $n_1+n_2+1 \neq n$ or, 
    $2n_2\neq n$ and $\alpha = \epsilon$, or $2n_2\neq n-1$ and
    $\alpha = a^\outp$, then the two arguments of $L$ are not
    equivalent, and therefore $L = +\infty$.

    Now, assume that $n = n_1 + n_2 + 1$ and $2n_2 = n$ if $\alpha =
    \epsilon$, and $2n_2 = n-1$ if $\alpha = a^\outp$. The worst case
    for the lag is when $\alpha = \epsilon$. It is not difficult to
    show that $\lag((a^\inp
    a^\outp)^n, (a^\inp)^{n_1+1} (a^\inp a^\outp
    a^\outp)^{n_2}) \geq max(n_1,n_2-n_1)$, and therefore $L\geq max(n_1,n_2-n_1)$. Since $n-1 = n_1+n_2$,
    one gets $L\geq max(n_1, n-1-2n_1)$. We now consider two
    cases:
\begin{itemize}
    \item if $n-1 > 3n_1$, then $n-1-2n_1 > n_1$,
      therefore $max(n_1,n-1-2n_1) = n-1-2n_1$ and 
      $L\geq n-1-2n_1$. From $n-1-2n_1 > n_1$ one gets $n-1>3n_1$ and
      $2n_1 < 2(n-1)/3$. Hence, $L \geq n-1-2n_1 > n-1-2(n-1)/3 =
      (n-1)/3$. 
    \item if $n-1 \leq 3n_1$, then $n-1-2n_1
      \leq   n_1$, 
      therefore $max(n_1,n-1-2n_1) = n_1$ and $L\geq n_1 \geq
      (n-1)/3$. 
\end{itemize}

Therefore $L\geq (n-1)/3$. \hfill \textit{End of Proof of the Claim.}

    \vspace{2mm}

    Let $n\geq 1$. Since $D$ uniformises $\tra$ and $a^n$ is the only output
    to $a^nA$ by $\tra$, we have $(a^nA, a^n)\in\rel{D}$. Therefore,
    there exist $u\in(\inouta)^*$ and $n\geq c\geq 0$ such that $(i)$
    $uA^\inp (a^\outp)^c\in \lang{D}$ $(ii)$ $\pi_\inp(u) = a^n$ and
    $\pi_\outp(u) = a^{n-c}$. Since $D$ is sequential, there exists
    $c'\geq 0$ such that $uB^\inp (a^\outp)^{c'}\in\lang{D}$. Now, we
    have, by Lemma~\ref{lem:lag}, for all $v\in(\inouta)^*$ such
    that $\pi_\inp(u) = \pi_\inp(v) = a^n$, 
    $$
    \begin{array}{rcl}
      \lag(uA^\inp (a^\outp)^c, (a^\inp a^\outp)^n A^\inp) & \geq & \lag(u,
                                                                    (a^i a^\outp)^n) \\
      \lag(uB^\inp (a^\outp)^{c'}, v B^\inp) & \geq & \lag(u,v)
    \end{array}
    $$
    Since $\lag$ is a metric, we get:
    $$
    \begin{array}{rcl}
      \lag(uA^\inp (a^\outp)^c, (a^\inp a^\outp)^n A^\inp) +
      \lag(uB^\inp (a^\outp)^{c'}, v B^\inp)  &  \geq & \lag(u,
                                                                    (a^i
                                                        a^\outp)^n) +
                                                        \lag(u,v)\\
      & \geq & \lag((a^ia^\outp)^n,v)
    \end{array}
    $$
    By taking $v$ of the form $(a^\inp)^{n_1} a^\inp \alpha  (a^\inp a^\outp
    a^\outp)^{n_2}$ as in the claim, such that $n = n_1+n_2+1$, one can apply the claim to
    $\lag((a^ia^\outp)^n,v)$ and we get:
    $$      \lag(uA^\inp (a^\outp)^c, (a^\inp a^\outp)^n A^\inp) +
      \lag(uB^\inp (a^\outp)^{c'}, v B^\inp)\geq \frac{n-1}{3}
      $$
      Therefore, we have:
    $$     
    \lag(uA^\inp (a^\outp)^c, (a^\inp a^\outp)^n A^\inp)\geq \frac{n-1}{6}
    \text{ or }
    \lag(uB^\inp (a^\outp)^{c'}, v B^\inp)\geq \frac{n-1}{6}
    $$
      
      Finally, since $(a^\inp a^\outp)^n A^\inp,vB^\inp\in \lang{T}$
      and $uA^\inp (a^\outp)^c,uB^\inp
      (a^\outp)^{c'}\in\lang{D}$, since $(a^\inp a^\outp)^n A^\inp$
      and $uA^\inp (a^\outp)^c$ are equivalent, and since 
      $uB^\inp (a^\outp)^{c'}$ and $v B^\inp$ are equivalent, we have
      found two pairs of equivalent words such that either the first
      pair or the second one has a lag larger than $\frac{n-1}{6}$.
      Since this holds for all $n\geq 1$, $D$ generates arbitrary
      large lags with $\tra$, which contradicts the fact that
      $D\subseteq_k \tra$ for some $k$.
\end{proof}

Note that the latter proposition has been shown for a unary alphabet, for
inclusion and equivalence. We can also strengthen it to a unary
alphabet for sequential uniformisation by replacing in $\tra$
(Fig.\ref{fig:notk}) the property ``ending with $A$'' by
``having an odd number of $a$ symbols'', and ``ending with $B$'' by
``having an even number of $a$ symbols''.

\thmsyncdel*

\begin{proof}
    It suffices to show statement $(1)$. Indeed, statement $(2)$ is
    clearly a consequence of $(1)$. To show $(3)$, assume that $(1)$
    holds and $\tra$ is $\sync$-uniformisable by a sequential
    transducer $U$. Then, $U\subseteq_\sync \tra$, and by $(1)$,
    $U$ is also a $k$-uniformiser.

    The proof of statement $(1)$ is based on the following claim:
    there exists a computable $k\geq 0$ such that for
    all $u\in \lang{\tra_1}$, all $v\in \lang{\tra_2}$ such that $u\in\sync(v)$, $\lag(u,v)\leq k$. 
    Before proving this claim, let us show it implies $(1)$. Suppose
    that $\tra_1\subseteq_\sync \tra_2$, and let
    $u\in\lang{\tra_1}$, then there
    exists $v\in\lang{\tra_2}$ such that $u\in \sync(v)$. By the
    claim, $\lag(u,v)\leq k$ and therefore
    $u\in\dsync_k(\lang{\tra_2})$.

    It remains to prove the claim. 
    The idea is that $\sync$ can advance or delay the productions of
    $\tra_2$, but since it has to preserve the equivalence
    $\sim_{\inout}$ between inputs and outputs, it cannot advance or
    delay arbitrarily far with a finite set of states, as long as
    $\tra_1$ and $\tra_2$ are real-time.

    % We define a mapping $\ell : (\inouta)^2\rightarrow \mathbb{N}$
    % that returns the length of delay between two words $u$ and $v$ up to their common 
    % input projection. That is, if $u = u_1 i_1 u_2 i_2\dots u_n i_n
    % u_{n+1}$ and $v = v_1 j_1 v_2 j_2 \dots v_m j_m v_{m+1}$, where 
    % $u_i,v_i\in\outa^*$ and $i_p,j_p\in\ina$, and if $p_0$ is the
    % largest index such that $i_1\dots i_{p_0} = j_1\dots j_{p_0}$,
    % then $\ell(u,v) = |\del(u_1 i_1 \dots i_{p_0} u_{p_0}, v_1j_1\dots
    % j_{p_0} v_{p_0})|$. Note that if $u\sim_{\inout} v$, then
    % $\ell(u,v) = |\del(u,v)|$. 

    Now, take $u\in\lang{\tra_1}$ and $v\in\lang{\tra_2}$ such that
    $u\in \sync(v)$. By definition of synchronisers, $u\sim_{\inout}
    v$ and therefore they can be decomposed into 
    $u = u_1 i_1 \dots u_n i_n u_{n+1}$ and $v = v_1 i_1 \dots v_n i_n
    v_{n+1}$ such that $u_i,v_i\in \outa^*$ and $i_j \in \ina$. 
    Since $\tra_1$ and $\tra_2$ are real-time, 
    there exists $M$ such that $|u_i|,|v_i|\leq M$ for all
    $i=1,\dots,n+1$ and this $M$ only depends on $\tra_1$ and
    $\tra_2$. Let $\tra_\sync$ be the transducer defining $\sync$, and
    consider the transducer $\tra'_\sync$ obtained by restricting the domain
    of $\tra_\sync$ to $\lang{\tra_2}$ and its range to
    $\lang{\tra_1}$ (it can be easily defined by a
    product construction between $\tra_\sync$ and the underlying
    automata of $\tra_1$ and $\tra_2$). Consider an accepting run $r$
    of $\tra'_\sync$
    on input $v$ and output $u$. Now, 
    assume there is a loop in $r$, i.e. $r$ can be decomposed into 
    $$
    r\ :\ \alpha_0\xrightarrow{w_1\mid w'_1}
    \alpha\xrightarrow{w_2\mid w'_2} \alpha \xrightarrow{w_3\mid w'_3} \alpha_f
    $$
    such that $w_1w_2w_3 = v$ and $w'_1w'_2w'_3 = u$, and
    $w_2\neq\epsilon$ or $w'_2\neq \epsilon$. First, by definition of
    synchronisers, for all $i\geq 0$, we have $w_1(w_2)^i w_3
    \sim_{\inout} w'_1(w'_2)^i w'_3$. Let us show that it implies 
    $|\pi_{\inp}(w_2)| = |\pi_{\inp}(w'_2)|$ and $|\pi_{\outp}(w_2)| =
    |\pi_{\outp}(w'_2)|$. Let $x\in\{\inp,\outp\}$. For all $i\geq 0$,
    we have $\pi_x(w_1(w_2)^iw_3) = \pi_x(w'_1(w'_2)^iw'_3)$, i.e. 
    $\pi_x(w_1)\pi_x(w_2)^i\pi_x(w_3) =
    \pi_x(w'_1)\pi_x(w'_2)^i\pi_x(w'_3)$. It implies that 
    $$
|\pi_x(w_1)|+i|\pi_x(w_2)|+|\pi_x(w_3)| =
|\pi_x(w'_1)|+i|\pi_x(w'_2)|+|\pi_x(w'_3)|\text{ for all $i\geq 0$}
$$
and therefore $|\pi_x(w_2)| = |\pi_x(w'_2)|$, and therefore, both
$w_2$ and $w'_2$ are non-empty and $|w_2| = |w'_2|$.

    % Now, suppose that $\pi_{\inp}(w_2)  = 
    % \pi_{\outp}(w'_2) = \epsilon$. It implies that
    % $w_2,w'_2\in\outa^*$. If $w_2\neq \epsilon$, then by iterating the
    % loop sufficiently many times, say $i$ times, we get an input word 
    % $w_1w_2^i w_3$ such that $w_1w_2^i w_3\in \lang{\tra_2}$ and
    % $|w_2^i| > M$, which contradicts the fact that $\tra_2$ is
    % real-time. Similarly, using the fact that $\tra_1$ is real-time,
    % one can show that $w'_2\neq \epsilon$ yields a contradiction. So
    % far, we have shown that 
    % $w_2$ and $w'_2$ have the same sequence of
    % input symbols, and that sequence contains at least one symbol. 

    Now, take $p\in\{1,\dots,n\}$. We will bound the value
    $|\del(u_1\dots u_p, v_1\dots v_p)|$, thus proving the claim. 
    There are two cases: $u_1\dots u_p\preceq v_1\dots v_p$ or $v_1\dots v_p\preceq
    u_1\dots u_p$. We only consider the case $v_1\dots v_p \preceq
    u_1\dots u_p$, the other being symmetric.

Now, we decompose the run $r$ up to the transition reading input
symbol $i_p$ by taking maximal loops, as follows:
$$
\alpha_0 \xrightarrow{t_1\mid t'_1} \alpha_1\xrightarrow{w_1\mid w'_1}
\alpha_1\xrightarrow{t_2\mid t'_2} \alpha_2\xrightarrow{w_2\mid w'_2} \alpha_2\dots
\alpha_{\ell-1}\xrightarrow{t_\ell\mid t'_\ell}\alpha_\ell\xrightarrow{w_\ell\mid w'_\ell} \alpha_\ell
\xrightarrow{t_{\ell+1}i_ps_{\ell+1}\mid t'_{\ell+1}} \beta
$$
such that $w_i\neq \epsilon$, $t_1w_1\dots t_\ell w_\ell t_{\ell+1} =
v_1i_1v_2i_2\dots v_p$, $\alpha_\ell
\xrightarrow{t_{\ell+1}i_ps_{\ell+1}\mid t'_{\ell+1}} \beta$ is a
single transition of $\tra'_\sync$ and the decomposition is done as follows:
After reading $t_1$, $\alpha_1$ is the first occurrence of a
state that repeats later on (but before reading $i_p$), and after reading $t_1w_1$, it is the
last occurrence of $\alpha_1$. Then, after reading $t_2$, $\alpha_2$ is
the first state occurring after $\alpha_1$ that repeats later on,
whose lost occurrence is after reading $w_2$, and so on.

Suppose that $\tra'_\sync$ has 
    $m$ states, and assume that $M_\sync$ is the largest length of an input word
    occurring on the transitions of $\tra'_\sync$. Then,
    $|t_1|+|t_2|+\dots +|t_{\ell}| \leq m.M_{\sync}$, otherwise
    there would be a repeating state occurring in the subwords
    $t_i$. Moreover for all $i\in\{1,\dots, \ell\}$, all
    $x\in\{\inp,\outp\}$, we have
    $|\pi_{x}(w_i)| = |\pi_{x}(w'_i)|$. Now, we decompose $u$ up to
    the input symbol $i_p$ according to the loops, $u_1i_1\dots u_p = t'_1w'_1\dots
    t'_\ell w'_\ell t'$ for some $t'\in\inouta^*$. Now, we have:
    $$
    \begin{array}{lllllllll}
    |\del(u_1\dots u_p,v_1\dots v_p)| & = & |u_1\dots u_p| - |v_1\dots
      v_p| \\
      & = & |t'_1w'_1\dots t'_\ell w'_\ell t'| - |t_1w_1\dots w_\ell
            t_{\ell+1}| \\
      & = & |t'_1t'_2\dots t'_\ell t'| - |t_1t_2\dots t_{\ell+1}| \\
      &\leq & |t'_1t'_2\dots t'_\ell t'|
    \end{array}
    $$
    Since for all $i\in\{1,\dots,\ell\}$, $|\pi_{\inp}(w_i)| =
    |\pi_{\inp}(w'_i)|$ and $|\pi_{\inp}(t'_1w'_1\dots t'_\ell w'_\ell
    t')| = |\pi_{\inp}(t_1w_1\dots w_\ell t_{\ell+1})|$, we also have
    $|\pi_{\inp}(t'_1t'_2\dots t'_\ell t')| = |\pi_{\inp}(t_1t_2\dots
    t_{\ell+1})|$. Since $\tra_1$ is real-time, we also have that 
    $|t'_1\dots t'_\ell t'| \leq M.|\pi_{\inp}(t'_1t'_2\dots t'_\ell
    t')|$. Finally, as $|t_1|+\dots + |t_\ell| \leq m.M_\sync$ and
    $|t_{\ell+1}|\leq M_\sync$ (as the transition from $\alpha_\ell$
    to $\beta$ is a single transition), we get $|\pi_{\inp}(t_1t_2\dots
    t_{\ell+1})|\leq (m+1).M_\sync$, and therefore 
    $|t'_1\dots t'_\ell t'|\leq (m+1).M.M_\sync$, i.e. 
    $|\del(u_1\dots u_p,v_1\dots v_p)|\leq (m+1).M.M_\sync$. It
    suffices to take $k = (m+1).M.M_\sync$ to conclude. 
\end{proof}

\section{Details for Section \ref{sec:fvalued}}

\subsection{Proof of Theorem \ref{lem:finitelyambiguous}}

We will need the following two lemmas that express properties about
delays. 

\begin{lemma}\label{lem:sizedelay}
    Let $u_1,u_2,u_3,v_1,v_2,v_3\in \Sigma^*$, then:
    $$
    |\del(u_1u_2u_3,v_1v_2v_3)|\leq |\del(u_1u_2,v_1v_2)|+|u_3v_3|
    $$
\end{lemma}
\begin{proof}
    Let $\beta,\alpha\in\Sigma^*$ such that $\beta^{-1}\alpha =
    v_2^{-1}v_1^{-1}u_1u_2$ and $\beta^{-1}\alpha$ is irreducible. 
    Therefore $|v_2^{-1}v_1^{-1}u_1u_2| = |\beta|+|\alpha|$. 
    Clearly, $v_3^{-1}v_2^{-1}v_1^{-1}u_1u_2u_3 =
    v_3^{-1}\beta^{-1}\alpha u_3$, and therefore 
    $|v_3^{-1}v_2^{-1}v_1^{-1}u_1u_2u_3| \leq
    |v_3|+|\beta|+|\alpha|+|u_3| = 
    |v_2^{-1}v_1^{-1}u_1u_2| + |u_3v_3|$.
\end{proof}

The following lemma is a folklore result that we prove for the sake of
completeness, as we use it intensively in this section. 

\begin{lemma}\label{lem:folkore}
    Let $v_1,w_1,v_2,w_2\in \Sigma^*$, then:
    $$
\begin{array}{c}
\del(v_1,w_1)\neq \del(v_1v_2,w_1w_2) \\ \implies \\ \forall 0\leq
  i<j,\ \del(v_1v_2^i,w_1w_2^i)\neq \del(v_1v_2^j,w_1w_2^j)
\end{array}
    $$
\end{lemma}

\begin{proof}
    First, note that $v_2$ and $w_2$ are not both equal to $\epsilon$, since $\del(v_1,w_1)\neq \del(v_1v_2,w_1w_2)$.
    Suppose that $v_1$ is not a prefix of $w_1$ and $w_1$ is not a
    prefix of $v_1$, i.e. $v_1 = u \alpha v'_1$ and $w_1 = u\beta
    w'_1$ for $u,v'_1,w'_1\in\Sigma^*$, $\alpha,\beta\in \Sigma$ and
    $\alpha \neq \beta$.

    Then for all $0\leq i<j$, the following two words in $(\Sigma\cup \overline{\Sigma})^*$
    $$
    \begin{array}{lllllllllllll}
    & \del(v_1v_2^i,w_1w_2^i) & = &
    v_2^{-i}v_1^{-1}w_1w_2^i & = & v_2^{-i} v'^{-1}_1\alpha^{-1}\beta
    w'_1w_2^i \\
    \text{and} & \del(v_1v_2^j,w_1w_2^j) & = &
    v_2^{-j}v_1^{-1}w_1w_2^j & = & v_2^{-j} v'^{-1}_1\alpha^{-1}\beta
    w'_1w_2^j \\
    \end{array}
    $$ 
    are both irreducible, and since they have different lengths
    ($i\neq j$), they are different.

    Assume now that $w_1 = v_1s$ for some $s\in\Sigma^*$ (the
    case where $w_1$ is a prefix of $v_1$ is symmetric and therefore
    untreated). Then $\del(v_1,w_1) = s$. Let $0\leq i<j$, we have:
    $$
    \begin{array}{lllllllllllll}
    & \del(v_1v_2^i,w_1w_2^i) & = &
    v_2^{-i}sw_2^i \\
    \text{and} & \del(v_1v_2^j,w_1w_2^j) & = &
    v_2^{-j}sw_2^j
    \end{array}
    $$ 
    Suppose that they are equal and let us derive a contradiction,
    i.e. suppose that  $v_2^{-i}sw_2^i = v_2^{-j}sw_2^j$. It is
    equivalent to $v_2^{j-i}s = s w_2^{j-i}$. Since $j>i$, we get
    $|w_2| = |v_2|$ and since $\del(v_1,w_1)\neq \del(v_1v_2,w_1w_2)$,
    we get that $v_2\neq \epsilon$ and $w_2\neq \epsilon$. 
    From the equality $v_2^{j-i}s = s w_2^{j-i}$ and the fact that
    $v_2^{j-i}$ and $w_2^{j-i}$ are nonempty, and $|v_2|=|w_2|$, a fundamental property
    of word conjugacy (see for instance 
    Proposition 1.3.4 of \cite{Lot83}) implies that there exist
    $p,q\in\Sigma^*$, $k>0$ and $\ell\geq 0$ such that 
    $$v_2 = (pq)^k\qquad w_2 =
    (qp)^k\qquad s = p(qp)^{\ell}$$ 
    Therefore, 
    $$
\begin{array}{llllll}
    \del(v_1v_2,w_1w_2) = v_2^{-1}s w_2 = (pq)^{-k} p(qp)^\ell (qp)^k
    & = p(qp)^{\ell} (pq)^{-k} (qp)^k \\ & = p(qp)^\ell = s = \del(v_1,w_1)
\end{array}
    $$
    from which we get a contradiction.
\end{proof}

We also need the following lemma that decomposes two words with a
large lag into smaller words with strictly increasing delays, the
number of which is greater than a value that depends on the lag.

\begin{lemma}\label{lem:delayincreasing}
    Let $w,w'\in\inouta$ such that $\pi_\inp(w) = \pi_\inp(w')$,
    i.e. $w$ and $w'$ can be decomposed into
$$w =  a_1s_1a_2s_2\dots a_ns_n\qquad w' = a_1t_1a_2t_2\dots a_nt_n$$
where $a_i\in \ina$ and $s_i,t_i\in \outa^*$.

For all $O,D\in \mathbb{N}$, if $|s_i|\leq O$, $|t_i|\leq O$ for all
$i\in\{1,\dots,n\}$, and if $\lag(w,w')\geq 2DO$, then there exist 
$\lambda > D$ and a sequence $1 = i_0\leq i_1<\dots < i_\lambda\leq
n$ such that for all $j\in\{0,\dots,\lambda-1\}$, we have 
$$
|\del(s_1\dots s_{i_j}, t_1\dots t_{i_j})|< |\del(s_1\dots s_{i_{j+1}}, t_1\dots t_{i_{j+1}})|
$$
\end{lemma}

\begin{proof}
Assume first that such a $\lambda$ exists. One defines the sequence 
$1 = i_0 < i_1 <\dots < i_{\lambda-1} < i_\lambda\leq n$ as follows:
For all  $0<j<\lambda$, 
$$
i_{j+1} = \text{min}\ \{ i\mid i_j < i \leq n, |\del(s_1\dots
s_{i_j}, t_1\dots t_{i_j})| < |\del(s_1\dots s_{i}, t_1\dots t_{i})|\}
$$
Let us now establish a lower bound for the value $\lambda$. Let $0\leq j< \lambda$. By
definition of the sequence, we have:
$$
\begin{array}{l}
|\del(s_1\dots
s_{i_{j+1}-1},t_1\dots t_{i_{j+1}-1})| \leq \\
\qquad\qquad  \qquad\qquad|\del(s_1\dots s_{i_j},t_1\dots t_{i_j})|  <  \\
\qquad\qquad\qquad\qquad\qquad\qquad\qquad\qquad|\del(s_1\dots
s_{i_{j+1}},t_1\dots t_{i_{j+1}})|
 \end{array}
$$
By Lemma \ref{lem:sizedelay} we get:
$$
|\del(s_1\dots s_{i_{j+1}},t_1\dots t_{i_{j+1}})| < |\del(s_1\dots
s_{i_{j+1}-1},t_1\dots t_{i_{j+1}-1})| + |s_{i_{j+1}}t_{i_{j+1}}|
$$
and therefore 
$$
\begin{array}{lllllll}
|\del(s_1\dots s_{i_{j+1}},t_1\dots t_{i_{j+1}})| & < & |\del(s_1\dots
s_{i_{j}},t_1\dots t_{i_{j}})| + |s_{i_{j+1}}t_{i_{j+1}}| \\
 & \leq & |\del(s_1\dots
s_{i_{j}},t_1\dots t_{i_{j}})| + 2O
\end{array}
$$

Therefore, in between any two successive indices $i_j$ and $i_{j+1}$,
the delay can increase of at most $2O$. Since $\lag(w,w') > k =
2OD$, by definition of lag, the maximal length of
the delay between prefixes of $w$ and $w'$ is at least
$2OD$. Therefore, the delay has to increase
of at least $D$ symbols to reach a delay of
length greater than $2OD$. Therefore, $\lambda> D$.
\end{proof}

We will also need the following lemma which states that for any finitely
ambiguous automaton $A$, long enough words $u$ can be decomposed into 
$u = u_1u_2u_3$ such that $u_2\neq \epsilon$ and iterating $u_2$ does
not increase or decrease the number of accepting runs of $A$. 
For $u\in\Sigma^*$, we denote by $\#_A(u)$ the number of accepting
runs of $A$ on $u$. We state the latter result more generally for
long enough concatenation of words. This formulation will ease the
proof of Theorem~\ref{lem:finitelyambiguous}.

\begin{lemma}\label{lem:nomoreruns}
    Let $A$ be an $m$-ambiguous automaton with $n$ states. Let $N>
    n^m2^n$ and $u_1,\dots,u_N\in \Sigma^+$. Then, there exist two
    integers $\ell_1\leq \ell_2$ in $\{1,\dots,N\}$ such that 
    $$\#_A(u_1\dots u_N) = \#_A(u_1\dots u_{\ell_1-1}(u_{\ell_1}\dots
    u_{\ell_2})^iu_{\ell_2+1}\dots  u_N) \text{ for all } i\geq 1$$

    Moreover for $k = \#_A(u_1\dots u_N)$, there exists $k$ states 
    $q_1,\dots,q_k$, $k$ initial states $p_1,\dots,p_k$ and $k$ final
    states $s_1,\dots,s_k$ such that for all $i\geq 0$, we have:
    $$
    \begin{array}{lllllllllllllllllll}
      p_1 & \xrightarrow{u_1\dots u_{\ell_1-1}} & q_1 & \xrightarrow{(u_{\ell_1}\dots
    u_{\ell_2})^i} & q_1 & \xrightarrow{u_{\ell_2+1}\dots u_N} & s_1
      \\
\dots\\
      p_k & \xrightarrow{u_1\dots u_{\ell_1-1}} & q_k & \xrightarrow{(u_{\ell_1}\dots
    u_{\ell_2})^i} & q_k & \xrightarrow{u_{\ell_2+1}\dots u_N} & s_k
      \\
    \end{array}
    $$
    and any accepting run on $u_1\dots u_{\ell_1-1}(u_{\ell_1}\dots
    u_{\ell_2})^iu_{\ell_2+1}\dots  u_N$ has one of the above form. 
\end{lemma}

\begin{proof}
    Let $u = u_1\dots u_N$ and $k = \#_A(u)$. Let us denote by $r_1,\dots,r_k$ all the
    (pairwise different) accepting runs of $A$ on
    $u$. For all $i\in\{0,\dots,N\}$ and all runs $r_j$,
    denote by $r_j[i]$ the state of $r_j$ after
    reading the prefix $u_1\dots u_{i}$ (where $u_1\dots u_i = \epsilon$ if $i=0$), and by
    $R_i$ the set of states $q$ such that there exists a run of $A$
    from some initial state to $q$ on the prefix $u_1\dots u_i$.

    Since $N>n^m 2^n$ and $m\geq k$, there exist two integers
    $i_1<i_2$ in $\{0,\dots,N\}$ such that for all $j\in\{1,\dots,k\}$, 
    $r_j[i_1] = r_j[i_2]$ and $R_{i_1} = R_{i_2}$. In other words, the
    two positions $i_1$ and $i_2$ form a cycle in each of the runs $r_j$, 
    and at these two positions, the set of states reached by $A$ 
    are the same. We finally decompose $u$ into
    $$
    v_1 = u_1\dots u_{i_1}\quad v_2 = u_{i_1+1}\dots u_{i_2}\quad
    v_3 = u_{i_2+1}\dots u_N
$$ 
    % $v_2 = u_1\dots u_{i_1}$ u[i_1{+}1{:} i_2{-}1]$ and $v_3 = u[i_2{:}]$ (if $i_2 =
    % i_1+1$, then $u[i_1{+}1{:} i_2{-}1]=\epsilon$). 

    We now show that $\#_A(v_1v_2v_3) = \#_A(v_1v_2v_2v_3)$ (which
    easily generalises to $\#_A(v_1v_2v_3) = \#_A(v_1(v_2)^iv_3)$ for
    all $i\geq 1$, and gives the lemma by taking $\ell_1 = i_1+1$ and
    $\ell_2 = i_2)$.

    There at least $k$ accepting runs on $v_1v_2v_2v_3$ obtained by
    iterating the loops on $v_2$ of the $k$ accepting runs on
    $v_1v_2v_3$. Therefore $\#_A(v_1v_2v_3) \leq
    \#_A(v_1v_2v_2v_3)$. 

    To show that $\#_A(v_1v_2v_3) \geq \#_A(v_1v_2v_2v_3)$, take an accepting run $\rho$ on $v_1v_2v_2v_3$. Let $r_1,r_2,r_3$ be the states
    of $\rho$ after reading $v_1$, $v_1v_2$ and $v_1v_2v_2$
    respectively. 
    We decompose $\rho$ into $\rho_1\rho_2\rho'_2\rho_3$ where $\rho_1$ is
    the part of $\rho$ on $v_1$, $\rho_2$ on $v_2$ from $r_1$ to $r_2$,
    $\rho'_2$ on $v_2$ from $r_2$ to $r_3$ and $\rho_3$ on $v_3$ from
    $r_3$. 
    We show that $r_1 = r_2 = r_3$.
    As $A$ is finitely ambiguous, there exists at most one loop over $r_2$ on input $v_2$, 
    hence this implies that $\rho_2 = \rho_2'$, which means that $\rho$
    is obtained by iterating once the loop on $v_2$ of a run on
    $v_1v_2v_3$, and that $\#_A(v_1v_2v_3) \geq
    \#_A(v_1v_2v_2v_3)$. Note that the sets of states reached by $A$ after
    reading $v_1$, $v_1v_2$ and $v_1v_2v_2$ are all equal to 
    $R := R_{i_1} = R_{i_2}$, and therefore $r_1,r_2,r_3\in R$. 
    Therefore, there exists a run $\rho'_1$ of $A$ on $v_1$
    from an initial state to $r_2$, and hence, the run
    $\rho'_1\rho'_2\rho_3$ is accepting on $v_1v_2v_3$, which implies
    that $r_2=r_3$ since on $v_1v_2v_3$, all the accepting runs loop
    on $v_2$ by assumption. From $r_2=r_3$, we get that 
    $\rho_1\rho_2\rho_3$ is a run of $A$ on $v_1v_2v_3$, and therefore
    $r_1 = r_2$ for the same reason as before. 
\end{proof}

We now proceed to the proof of Theorem~\ref{lem:finitelyambiguous},
restated below. 

\thmfinitelyambiguous*

\begin{proof}
Let $M = max \{ |v|\mid \exists (q,u,v,q')\in\Delta_1\cup \Delta_2, \text{
  or } v = f_j(q_j) \text{ for some } q_j\in Q_j\}$. We let 
$$k = 4M |Q_1|(
(|\Delta_2||Q_2|)^m2^{|Q_2||\Delta_2|}+1)$$
As explained in the sketch of proof, a \emph{witness} is a word $w \in \Sigma^*$ such that for all words
    $w'\in \lang{\tra_2}$ equivalent to $w$, $\lag(w,w')>k$, and the
    following claim implies a contradiction if we assume that
    $\tra_1\not\subseteq_k \tra_2$:
\begin{claim}
For all witnesses $w\in\lang{\tra_1}$, there exists a witness
      $t\in\lang{\tra_1}$ such that $N_{\tra_2}(t)<N_{\tra_2}(w)$.
\end{claim}

\noindent\textbf{Proof overview of the claim} Let $A_1,A_2$ be the underlying automata
of $\tra_1$ and $\tra_2$ respectively. Note that since $\tra_1$ and $\tra_2$ are
real-time, the transitions of $A_1$ and $A_2$ are labeled by words in 
$\ina(\outa)^*$. Wlog we assume that $\tra_2$ satisfies the following property: 
for any two transitions $(p,a,w_1,q_1), (p,a,w_2,q_2)$ of $\tra_2$, if
$w_1\neq w_2$, then $q_1\neq q_2$. This can be enforced by taking the
product of the states of $\tra_2$ with its set of transitions (this is
why we have the factor $|\Delta_2|$ in the constant $k$). With
such an assumption, for all words $u\in\Sigma^*$, there is a one-to-one correspondence between the
accepting runs of $\tra_2$ on $u$ and that of its input automaton
$I_2$ on $u$. Therefore, $I_2$ is $m$-ambiguous.

Now, suppose that $w$ is a witness, and let $u =
\pi_\inp(w)$ the input of $w$, $v = \pi_\outp(w)$ the output of $w$. 
We consider the set of accepting runs $R = \{\rho_1,\dots,\rho_\alpha\}$ of $A_2$ on the words $w'$ such
that $\pi_\inp(w') = u$, and an accepting run $\rho$ of $A_1$ on $w$. 
We exhibit a decomposition of the runs $\rho_i$ and the run $\rho$ such that all
these runs loop synchronously on the same input factor, and some run
$\rho_\beta$ has the same output as $\rho$, but a different output
delay with $\rho$ before and after the loop. The idea is that by
iterating this loop a sufficient number of times $\ell$, we know,
based on a folklore result about transducer delays, that it will generate a
delay so long that it will not be recovered by the output suffixes of
$\rho$ and $\rho_\beta$ after the loop. Therefore, the iterations of
$\rho$ and $\rho_\beta$ $\ell$ times, denoted by $\rho^{(\ell)}$ and
$\rho_\beta^{(\ell)}$, will have different outputs. We choose this
loop according to Lemma \ref{lem:nomoreruns} applied on $I_2$ to
make sure that we do not generate accepting runs that are not
iterated versions of the accepting runs of $A_2$.

The word $w_\ell$ accepted by $\rho^{(\ell)}$ is a good candidate to be
a new witness, because we have ruled out one run $\rho_\beta$ of $A_2$
whose accepted word was equivalent to $w$, in the sense that
$\rho^{(\ell)}$ and $\rho_\beta^{(\ell)}$ does not accept equivalent
words anymore. However, $w_\ell$ may now be equivalent to the word
accepted by the iteration of another accepting run $\rho_{\gamma}$ in
$R$, while $\rho_\gamma$ was accepting a word non-equivalent to
$w$. We show that by taking $\ell$ large enough, this cannot happen, i.e., necessarily $\rho_\gamma$ and $\rho$
were accepting equivalent words and their delay after and before the
loop were equal (otherwise we could have made their iterated output
different by taking a sufficiently large $\ell$). If they have the
same delay before and after the loop, then iterating the loop does not
change the overall lag. Therefore, if by iterating the loop a
sufficient number of times $\ell$, the word $w'_\ell$ accepted by some iterated run
$\rho_\gamma^{(\ell)}$ is equivalent to $w_\ell$, then
$\lag(w_\ell,w'_\ell) > k$. 

\vspace{2mm}
\noindent \textbf{Detailed proof of the claim} Let $u=\pi_\inp(w)$, $v =
\pi_\outp(w)$, $\rho$ be an accepting run of $A_1$ on $w$, and let
$R =\{\rho_1,\dots,\rho_\alpha\}$ the set of accepting runs of $A_2$ on
words $w'$ such that $\pi_\inp(w) = \pi_\inp(w')$. For a run $\rho_i$,
we denote by $w_i$ the word it accepts. The set $R$ can be
partitioned into two sets $R_=$ and $R_{\neq}$ depending on whether
$\pi_\outp(w_i) = v$ or $\pi_\outp(w_i)\neq v$. Since $\tra_1\subseteq
\tra_2$, we have $R_=\neq \varnothing$. Wlog, assume
that $\rho_1 \in R_=$. By assumption we have $\lag(w,w_1)>k$.

Since $w$ and $w_1$
are equivalent and $\tra_1, \tra_2$ are real-time, they can be decomposed
into: 
$$
\begin{array}{llllllll}
  w & = & a_1s_1a_2s_2\dots a_ns_n & \qquad & w_1 & = & a_1t_1a_2t_2\dots a_nt_n
\end{array}
$$
such that $a_i\in\ina$, $s_i,t_i\in(\outa)^*$ and $s_1\dots s_n =
t_1\dots t_n$. 

By Lemma \ref{lem:delayincreasing} (applied\footnote{It is not $2M$ because we
have to consider the case where $s_n$ or $t_n$ could be the
concatenation of  a word on a transition and a word produced by the
output functions} with $D = k =
|Q_1|(|Q_2|^m 2^{|Q_2|}+1)$ and $O = 2M$), there exist $\lambda > k$ and a sequence 
$1 = i_0 < i_1 <\dots < i_{\lambda-1} < i_\lambda\leq n$
such that for all $0<j<\lambda$, 
$$
i_{j+1} = \text{min}\ \{ i\mid i_j < i \leq n, |\del(s_1\dots
s_{i_j}, t_1\dots t_{i_j})| < |\del(s_1\dots s_{i}, t_1\dots t_{i})|\}
$$

Now, consider the run $\rho$ of $A_1$ on $w$. Since $\lambda> |Q_1|(|Q_2|^m
2^{|Q_2|}+1)$, there is a state that repeats $N+1$ times, with $N >
|Q_2|^m 2^{|Q_2|}$, on reading the prefixes $a_1s_1\dots a_{i_j} s_{i_j}$ for all
$j\in\{ 0,\dots,\lambda\}$. Formally, there exists a subsequence $j_0,\dots,
j_N\in\{i_0,\dots,i_\lambda\}$ such that $j_0<j_1<\dots<j_N$ and a state
$q\in Q_1$ such that after reading $a_1s_1\dots a_{j_\ell}
s_{j_\ell}$, for $\ell\in\{0,\dots,N\}$, $\rho$ is in state $q$.

We decompose the input word $u$ into $u_1\dots u_N$ according to this latter sequence of
indices: For all $\ell\in\{1,\dots,N\}$, let $u_\ell = a_{j_{\ell-1}}\dots a_{j_\ell}$. Since
$N> |Q_2|^m 2^{|Q_2|}$, we can now apply Lemma
\ref{lem:nomoreruns} on the input automaton of $\tra_2$
and get a decomposition of $u$ into $v_1v_2v_3$ such that, when
iterating $v_2$, the number of accepting runs of $\tra_2$ on the
iterated input stays the same. Let $v_1 = a_1\dots a_{i}$, $v_2 =
a_{i+1}\dots a_j$ and $v_3 = a_{j+1}\dots a_n$ for some $i< j$. By the choice of our
decomposition and Lemma
\ref{lem:nomoreruns}, we obtain that for all $\ell\geq 0$, there are
exactly $\alpha$ accepting runs
$\rho_1^{(\ell)},\dots,\rho_\alpha^{(\ell)}$ of $A_2$ on words $u'$
such that $\pi_\inp(u') = v_1v_2^\ell v_3$, and accepting run of $A_1$
on $a_1s_1\dots a_{i}s_{i}(a_{i+1}s_{i+1}\dots a_js_j)^\ell
a_{j+1}s_{j+1}\dots a_ns_n$ which have the following form:
    $$
    \begin{array}{lclclclcllllllllllllll}
      \rho_\alpha^{(\ell)} & : & p_{\alpha} &
                                              \xrightarrow{a_1s_{\alpha1}\dots
                                              a_{i}s_{\alpha i}}
      & q_{\alpha} & \xrightarrow{(a_{i+1}s_{\alpha (i+1)}\dots
    a_js_{\alpha j})^\ell} & q_\alpha & \xrightarrow{a_{j+1}s_{\alpha (j+1)}\dots a_ns_n}
      & r_\alpha \\
 & & \vdots & & \vdots & & \vdots  & & \vdots \\
      \rho_2^{(\ell)} & : & p_2 & \xrightarrow{a_1s_{21}\dots a_{i}s_{2i}} & q_2 & \xrightarrow{(a_{i+1}s_{2(i+1)}\dots
    a_js_{2j})^\ell} & q_2 & \xrightarrow{a_{j+1}s_{2(j+1)}\dots a_ns_n} & r_2
      \\
      \rho_1^{(\ell)} & : & p_1 & \xrightarrow{a_1t_1\dots a_{i}t_{i}} & q_1 & \xrightarrow{(a_{i+1}t_{i+1}\dots
    a_jt_j)^\ell} & q_1 & \xrightarrow{a_{j+1}t_{j+1}\dots a_nt_n} & r_1
      \\

     \rho^{(\ell)} & : &  p & \xrightarrow{a_1s_1\dots a_{i}s_{i}} & q & \xrightarrow{(a_{i+1}s_{i+1}\dots
    a_js_j)^\ell} & q & \xrightarrow{a_{j+1}s_{j+1}\dots a_ns_n} & r
      \\
    \end{array}
    $$
    Moreover, our decomposition  guarantees that
    $|\del(s_1\dots s_i,t_1\dots t_i)|< |\del(s_1\dots s_j,t_1\dots
    t_j)|$ and so $\del(s_1\dots s_i,t_1\dots t_i)\neq \del(s_1\dots s_j,t_1\dots
    t_j)$.

    Now that we have obtained such a decomposition, we can proceed to
    the end of the proof, based on 
    Lemma~\ref{lem:folkore}. This lemma and the fact that 
    $$
    \del(s_1\dots s_i,t_1\dots t_i)\neq \del(s_1\dots s_j,t_1\dots
    t_j)$$
    implies that for a sufficiently large value of $\ell$, the
    delay $$\del(s_1\dots s_i(s_{i+1}\dots s_j)^\ell,t_1\dots
    t_i(t_{i+1}\dots t_j)^\ell)$$ is so long that it cannot be
    recovered with the suffixes $t_{j+1}\dots t_n$ nor $s_{j+1}\dots
    s_n$. And therefore, for some value $\ell_0$ we have:
    $$
    \forall \ell\geq \ell_0,\ s_1\dots s_i(s_{i+1}\dots s_j)^{\ell} s_{j+1}\dots s_n\neq
    t_1\dots t_i (t_{i+1}\dots t_j)^{\ell} t_{j+1}\dots t_n
    $$
    Remind that we had the equality for $\ell = 1$. 

    However it is not sufficient to conclude the proof, because it
    could be the case that
    the output of some run $\rho_\beta$ was different from that of
    $\rho$ but the output of $\rho^{(\ell)}$ is equal to that of
  $\rho_{\beta}^{(\ell)}$ for all $\ell\geq \ell_0$. Suppose it is the case, that is for some $\beta\in\{2,\dots,\alpha\}$ we
  have, for all $\ell\geq \ell_0$:
  $$
  \begin{array}{rcl}
    s_1\dots s_n & \neq & s_{\beta 1}\dots s_{\beta n} \\
    s_1\dots s_i(s_{i+1}\dots s_j)^{\ell}s_{j+1}\dots s_n & = &
                                                                s_{\beta
                                                                1}\dots
                                                                s_{\beta
                                                                i}
                                                                (s_{\beta
                                                                (i+1)}
                                                                \dots
                                                                s_{\beta
                                                                j})^\ell
                                                                s_{\beta
                                                                (j+1)}\dots
                                                                s_{\beta n}
  \end{array}
  $$
Then necessarily $\del(s_1\dots s_i,s_{\beta 1} \dots s_{\beta i})\neq 
\del(s_1\dots s_j,s_{\beta 1} \dots s_{\beta j})$, otherwise by iterating
the loop, the difference  $s_1\dots s_n  \neq  s_{\beta 1}\dots s_{\beta n}$
would be preserved. By applying the same
argument as before (Lemma~\ref{lem:folkore}), one can take some $\ell_\beta$ large enough so
that for all $\ell\geq \ell_\beta$, $$\del(s_1\dots s_i(s_{i+1}\dots s_j)^{\ell},
                                                                s_{\beta
                                                                1}\dots
                                                                s_{\beta
                                                                i}
                                                                (s_{\beta
                                                                (i+1)}
                                                                \dots
                                                                s_{\beta
                                                                j})^\ell)$$
                                                              is so
                                                              large
                                                              that it
                                                              cannot
                                                              be
                                                              recovered,
                                                              and
                                                              therefore
                                                              for all
                                                              $\ell\geq
                                                              \ell_\beta$,
                                                              we have
                                                              necessarily 
    $s_1\dots s_i(s_{i+1}\dots s_j)^{\ell}s_{j+1}\dots s_n \neq
                                                                s_{\beta
                                                                1}\dots
                                                                s_{\beta
                                                                i}
                                                                (s_{\beta
                                                                (i+1)}
                                                                \dots
                                                                s_{\beta
                                                                j})^\ell
                                                                s_{\beta
                                                                (j+1)}\dots
                                                                s_{\beta n}$.

We let $\ell_*$ be the maximal value among $\ell_0$ and $\ell_\beta$
for all such $\beta$. By the choice of $\ell_*$, we obtain the following property: for all
$\ell\geq \ell_*$, 
$$
N_{T_2}(a_1s_1\dots a_is_i(a_{i+1}s_{i+1}\dots a_js_j)^{\ell} a_{j+1} s_{j+1}\dots a_n
s_n < N_{T_2}(w).
$$
Yet it is not sufficient to obtain a new witness $t$ because the
definition of witness requires that $\lag(t,t')>k$ for
all words $t'\in\lang{T_2}$ equivalent to $t$. We finally solve this
problem. Consider some $\beta\in\{2,\dots,n\}$ and $\ell\geq \ell_*$
such that the output $o_\beta$ of $\rho^{(\ell)}$ is equal to the
output $o$ of $\rho_\beta^{(\ell)}$. Then, consider the following two cases:
\begin{enumerate}
  \item $\del(s_1\dots s_i, s_{\beta 1}\dots s_{\beta i}) =
    \del(s_1\dots s_j, s_{\beta 1}\dots s_{\beta j})$
  \item $\del(s_1\dots s_i, s_{\beta 1}\dots s_{\beta i}) \neq
    \del(s_1\dots s_j, s_{\beta 1}\dots s_{\beta j})$
\end{enumerate}
If we are in case $(1)$, then necessarily the output of
$\rho_\alpha$ and $\rho$ were equal, and iteration preserves the
lag. Therefore $\lag(o,o_\beta) > k$. In case $(2)$, it suffices, with
the same arguments as before, to take $\ell$ large enough (larger than
some $\ell_\beta$) so that
necessarily $o$ and $o_\beta$ are different, thus creating a
contradiction. 

Finally, we take $z$ as the maximal value between $\ell_*$ and all the
values $\ell_\beta$ just defined. We let 
$$
t = a_1s_1\dots a_is_i(a_{i+1}s_{i+1}\dots a_js_j)^z a_{j+1}
s_{j+1}\dots a_n s_n
$$
By construction, $t$ is a new witness such that $t\in \lang{T_1}$ and
$N_{T_2}(t) < N_{T_2}(w)$. 
\end{proof}

%%%%%%%%%%%%%%%%%%%%%%%%%%%%%%%%%%%%%%%%%
\subsection{Sequential uniformisation of finite-valued transducers}
%%%%%%%%%%%%%%%%%%%%%%%%%%%%%%%%%%%%%%%%%

The decidability of the sequential uniformisation problem
(Corollary~\ref{fv_unif}) was based on the key statement,
Theorem~\ref{f_a_det}, which reduces, for any finitely ambiguous
transducer $T$, the sequential uniformisation problem to a sequential
$\newbound$-uniformisation problem. The goal of this section is to prove
this theorem. That is, we want to prove the following theorem:

\thmunifvalued*

%%%%%%%%%%%%%%%%%%%%%%%%%%%%%%%%%%%%%%%%%
\paragraph*{Important notations and assumptions}
%%%%%%%%%%%%%%%%%%%%%%%%%%%%%%%%%%%%%%%%%

\begin{itemize}
\item \textbf{Original transducer $T$:} In all this section, $\tra =
(Q_{\tra},I_{\tra},F_{\tra},\Delta_{\tra},f_{\tra})$ denotes a trim
real-time transducer given as a finite union of unambiguous
transducers, assumed to be sequentially uniformisable. 

\item\textbf{Sequential $k$-uniformiser $U$:} As explained in the sketch of proof from Section~\ref{sec:fvalued}, if
$T$ is sequentially uniformisable, there exists an integer that we denote by $k$ in all this
section such that $\tra$ is sequentially
$k$-uniformisable by some sequential transducer that we denote by $U$.

\item\textbf{Infinite Sequential $\newbound$-uniformiser $U'$:} The goal of this
section is to construct a sequential $\newbound$-uniformiser for $T$ with an
infinite number of states, denoted by $U' = (Q', I', F', \Delta',
f')$. We will see (Lemma~\ref{lem:mainlem}) that it implies the existence of a
(finite) sequential $\newbound$-uniformiser for $T$. The value $\newbound$ will be defined later.

\item \textbf{Maximal output length $m_\tra$:} We let $m_{\tra}$ denote the maximal
  length of the output words labelling the transitions of $\tra$.

\item \textbf{Runs of $T$ and $U$} For every state $q\in Q_T$ and word
  $w\in\Sigma^*$, as $\tra$ is trim and a (disjoint) union of
  unambiguous transducers, if there exists a run of $\tra$ on
  $v$ from an initial state to $q$, this run is unique and we denote
  it by $r_{\tra}(w,q)$.

  Similarly, for any sequential transducer $V$, we denote by $r_{V}(w)$ the unique run from $V$ on $w$
  (if it exists). Note that we do not need a target state here because
  $V$ is sequential.

\item \textbf{Lags and delays:} For all $p,q\in Q_T$ and
  $v\in\Sigma^*$, we denote by $\del_{p,q}(v)$ the delay between the
  outputs of the runs $r_{\tra}(v,p)$ and $r_{\tra}(v,q)$, if they
  exist, and by $\lag_{p,q}(v)$ we denote their lag. For any two
  transducers $T_1,T_2$, and runs
  $r_1$ and $r_2$ on the same input $u$, we will denote
  by $\lag(r_1,r_2)$ the lag between the words in $\inouta$ accepted
  by $r_1$ and $r_2$ respectively, assuming that they project on the
  same input. Similarly, we also define $\del(r_1,r_2)$.

\item \textbf{Choice function:} The transducer
  $\uni'$ will be constructed in such a way that it filters out runs 
  of $T$, on the same input, that are far (in terms of lag) to the
  run of $U$. This is formalised through a notion of choice function. 
Formally, for all $d\in\mathbb{N}$, we denote the choice function,
parameterized by $d$, by 
$\choice{d}{\uni} : \Sigma^* \rightarrow 2^{Q_{\tra}}$. It maps each word $v$ to the set
\[
\choice{d}{\uni}(v) = \{ q \in Q_{\tra} | \textup{the run $r_{\tra}(v,q)$ exists, and $\lag(r_{\tra}(v,q),r_{\uni}(v)) \leq d$}\}.
\]
\end{itemize}

\paragraph*{Structure of the proof} In \ref{subsec:expan2unif} we show how to construct
from the sequential uniformiser $U$, a sequential $\newbound$-uniformiser
of $T$, assuming the existence of a function $\rho:\Sigma^*\rightarrow
\Sigma^*$ which satisfies three properties $\prop{1}-\prop{3}$ given
in \ref{subsec:expan2unif}. Then the rest of the section is devoted to
the proof of the existence of such a function $\rho$. In
\ref{subsec:transmod}, we introduce the notion of transition monoid
for transducers, and prove technical lemmas about the structural
properties of idempotent elements of transitions monoids, and
useful properties of delays. In \ref{subsec:defrho}, we proceed with
the definition of $\rho$. Finally, in \ref{subsec:p1}, \ref{subsec:p2}
and \ref{subsec:p3}, we prove that the previously defined function
$\rho$ satisfies the properties $\prop{1}$ to $\prop{3}$
respectively.

%%%%%%%%%%%%%%%%%%%%%%%%%%%%%%%%%%%%%%%%%
\subsubsection{Towards the construction of a bounded delay sequential uniformiser}\label{subsec:expan2unif}
%%%%%%%%%%%%%%%%%%%%%%%%%%%%%%%%%%%%%%%%%

The construction of the sequential $\newbound$-uniformiser $U'$ for $T$
relies the existence of a function $\rho : \Sigma^* \rightarrow
\Sigma^*$ which satisfies the following properties, for every word $w \in \Sigma^*$:

\begin{description}
\item[\prop{1}]
for every state $q$, $r_{\tra}(w,q)$ exists if and only if $r_{\tra}(\rho(w),q)$ does;
\item[\prop{2}]
for every pair $p,q \in \choice{2k}{\uni}(\rho(w))$, $\lag_{p,q}(w) \leq \oldbound$;
\item[\prop{3}]
for every $a \in \Sigma$, if $wa$ is a prefix of a word of $\dom(\tra)$, there exist $q \in \choice{k}{\uni}(\rho(wa))$, $p \in \choice{k}{\uni}(\rho(w))$ and $s \in \Sigma^*$ such that $(p,a,s,q) \in \Delta_{\tra}$.
\end{description}

The following lemma states that if such a function exist, then one can
construct a sequential $\newbound$-uniformiser for $T$. 

\begin{lemma}\label{lem:mainlem}
    If there exists a function $\rho$ (with an integer $\oldbound$) which satisfies properties
    $\prop{1}-\prop{3}$, then one can construct a 
    sequential $\newbound$-uniformiser of $T$, where $\newbound = 2 |Q_{\tra}| \oldbound$. 
\end{lemma}

\begin{proof}
The proof is based on the construction of an infinite sequential
$\newbound$-uniformiser $U'$ of $T$, i.e. a sequential transducer
with an infinite number of states. First, we show that it implies the
existence of a (finite) sequential $\newbound$-uniformiser. Then, we proceed
to the construction of $U'$. 

\vspace{2mm}
\noindent \textbf{From infinite uniformisers to finite uniformisers} Using the function $\rho$, we will build an infinite seq-$\newbound$-uniformiser $U'$ of $\tra$, i.e., a seq-$\newbound$-uniformiser with an infinite set of states.
Let us prove that the existence of this infinite uniformiser $U'$ implies the seq-$\newbound$-uniformisability of $\tra$.
Let $\sync$ be the $\newbound$-delay resynchroniser.
By Nivat's theorem (Theorem~\ref{thm:nivat}), we know that there exists a transducer 
$\tra^{\sync}$ such that 
$\lang{\tra^{\sync}} = \sync(\lang{\tra})$.
Then $U'$ is an infinite seq-$\mathbb{I}$-uniformiser of $\tra^{\sync}$.
Such a uniformiser corresponds to a winning strategy for Player \textsf{Out} in the safety game 
$G_{\tra^{\sync}}$ defined in the proof of Proposition~\ref{prop:Iinclusion}.
However, we saw that, as safety games are memoryless determined, $\tra^{\sync}$ also admits a (finite) seq-$\mathbb{I}$-uniformiser $U''$.
Then $U''$ is a (finite) seq-$\newbound$-uniformiser of $\tra$, which proves that $\tra$ is sequentially $\newbound$-uniformisable.

\vspace{3mm}
\noindent \textbf{Construction of $\uni'$} Let us now define the infinite seq-$\newbound$-uniformiser $\uni'$.
The labeling of its transitions and its accepting states is based on
a partial function $\outru{U'}:\Sigma^*\rightarrow \Sigma^*$ and
a total function $\outf{\uni'}:\dom(T)\rightarrow \Sigma^*$ such that,
for all $w\in\Sigma^*$, $a\in\Sigma$, if $wa\in \dom(\outru{U'})$,
then $w\in \dom(\outru{U'})$ and $\outru{\uni'}(w)\preceq\outru{\uni'}(wa)$, and such that if $w\in \dom(T)$, then $w\in
\dom(\outru{\uni'})$ and $\outru{\uni'}(w)\preceq \outf{\uni'}(w)$. They
will be defined just after the following definition of $\uni'$:
\begin{itemize}
\item
$Q' = \Sigma^*$;
\item
$I' = \{ \epsilon \}$;
\item
$F' = \dom(\tra)$;
\item
$\Delta' = \{(w,a,\outru{\uni'}(w)^{-1}\outru{\uni'}(wa),wa) | w \in \Sigma^*, a \in \Sigma, \textup{  $wa$ is a prefix of $\dom(\tra)$} \}$;
\item
$f' : F \rightarrow \Sigma^*, v \mapsto \outru{\uni'}(v)^{-1}\outf{\uni'}(v)$.
\end{itemize}

% \vspace{3mm}
% \noindent \textbf{Output functions} 

% The output labeling the transitions
% of $U'$ will be based on function $\outru{\uni'}:\Sigma^*\rightarrow
% \Sigma^*$ and a function $\outf{\uni'} : \dom(T)\rightarrow
% \Sigma^*$. For all sequential uniformisers $V$ of $\tra$, $\outru{V}$
% maps $v \in \Sigma^*$ to the output of the run $r_{V}(v)$ (if it
% exists, otherwise the function is undefined), and
% $\outf{V}$ maps $v \in \dom(\tra)$ to $V(v)$ (if it is defined, otherwise
% the function is undefined). Note that those functions are consistent
% with the prefix relation: given $w \in \Sigma^*$ and $a \in \Sigma$,
% $\outru{V}(w)$ is a prefix of $\outru{V}(wa)$ (because $V$ is
% sequential), and given $v \in \dom(\tra)$, $\outru{V}(v)$ is a prefix
% of $\outf{V}(v)$ (if all these values are defined).
% In order to build the infinite seq-$N_{\tra}$-uniformiser $\uni'$, we
% first define the output function we want it to satisfy.

Note that in the above definition, the behaviour of $U'$ almost only
depends on the definitions of $\outru{\uni'}$ and $\outf{\uni'}$. We
now proceed to their definition. 
We first define an auxiliary function $\heap{}{} : \Sigma^+ \rightarrow 2^{Q_{\tra}}$.

\vspace{1mm}
\noindent \textbf{Definition of $\heap{}{}$}
Let $v \in \Sigma^*$ be a prefix of a word of $\dom(\tra)$.
If $v = \epsilon$, let $q_{v}$ be any initial state.
If $v = wa$, where $w \in \Sigma^*$ and $a \in \Sigma$, by \prop{3}, there exist $q_{v} \in \choice{k}{\uni}(\rho(wa))$, $p_{w} \in \choice{k}{\uni}(\rho(w))$ and $s \in \Sigma^*$ such that $(p,a,s,q) \in \Delta_{\tra}$.

The function $\heap{}{}$ maps $v$ to the union of the sets $\heap{i}{v}$, $i \in \mathbb{N}$, defined as follows.
\begin{itemize}
\item
$\heap{0}{v} = \{ q_v \}$;
\item
$\heap{i+1}{v} = \{ q \in Q_{\tra} | r_{\tra}(v,q) \textup{ exists, and } \exists q' \in \heap{i}{v} \textup{ s.t. } \lag_{q,q'}(v) \leq \oldbound \}$. 
\end{itemize}
Note that for every state $q$, as long as $r_{\tra}(v,q)$ exists, $\lag_{q,q}(v) = 0$, hence $\heap{i}{v} \subseteq \heap{i+1}{v}$.
More generally, given a pair $i<j$ of integers, $\heap{i}{v} \subseteq \heap{j}{v}$.
Moreover, if $\heap{d+1}{v} = \heap{d}{v}$ for some integer $d$, then $\heap{n}{v} = \heap{d}{v}$ for every $n \geq d$.
As the sets $\heap{i}{v}$, $1 \leq i \leq |Q_{\tra}|$, are strictly increasing subsets of $Q_{\tra}$, there exists such a $d$ between $1$ and $|Q_{\tra}|$.
Therefore $\heap{|Q_{\tra}|}{v} = \bigcup_{i \in \mathbb{N}}\heap{i}{v}$.
We define $\heap{}{} : \Sigma^+ \rightarrow 2^{Q_{\tra}}$ as the function mapping $v$ to $\heap{|Q_{\tra}|}{v}$.
Let us now show that it satisfies the following properties.
\begin{description}
\item[\propH{1}]
$\choice{k}{\uni}(\rho(v)) \subseteq \heap{1}{v}$;
\item[\propH{2}]
for every pair $q,q' \in \heapf{v}$, $\lag_{q,q'}(v) \leq \newbound = 2|Q_{\tra}|\oldbound$;
\item[\propH{3}]
for every state $q \in \heapf{wa}$, there exists $p \in \heapf{w}$ and $s \in \Sigma^*$ such that $(p,a,s,q) \in \Delta_{\tra}$.
\end{description}

\begin{enumerate}
\item
By definition of $\heap{0}{v} = \{ q_v \}$, $q_v \in \choice{k}{\uni}(\rho(v))$.
Hence by \prop{2}, for every $q \in \choice{k}{\uni}(\rho(v))$, $\lag_{q_v,q}(v) \leq \oldbound$, and therefore $q \in \heap{1}{v}$;
\item
Let $n = |Q_{\tra}|$.
For every $q,q' \in \heapf{v} = \heap{n}{v}$, there exist two sequences of states $q_0,q_1,\ldots,q_{n}$ and $q'_0,q'_1,\ldots,q'_{n}$ such that $q_{n} = q$, $q'_{n} = q'$, $q_0  = q_v = q_0'$, and for every $0 \leq i \leq n-1$, $q_i, q'_i \in \heap{i}{v}$, $\lag_{q_i,q_{i+1}}(v) \leq \oldbound$ and $\lag_{q'_i,q'_{i+1}}(v) \leq \oldbound$.
Then 
\[
\begin{array}{lll}
\lag_{q,q'}(v) \leq \lag_{q_{n},q_{n-1}}(v) + \ldots  + \lag_{q_1,q_{0}}(v) + \lag_{q_0',q_{1}'}(v) + \ldots + \lag_{q'_{n-1},q'_{n}}(v) \leq 2n\oldbound.
\end{array}
\]
\item
By definition of $\heap{0}{wa} = \{ q_{wa} \}$, there exist $p_w \in \choice{k}{\uni}(\rho(w))$ and $s \in \Sigma^*$ such that $(p_{w},a,s,q_{wa}) \in \Delta_{\tra}$.
For every state $q \in \heapf{wa}$, there exist $d \in \mathbb{N}$ and a sequence of states $q_0,q_1,\ldots,q_d$ such that $q_{d} = q$, $q_0  = q_{wa}$, and for every $0 \leq i \leq d-1$, $q_i \in \heap{i}{wa}$ and $\lag_{q_i,q_{i+1}}(wa) \leq \oldbound$.
Then, consider the sequence $p_0,p_1,\ldots,p_d$, where for every $1 \leq i \leq d$, $p_i$ is the state preceding $q_i$ on the run $r_{\tra}(wa,q_i)$, which exists by definition of $\heapf{wa}$.
Then $p_0 = p_w \in \choice{k}{\uni}(\rho(w)) \subseteq \heap{1}{w}$, where the inclusion is implied by \propH{1}.
Hence for every $0 \leq i \leq d-1$, as $\lag_{p_i,p_{i+1}}(w) \leq \lag_{q_i,q_{i+1}}(wa) \leq \oldbound$, $p_i \in \heap{i+1}{w}$.
In particular, $p_{d} \in \heap{d+1}{w} \subseteq \heapf{w}$.
Finally, as $p_d$ is the state preceding $q_d$ on the run $r_{\tra}(wa,q_d)$, there exists $s_d \in \Sigma^*$ such that $(p_d,a,s_d,q_d) \in \Delta_{\tra}$.
\end{enumerate}

\vspace{1mm}
\noindent \textbf{Definition of $\outru{\uni'}$} We define it on the
domain of the words which are prefixes of some word in $\dom(\tra)$,
otherwise it is undefined. Let $v$ be a prefix
of $\dom(\tra)$. The idea is that $U$ can be seen as a selection of
states of $T$, through the function $\choice{k}{\uni}$. We want that
$U'$, on input $v$,  selects at least all the states that $U$ chooses on input
$\rho(v)$, i.e. that $\choice{k}{\uni}(\rho(v)) \subseteq
\choice{\newbound}{\uni'}(v)$. We define $\outru{\uni'}$ in such a way
that the latter inclusion will hold true, which will be shown
in the proof of correctness of the construction of $U'$. 

Now, given a word $v$ that is a prefix of a word of $\dom(\tra)$ let us define $\outru{\uni'}(v)$.
If $v = \epsilon$, let $\outru{\uni'}(v) = \epsilon$.
If $v = wa$, let $\outru{\uni'}(v)$ be the longest common prefix of the outputs of
the runs $r_{\tra}(v,q)$, for all $q \in \heapf{v}$. 

We prove that $\outru{\uni'}$ satisfies the condition required by
the definition of $U'$, i.e. given a word $w \in \Sigma^*$ and $a \in
\Sigma$ such that $wa$ is the prefix of a word of $\dom(\tra)$, the
word $\outru{\uni'}(w)$ is a prefix of $\outru{\uni'}(wa)$.
By \propH{3}, for every $q \in \heapf{wa}$, there exists $p \in \heapf{w}$ and $s \in \Sigma^*$ such that $(p,a,s,q) \in \Delta_{\tra}$.
Therefore the output of the run $r_{\tra}(w,p)$ is a prefix of the output of the run $r_{\tra}(wa,q)$.
This proves that $\outru{\uni'}(w)$ is a prefix of $\outru{\uni'}(wa)$.

\vspace{1mm}
\noindent \textbf{Definition of $\outf{\uni'}$}
Let $v \in \dom(\tra)$.
By \prop{1},  $\rho(v) \in \dom(\tra)$, hence $\choice{k}{\uni}(\rho(v)) \cap F_{\tra}$ is not empty.
Let us choose some state $q$ in this intersection, let $u$ be the output of the run $r_{\tra}(v,q)$, which exists by \prop{1}, and let $\outf{\uni'}(v) = uf(q)$.
Once again, we need to check that $\outru{\uni'}(v)$ is a prefix of $\outf{\uni'}(v)$.
By \propH{1}, $q \in \heap{1}{v} \subseteq \heapf{v}$, hence, by definition of $\outru{\uni'}$, $\outru{\uni'}(v)$ is a prefix of $u$, and the desired result follows.

\vspace{3mm}
\noindent \textbf{Correctness of the construction} We finally
demonstrate that $\uni'$ is an infinite seq-$\newbound$-uniformiser. By
definition of $\Delta'$, $\uni'$ is sequential and by definition of
$F'$ and the fact that $\outru{\uni'}$ is defined for all prefixes of
$\dom(\tra)$, we have $\dom(\uni') = \dom(\tra)$.

It remains to show that $\uni'$ is $\newbound$-included into
$\tra$. We first show that it is the case if one assumes
$\choice{k}{\uni}(\rho(v)) \subseteq \choice{\newbound}{\uni'}(v)$ for
$v$ a prefix of $\dom(\tra)$, and then prove this inclusion. 
Let $v \in \dom(\uni') = \dom(\tra)$.
Since $\uni$ is a $k$-uniformiser of $\tra$, $\choice{k}{\uni}(v)
\cap F_{\tra}$ is non-empty. By definition of $\outf{\uni'}$, 
 $\outf{\uni'}(v) = uf(q)$ for some $q$ in this intersection, where $u$
is the output of the run $r_{\tra}(v,q)$. Then $q \in
\choice{\newbound}{\uni'}(v)$, since we assumed
$\choice{k}{\uni}(\rho(v)) \subseteq \choice{\newbound}{\uni'}(v)$. In
other words, the run of $U'$ on $v$ is $\newbound$-close (in terms of lag)
to one of the runs of $T$ on $v$. This proves that $\uni'$ is
$\newbound$-included into $\tra$.

Finally, let us show the inclusion $\choice{k}{\uni}(\rho(v))
\subseteq \choice{\newbound}{\uni'}(v)$.
By \propH{1}, it is enough to show that $\heapf{v} \subseteq \choice{\newbound}{\uni'}(v)$, which we now prove by induction on the length of $v$.
If $v = \epsilon$, $\heapf{v} = I_{\tra} = \choice{\newbound}{\uni'}(v)$.
Now suppose that $v = wa$, and that the result is true for $w$.
Let $q \in \heapf{wa}$.
By \propH{3} there exist $p \in \heapf{w}$ and $s \in \Sigma^*$ such that $(p,a,s,q) \in \Delta_{\tra}$.
By induction hypothesis, $p \in \choice{\newbound}{\uni'}(w)$, hence $\lag(r_{\tra}(w,p),r_{\uni'}(w)) \leq \newbound$.
Suppose ab absurdo that $q \notin \choice{\newbound}{\uni'}(wa)$.
Therefore $\lag(r_{\tra}(wa,q),r_{\uni'}(wa)) > \newbound$.
As the lag is smaller than $\newbound$ if those runs are restricted to the prefix $w$ of $wa$, this implies that the delay between the outputs of the runs $r_{\tra}(wa,q)$ and $r_{\uni'}(wa)$ is longer than $\newbound$.
By construction of $\uni'$, the output of $r_{\uni'}(wa)$ is $\outru{\uni'}(wa)$, the longest common prefix of the outputs of the runs $r_{\tra}(wa,q')$, $q' \in \heapf{wa}$.
This implies the existence of a state $q'' \in \heapf{wa}$ such that the delay between the outputs of $r_{\tra}(wa,q)$ and $r_{\tra}(wa,q'')$ is longer than $\newbound$.
However, this contradicts \propH{2}.
\end{proof}

% We want the behaviour of $\uni'$ on $w$ to correspond to the behaviour of $\uni$ on $\rho(w)$.
% In particular, we want to ensure that $\choice{2k}{\uni}(\rho(v)) \subseteq
% \choice{\newbound}{\uni'}(v)$ (the value $2k$ will become clear later when
% proving the existence of $\rho$ in the following sections).

The rest of this section is devoted to proving the existence of a
computable integer $\oldbound$ and the existence of a
function $\rho$ satisfying properties $\prop{1}-\prop{3}$.

%%%%%%%%%%%%%%%%%%%%%%%%%%%%%%%%%%%%%%%%%
\subsubsection{Transition monoid of a transducer and properties of delays}\label{subsec:transmod}
%%%%%%%%%%%%%%%%%%%%%%%%%%%%%%%%%%%%%%%%%

In this section, we define the transition monoid of a transducer,
and study the structural properties of its idempotent elements when the
transducer is a finite union of unambiguous transducers. We also prove
properties of delays that are necessary in the following sections.

\vspace{3mm}
\noindent \textbf{Transition monoid} Consider the monoid $\mathcal{M}$
of binary relations $m \subseteq Q_{\tra} \times Q_{\tra}$, where for
any pair $m_1,m_2 \in \mathcal{M}$, $m_1 \cdot m_{2} = \{ (x,z) |
\exists y \in Q_{\tra} \textup{ s.t. } (x,y) \in m_1, (y,z) \in m_2
\}$.

Let $\sigma_{\tra} : \Sigma^{*} \rightarrow \mathcal{M}$ be the monoid morphism mapping any word $w$ to the relation $\sigma_{\tra}(w)$ containing the pairs $(p,q) \in \sigma_{\tra}(w)$ such that there exists a run of $\tra$ on input $w$ between $p$ and $q$.
The \textit{transition monoid} $\mathcal{M}_{\tra}$ of $\tra$ is the
image of $\Sigma^*$ by the morphism $\sigma_{\tra}$.
An element $m \in \mathcal{M}$ is called an \textit{idempotent} if $m^2 = m$.
An element $m \in \mathcal{M}$ is called an \textit{$s$-form} if there exist two distinct elements $q_1,q_2 \in Q_{\tra}$ such that $(q_1,q_1),(q_1,q_2),(q_2,q_2) \in \mathcal{M}$.

The next lemmas present some properties of the elements of the transition monoid.
As we shall see, requiring $\tra$ to be finitely ambiguous greatly reduces the structural complexity of its transition monoid.

\begin{lemma}\label{z-form}
Let $m \in \mathcal{M}_{\tra}$ be an idempotent. 
Then for every element $(q_1,q_2) \in m$, there exists $q \in Q_{\tra}$ such that $(q_1,q),(q,q),(q,q_2) \in m$.
\end{lemma}

\begin{proof}
Let $p_{0}, p_{1}, \ldots, p_{n}$ be a maximal sequence of elements of $Q_{\tra}$ satisfying
\begin{itemize}
\item
$(q_1,p_{0}) \in m$;
\item
$(p_{n}, q_2) \in m$;
\item
for every $0 \leq i \leq n-1$, $(p_{i},p_{i+1}) \in m$;
\item
if $i \neq j$, $p_{i} \neq p_j$.
\end{itemize}
As $(q_1,p_0) \in m$ and $m = m^2$, there exists $q \in Q_{\tra}$ such that $(q_1,q),(q,p_0) \in m$.
Note that the sequence $q, p_{0}, p_{1}, \ldots, p_{n}$ satifies the first three properties, therefore, by maximality of $p_{0}, p_{1}, \ldots, p_{n}$, there exists $0 \leq j \leq n$ such that $q = p_j$.
By supposition, $(p_{n}, q_2) \in m$, and for every $0 \leq i \leq n$, $(p_{i},p_{i+1}) \in m$.
Therefore, as $m^{j} = m$, $(q,q) = (q,p_j) \in m$, and as $m^{n+1} = m$, $(q,q_2) \in m$.
\end{proof}

The next lemma, proved via Ramsey's theorem, 
states that for sufficiently long sequences of words, 
there are necessarily three consecutive blocks of words 
whose concatenations is the same idempotent element in the transition
monoid of $T$.

\begin{lemma}\label{fac_for}
There exists a computable integer $C_{\tra}$ such that for every sequence $v_1, \ldots, v_{C_\tra}$ of $C_\tra$ words, there exist $0 \leq i_1 < i_2 < i_3 < i_4 \leq C_\tra$ such that 
\[
\sigma_{\tra}(v_{i_1} \ldots v_{i_2-1}) = \sigma_{\tra}(v_{i_2} \ldots v_{i_3-1}) = \sigma_{\tra}(v_{i_3} \ldots v_{i_4-1})
\]
is an idempotent.
\end{lemma}

\begin{proof}
We use Ramsey's theorem. Given a sequence $s = (v_1, \ldots, v_{C_\tra})$, let $G_s$ be the complete graph on $C_\tra +1$ vertices $\{s_0, \ldots, s_j\}$ whose edges are coloured in $\mathcal{M}_{\tra}$, as follows.
For every $0 \leq i < j \leq n$ the edge $\{s_i,s_j\}$ is coloured by $\sigma_{\tra}(v_{i+1} \ldots v_{j}) \in \mathcal{M}_{\tra}$.
Then, for every $0 \leq i_1 < i_2 < i_3 < i_4 \leq C_\tra$,
\[
\sigma_{\tra}(v_{i_1} \ldots v_{i_2-1}) = \sigma_{\tra}(v_{i_2} \ldots v_{i_3-1}) = \sigma_{\tra}(v_{i_3} \ldots v_{i_4-1}).
\]
is an idempotent if and only if $s_{i_1},s_{i_2},s_{i_3},s_{i_4}$ forms a monochromatic clique in $G$.
Therefore, the desired result follows from Ramsey's theorem.
\end{proof}

The following lemmas are technical lemmas about the transition monoid
of $\tra$ and properties of delays. They are used in the proof that the
function $\rho$ defined in the next section satisfies the properties
$\prop{1}-\prop{3}$. They are not necessary to understand the
construction of $\rho$.

The next lemma shows that the transition monoid of $\tra$ does not
contain any $s$-form. 

\begin{lemma}\label{s-form}
There is no word $w \in \Sigma^*$ such that $\sigma_{\tra}(w)$ is an $s$-form.
\end{lemma}

\begin{proof}
Suppose ab absurdo that there exists a word $w \in \Sigma^*$ such that $\sigma_{\tra}(w)$ is an $s$-form.
Then there exist two distinct states $q_1,q_2 \in Q_{\tra}$ such that
$(q_1,q_1),(q_1,q_2),(q_2,q_2) \in \sigma_{\tra}(w)$.
Then there exists at least two distinct runs between $q_1$ and $q_2$ on input $ww$, which contradicts the fact that $\tra$ is a trim union of unambiguous transducers.
\end{proof}

\begin{lemma}\label{idem_loop}
Let $m \in \mathcal{M}_{\tra}$ be an idempotent, and let
$q_1,q_2,q_2',q_3 \in Q_{\tra}$. If $(q_1,q_2)\in m$, $(q_2,q_2')\in m$ and $(q_2',q_3) \in m$, then $q_2' = q_3$.
\end{lemma}

\begin{proof}
As $m$ is an idempotent, by Lemma \ref{z-form}, there exist $q^-$ and
$q^+$ such that
$$
\{(q_1,q^-),(q^-,q^-), (q^-,q_2), (q_2',q^+), (q^+,q^+),
(q^+,q_3)\}\subseteq m.
$$
Moreover, as $(q^-,q_2)\in m$, $(q_2,q_2')\in m$ and $(q_2',q^+) \in
m$ and $m$ is idempotent, we have $(q^-,q^+) \in m$.
Therefore $q^- = q^+$, otherwise $m$ would be an $s$-form, which is not possible by Lemma \ref{s-form}.
Therefore, as $(q_2,q_2'),(q_2',q^+),(q^-,q_2) \in m$, $(q_2,q_2) \in m$, and as $(q_2',q^+),(q^-,q_2),(q_2,q_2') \in m$, $(q_2',q_2') \in m$.
Hence $q_2 = q_2'$, otherwise $m$ would be an $s$-form.
\end{proof}

This lemma can be used to detect loops in the runs of $\tra$, as shown by the next corollary.

\begin{corollary}\label{run_loop}
Let $x, y, z \in \Sigma^*$ such that $\sigma_{\tra}(y)$ is an
idempotent, and suppose that there exists a run of $T$ of the form
$$
r\ :\ q_0 \xrightarrow{x|u_1} q_2 \xrightarrow{y|u_{2}} q_2'
\xrightarrow{z|u_{3}} q
$$
If there exist two states $q_1,q_3$ such that $(q_1,q_2),(q_2',q_3) \in \sigma_{\tra}(y)$, then $q_2 = q_2'$, i.e. $q_2 \xrightarrow{y|u_{2}} q_2'$ is a loop.
\end{corollary}

The next lemma is used to decompose input words for which there exist
two runs of $T$ with a sufficiently large lag, into sufficiently many
consecutive subwords on which the delay strictly increases.

\begin{lemma}\label{del_incr}
Let $n \in \mathbb{N}$, let $v \in \Sigma^*$, and let $p$ and $q$ be states such that the runs $r_{\tra}(v,p)$ and $r_{\tra}(v,q)$ exist.
If $\lag_{p,q}(v) \geq 2\alpha m_{\tra}$, then there exists a decomposition $v = v_{1} \ldots v_{\alpha+1}$ of $v$ into non-epsilon subwords such that for every $1 \leq i \leq \alpha$
\[
|\del_{p_i,q_i}(v_1 \ldots v_{i})| < |\del_{p_{i+1},q_{i+1}}(v_1 \ldots v_{i+1})|,
\]
where $p_i$ and $q_i$ denote the states corresponding to the input $v_1 \ldots v_{i}$ in the runs $r_{\tra}(v,p)$ and $r_{\tra}(v,q)$, respectively.
\end{lemma}

\begin{proof}
It is an immediate consequence of Lemma~\ref{lem:delayincreasing}.     
\end{proof}

Now, we prove two lemmas concerning the evolution of the delay between two words obtained by iterating a subword.
The first one is related to Lemma \ref{lem:folkore}, and their proofs are very similar.

\begin{lemma}\label{incr_del}
Let $u_1,u_2,v_1,v_2 \in \Sigma^*$, let $n \in \mathbb{N}$ such that $n>0$, and suppose that $|\del(u_1,v_1)|\leq n$.

If $\del(u_1u_2,v_1v_2) \neq \del(u_1,u_2)$, then $|\del(u_1u_2^{3n},v_1v_2^{3n})| > n$.
\end{lemma}

\begin{proof}
    First, note that $u_2$ and $v_2$ are not both $\epsilon$, since
    $\del(u_1,v_1)\neq \del(u_1u_2,v_1v_2)$.

    Suppose that $|u_2| \neq |v_2|$.
    Then 
    $$
    \begin{array}{lllllllllllll}
    & |\del(u_1u_2^{3n},v_1v_2^{3n})| \geq
    ||u_1u_2^{3n}|-|v_1v_2^{3n}|| \geq 3n||u_2|-|v_2|| -||u_1|-|v_1|| \geq 2n > n
    \end{array}
    $$   
    Assume now that $|u_2| = |v_2| > 0$.

    Suppose that $u_1u_2^{2n}$ is not a prefix of $v_1v_2^{2n}$ and $v_1v_2^{2n}$ is not a
    prefix of $u_1u_2^{2n}$, i.e. $u_1u_2^{2n} = w \alpha u'$ and $v_1v_2^{2n} = w \beta
    v'$ for $u',v',w\in\Sigma^*$, $\alpha,\beta\in \Sigma$ and
    $\alpha \neq \beta$. 
 
    Then the word in $(\Sigma\cup \overline{\Sigma})^*$
    $$
    \begin{array}{lllllllllllll}
    & \del(u_1u_2^{3n},v_1v_2^{3n}) & = &
    u_2^{-3n}u_1^{-1}v_1v_2^{3n} & = & u_2^{-n} u'^{-1} \alpha^{-1}\beta
    v' v_2^{n}
    \end{array}
    $$ 
    is irreducible, hence $|\del(u_1u_2^{3n},v_1v_2^{3n})| > n$.
    
    Assume now that $v_1v_2^{2n} = u_1u_2^{2n}s$ for some $s\in\Sigma^*$ (the
    case where $v_1v_2^{2n}$ is a prefix of $u_1u_2^{2n}$ is symmetric and therefore
    untreated). Then, as $|u_2| = |v_2|$, there exist $s_0,s_1 \in \Sigma^*$ such that
    $v_1 = u_1s_0$, $v_1v_2 = u_1u_2s_1$, and $|s_0| = |s_1|$. Therefore, $s_0v_2^{2n} = u_2^{2n}s$, 
    and $s_1v_2^{2n-1} = u_2^{2n-1}s$. However, as $|s_0| = |\del(u_1,v_1)|\leq n$ by supposition,
    and $|v_2| \neq 0$, $s_0$ is a prefix of $u_2^{n}$, and so is $s_1$.
    Therefore, as $|s_0| = |s_1|$, they are equal, which contradicts the fact that
    $\del(u_1u_2,v_1v_2) \neq \del(u_1,v_1)$.
\end{proof}

\begin{lemma}\label{stable_del}
Let $x_1,x_2,x_3,y_1,y_2,y_3 \in \Sigma^*$, let $n \in \mathbb{N}$ such that $n>0$, and suppose that $|\del(x_1,y_1)| \leq n$ and $|\del(x_1x_2^{3n},y_1y_2^{3n})| \leq n$.
Then $\del(x_1x_2^{3n}x_3,y_1y_2^{3n}y_3) = \del(x_1x_3,y_1y_3)$.
\end{lemma}

\begin{proof}

Since $|\del(x_1,y_1)| \leq n$ and $|\del(x_1x_2^{3n},y_1y_2^{3n})|
\leq n$ by supposition, Lemma \ref{incr_del} implies that
$\del(x_1x_2,y_1y_2) = \del(x_1,y_1)$. This proves inductively that we
have the equality $\del(x_1x_2^{3n},y_1y_2^{3n}) = \del(x_1,y_1)$, and the desired result follows.
\end{proof}

By combining the previous results concerning the elements of the transition monoid, and the delay between words, we get the following lemma.

\begin{lemma}\label{loop_del}
Let $x,y,z \in \Sigma^*$ such that 
\begin{itemize}
\item
$\sigma_{\tra}(y)$ is an idempotent;
\item for every $(p,q) \in \sigma_{\tra}(x)$, there exists a state $p'$ such that $(p',q) \in \sigma_{\tra}(y)$;
\item for every $(p,q) \in \sigma_{\tra}(z)$, there exists a state $q'$ such that $(p,q') \in \sigma_{\tra}(y)$.
\end{itemize}
Let $p,q \in Q_{\tra}$ such that the runs $r_{\tra}(xy^{3n}z,p)$ and $r_{\tra}(xy^{3n}z,q)$ exist.
If $\lag_{p,q}(xy^{3n}z) \leq n$, then the runs $r_{\tra}(xz,p)$ and $r_{\tra}(xz,q)$ exist, and $\lag_{p,q}(xy^{3n}z) \geq \lag_{p,q}(xz)$.
\end{lemma}

\begin{proof}
By Corollary \ref{run_loop}, given a run on the input $xy^{3n}z$, the transducer $\tra$ loops on the input $y^{3n}$.
This proves the existence of the runs $r_{\tra}(xz,p)$ and $r_{\tra}(xz,q)$.

Let $P$ denote the set of prefixes of $xz$, and let $P'$ denote the set of prefixes of $xy^{3n}z$.
For every $w \in P$, let $p_w$ and $q_w$ denote the states corresponding to the input $w$ in the runs $r_{\tra}(xz,p)$ and $r_{\tra}(xz,q)$, respectively.
Similarly, for every $w \in P'$, let $p'_{w}$ and $q'_{w}$ denote the states corresponding to the input $w$ in the runs $r_{\tra}(xy^{3n}z,p)$ and $r_{\tra}(xy^{3n}z,q)$, respectively.
By definition, 
\[
\begin{array}{lll}
\lag_{p,q}(xz) & = & \max \{ \del_{p_w,q_w}(w) | w \in P \},\\
\lag_{p,q}(xy^{3n}z) & = & \max \{ \del_{p'_w,q'_w}(w) | w \in P' \}.
\end{array}
\]
In order to prove the lemma, we shall expose, for every $w \in P$, a prefix $w' \in P'$ such that $\del_{p_w,q_w}(w) = \del_{p'_{w'},q'_{w'}}(w')$.

If $w$ is a prefix of $x$, let $w' = w$, and if there exists $v$ such that $w = xv$, let $w' = xy^{3n}v$.
By Corollary \ref{run_loop}, we know that for any run on the input $xyz$, $\tra$ will loop on the input $y$.
Therefore, those runs are as follows.
\[
\begin{array}{ll ll}
r_{\tra}(xyz,p): & p_0 \xrightarrow{x|u_1} p_{x} \xrightarrow{y|u_{2}} p_{x} \xrightarrow{z|u_{3}} p, &  r_{\tra}(xy^{3n}z,p): & p_0 \xrightarrow{x|u_1} p_{x} \xrightarrow{y^{3n}|u_{2}^{3n}} p_{x} \xrightarrow{z|u_{3}} p,\\
r_{\tra}(xyz,q): & q_0 \xrightarrow{x|v_1} q_{x} \xrightarrow{y|v_{2}} q_{x} \xrightarrow{z|v_{3}} q, &  r_{\tra}(xy^{3n}z,q): & q_0 \xrightarrow{x|v_1} q_{x} \xrightarrow{y^{3n}|v_{2}^{3n}} q_{x} \xrightarrow{z|v_{3}} q.
\end{array}
\]
Hence $p'_{w'} = p_w$ and $q'_{w'} = q_w$.
If $w' = w$, the desired result follows immediately.
If $w' = xy^{3n}v$,
\[
\del_{p_x,q_x}(xy^{3n}) = \del(u_1u_2^{3n},v_1v_2^{3n}) \stackrel{\mathclap{\scriptscriptstyle\smash{(1)}}}{=} \del(u_1,v_1) = \del_{p_x,q_x}(x),
\]
where equality \ensuremath{(1)} follows from Lemma \ref{stable_del}, which can be applied, as both $|\del(u_1,v_1)| \leq n$ and $|\del(u_1u_2^{3n},v_1v_2^{3n})| \leq n$ since $\lag_{p,q}(xy^{3n}z) \leq n$.
Finally,
\[ \del_{p_w,q_w}(w) = \del_{p_w,q_w}(xv) = \del_{p_{w},q_{w}}(xy^{3n}v) = \del_{p_{w},q_{w}}(w') = \del_{p'_{w'},q'_{w'}}(w').\]
\end{proof}

%%%%%%%%%%%%%%%%%%%%%%%%%%%%%%%%%%%%%%%%%
\subsubsection{Definition of the function $\rho$}\label{subsec:defrho}
%%%%%%%%%%%%%%%%%%%%%%%%%%%%%%%%%%%%%%%%%

We now define the function $\rho$ and the integer $\oldbound$
mentioned in the proof of Theorem \ref{f_a_det}.

The function $\rho$ is defined inductively on words, and the main idea
is to pump the second occurence of any two consecutive subwords which
have the same idempotent element in the transition monoid. This pumping
is done by an auxiliary function $\phi$, based on the following idea.
Whenever two consecutive subwords of $v$ correspond to the same idempotent in the transition monoid, we iterate the second one $12k$ times.
As we are iterating idempotents, we preserve the corresponding element of the transition monoid, proving that $\phi$ satisfies \prop{1}.
Moreover, since only the idempotents that appear twice in a row are iterated, we obtain good properties concerning the lag.
We define $\oldbound$ as the product $2C_{\tra}m_{\tra}$, where $C_{\tra}$ denotes the integer defined in Lemma \ref{fac_for}.
By Lemma \ref{del_incr}, this ensures that whenever the lag between two runs on a word $v$ is greater than $\oldbound$, there exist three consecutive subwords $v_1$, $v_2$ and $v_3$ of $v$ such that those three words correspond to the same idempotent in the transition monoid, and the delay grows along them.
Then, as the word $v_2$ is iterated on $\phi(v)$, the delay between the two corresponding runs on $\phi(v)$ will explode, which we use to prove \prop{2}.
Unfortunately, $\phi$ does not satisfy \prop{3}.
This stems from the fact that, given a word $w$ and a letter $a$, $\phi(w)$ is not necessarily a prefix of $\phi(wa)$.
However by using an intermediate function $\phi'$, that maps each word
$w$ to a well chosen suffix of $\phi(w)$, we are able to define the
function $\rho$ inductively, in a way that some properties of
$\phi$ transfer to $\rho$, and such that $\rho$ will also satisfy  \prop{3}.

Let us now define those functions formally.
For every $v \in \Sigma^*$, let 
\[
U_v = ((w_1,x_1,y_1,z_1),\ldots,(w_n,x_n,y_n,z_n) ) \in (\Sigma^* \times \Sigma^* \times \Sigma^* \times \Sigma^*)^*
\]
 denote the sequence of decompositions $v = wxyz$ into four words such
 that $\sigma_{\tra}(x) = \sigma_{\tra}(y)$ is an idempotent, ordered as follows:
\begin{itemize}
\item
If $|w_ix_i| < |w_jx_j|$, then $i<j$;
\item
If $|w_ix_i| = |w_jx_j|$ and $|y_i| < |y_j|$, then $i<j$;
\item
If $|w_ix_i| = |w_jx_j|$, $|y_i| = |y_j|$, and $|x_i| < |x_j|$ then $i<j$;
\end{itemize}
We now consider the decomposition of $v = v_1 \ldots v_{n+1}$ into $n+1$ words such that for every $1 \leq i \leq n$, $v = v_1 \ldots v_i y_i z_i = w_ix_i v_{i+1} \ldots v_{n+1}$.
Moreover, let $1 \leq l \leq n+1$ be equal to $n+1$ if all the $z_i$ are different from $\epsilon$, and be equal to the smallest integer $d$ such that $z_d = \epsilon$ otherwise.
The function $\phi$ iterates all the idempotents $y_i$ $12k$ times, the function $\phi'$ maps $v$ to a suffix of $\phi(v)$, and $\rho$ is defined inductively, using $\phi'$.
\[
\begin{array}{lcllll}
\phi & : & \Sigma^* & \rightarrow & \Sigma^*,\\
& & v & \mapsto & v_1y_1^{12k} \ldots v_ny_n^{12k}v_{n+1}.\\
\phi' & : & \Sigma^* & \rightarrow & \Sigma^*,\\
& & v & \mapsto & y_l^{12k-1}v_{l+1}y_{l+1}^{12k} \ldots v_ny_n^{12k}v_{n+1}.\\
\rho & : & \Sigma^* & \rightarrow & \Sigma^*,\\
& & \epsilon & \mapsto & \epsilon,\\
& & wa & \mapsto & \rho(w)a\phi'(wa).
\end{array}
\]
We also define the two following sequences, that expose the decomposition of $v$.
\[
\begin{array}{lcllll}
S_v & = & \multicolumn{3}{l}{(v_1,y_1,\ldots,v_n,y_n,v_{n+1});}\\
S_v' & = & \multicolumn{3}{l}{(v_1 \ldots v_l,y_l,v_{l+1},y_{i+1}\ldots,v_n,y_n,v_{n+1})}.
\end{array}
\]

We now prove a technical result, stating that for every word $w$ and every letter $a$, the sequences $S_w$ and $S_{wa}$ are identical on the prefix of $S_{wa}$ that is dropped in $S'_{wa}$.

\begin{lemma}\label{same_seq}
Let $w \in \Sigma^*$ and $a \in \Sigma$, let 
$$
S_{wa} = (v_1,y_1,\ldots,v_m,y_m,v_{m+1})\qquad S_{w} =
(v_1',y_1',\ldots,v_n',y_n',v_{n+1}')
$$
Let $1 \leq l \leq m+1$ be the integer such that $S'_{wa} = (v_1 \ldots v_l,y_l,v_{l+1},y_{l+1}\ldots,v_m,y_m,v_{m+1})$.
Then for every $1 \leq i < l$, $v_i' = v_i$, $y_i' = y_i$, and $v_l' \ldots v_{n+1}'a = v_ly_l$.
\end{lemma}

\begin{proof}
Let 
\[
\begin{array}{llllllll}
U_{wa} & = & ((w_1,x_1,y_1,z_1), \ldots, (w_m,x_m,y_m,z_m))\\
U_w & = & ((w_1',x_1',y_1',z_1'), \ldots, (w_n',x_n',y_n',z_n')).
\end{array}
\]
By definition of $U_w$ and $U_{wa}$, 
$
\tau_a(U_w) := ((w_1',x_1',y_1',z_1'a), \ldots, (w_{n}',x_{n}',y_{n}',z_{n}'a))
$
is a subsequence of $U_{wa}$.
For every $1 \leq i < l$, by definition of $l$, $z_i \neq \epsilon$, hence, as $w_ix_iy_iz_i = wa$, there exists $u_i \in \Sigma^*$ such that $z_i = u_ia$.
Then $(w_i,x_i,y_i,u_i)$ belongs to $U_w$, and $(w_i,x_i,y_i,z_i) = (w_i,x_i,y_i,u_ia)$ belongs to $\tau_a(U_w)$.
Therefore for every $1 \leq i < l$, $(w_i',x_i',y_i',z_i'a) = (w_i,x_i,y_i,z_i)$, hence $y_i' = y_i$, $v_i' = v_i$, and, as $v_1' \ldots v_{n+1}'a = wa = v_1 \ldots v_ly_l$, $v_l' \ldots v_{n+1}'a = v_ly_l$.
\end{proof}

%%%%%%%%%%%%%%%%%%%%%%%%%%%%%%%%%%%%%%%%%
\subsubsection{Property 1}\label{subsec:p1}
%%%%%%%%%%%%%%%%%%%%%%%%%%%%%%%%%%%%%%%%%

Now, we shall prove that, as the subwords $y_i$ correspond to idempotents in the transition monoid, iterating them does not modify the corresponding element of the transition monoid.
This will yield the proof that $\rho$ satisfies \prop{1}.

\begin{lemma}\label{equ1}
Let $v \in \Sigma^*$, and let $S_v = (v_1,y_1, \ldots, v_n,y_n,v_{n+1})$.
Then for every $0 \leq i \leq n$, for every $0 \leq j \leq n-i$,
\[
\sigma_{\tra}(v_1v_2 \ldots v_{i}v_{i+1}y_{i+1}^{12k}v_{i+2}y_{i+2}^{12k} \ldots v_{i+j}y_{i+j}^{12k}) = \sigma_{\tra}(v_{1} \ldots v_{i+j}).
\]
\end{lemma}

\begin{proof}
We shall prove this by induction over $j$.
If $j = 0$, the result is immediate.
Now suppose that $j > 0$, and that the result is true for $j-1$.
By definition of $S_v$, there exists $w_{i+j},x_{i+j} \in \Sigma^{*}$ such that $\sigma_{\tra}(x_{i+j}) = \sigma_{\tra}(y_{i+j})$ is an idempotent, and $v_{1} \ldots v_{i+j} = w_{i+j}x_{i+j}$. 
Then
\[
\begin{array}{lllll}
\multicolumn{3}{l}{\sigma_{\tra}(v_1v_2 \ldots v_{i}v_{i+1}y_{i+1}^{12k}v_{i+2}y_{i+2}^{12k} \ldots v_{i+j}y_{i+j}^{12k})}\\
& \stackrel{\mathclap{\scriptscriptstyle\smash{(1)}}}{=} & \sigma_{\tra}(v_1v_2  \ldots v_{i+j-1}v_{i+j}y_{i+j}^{12k})\\
& = & \sigma_{\tra}(w_{i+j}x_{i+j}y_{i+j}^{12k})\\
& \stackrel{\mathclap{\scriptscriptstyle\smash{(2)}}}{=} & \sigma_{\tra}(w_{i+j}x_{i+j})\\
& = & \sigma_{\tra}(v_{1} \ldots v_{i+j}),
\end{array}
\]
where equality \ensuremath{(1)} follows from the induction hypothesis, and equality \ensuremath{(2)} follows from the fact that $\sigma_{\tra}(x_{i+j}) = \sigma_{\tra}(y_{i+j})$ is an idempotent.
\end{proof}

\begin{lemma}\label{equsame}
Let $v \in \Sigma^*$, and let $S_v = (v_1,y_1, \ldots, v_n,y_n,v_{n+1})$.
Then for every $0 \leq i \leq n$, for every $0 \leq j \leq n-i$,
\[
\sigma_{\tra}(y_{i+1}^{12k}v_{i+2}y_{i+2}^{12k} \ldots v_{i+j}y_{i+j}^{12k}v_{i+j+1}v_{i+j+2} \ldots v_n v_{n+1}) = \sigma_{\tra}(v_{i+1} \ldots v_{n+1}).
\]
\end{lemma}

\begin{proof}
This is proved similarly to Lemma \ref{equ1}.
By definition of $S_v$, for every $i+1 \leq d \leq i+j$, there exists $z_{d} \in \Sigma^{*}$ such that $v_{d+1} \ldots v_{n+1} = y_{d}z_{d}$. 
This allows us to absorb $y_d$ in the suffix $v_{d+1} \ldots v_{n+1}$, starting with $d = i+j$, until $d = i+1$.
\end{proof}

As a consequence of those lemmas, we have the following corollary.

\begin{corollary}\label{equ1'}
Let $v \in \Sigma^*$, $S_v = (v_1,y_1, \ldots, v_n,y_n,v_{n+1})$, $S'_v = (v_1 \ldots v_l,y_l, \ldots, v_n,y_n,v_{n+1})$.
Then
\begin{enumerate}
\item
$\sigma_{\tra}(\phi(v)) = \sigma_{\tra}(v)$;
\item\label{equ3}
$\sigma_{\tra}(\rho(v)) = \sigma_{\tra}(v)$.
\item\label{equ4}
$\sigma_{\tra}(\phi'(v)) = \sigma_{\tra}(y_l)$.
\end{enumerate}
\end{corollary}

\begin{proof}
\begin{enumerate}
\item
Take $i = 0$, $j = n$ in Lemma \ref{equ1}.
\item
We prove the desired result by induction over the length of $v$.
If $v = \epsilon$, the result is immediate, as $\rho(\epsilon) = \epsilon$.
Now suppose that $v = wa$, and that the result is true for $w$.
Let $S'_v = (v_1 \ldots v_l,y_l,v_{l+1},y_{l+1}, \ldots, v_n,y_n,v_{n+1})$
Note that, by definition of $S'_v$, $v = v_1 \ldots v_ly_l$.
Then,
\[
\begin{array}{lll}
\multicolumn{3}{l}{\sigma_{\tra}(\rho(wa))}\\
& = & \sigma_{\tra}(\rho(w)a\phi'(wa))\\
& \stackrel{\mathclap{\scriptscriptstyle\smash{(1)}}}{=} & \sigma_{\tra}(wa\phi'(wa))\\
& = & \sigma_{\tra}(v_1 \ldots v_l y_l \phi'(wa))\\
& = & \sigma_{\tra}(v_1 \ldots v_l y_ly_l^{12k-1}v_{l+1}y_{l+1}^{12k} \ldots v_ny_n^{12k}v_{n+1})\\
& = & \sigma_{\tra}(v_1 \ldots v_l y_l^{12k}v_{l+1}y_{l+1}^{12k} \ldots v_ny_n^{12k}v_{n+1})\\
& \stackrel{\mathclap{\scriptscriptstyle\smash{(2)}}}{=} & \sigma_{\tra}(v_1 \ldots v_{n+1})\\
& = & \sigma_{\tra}(v),
\end{array}
\]
where equality \ensuremath{(1)} follows from the induction hypothesis, and equality \ensuremath{(2)} follows from Lemma \ref{equ1} in the particular case $i = l-1$, $j = n - l +1$.
\item
By definition of $S_v'$, $v_{l+1} \ldots v_{n+1} = y_l$.
Therefore,
\[
\begin{array}{lll}
\multicolumn{3}{l}{\sigma_{\tra}(\phi'(wa))}\\
& = & \sigma_{\tra}(y_{l}^{12k-1} v_{l+1} y_{l+1}^{12k} \ldots v_ny_n^{12k}v_{n+1})\\
& \stackrel{\mathclap{\scriptscriptstyle\smash{(1)}}}{=} & \sigma_{\tra}(y_{l}^{12k-1} v_{l+1} \ldots v_{n+1})\\
& = & \sigma_{\tra}(y_{l}^{12k-1}y_l)\\
& = & \sigma_{\tra}(y_l),
\end{array}
\]
where equality \ensuremath{(1)} follows from Lemma \ref{equsame} in the particular case $i = l$, $j = n - l$.
\end{enumerate}
\end{proof}

\begin{corollary}\label{coro:prop1}
The function $\rho$ satisfies \prop{1}.
\end{corollary}

\begin{proof}
Let us state \prop{1} once again.
\begin{description}
\item[\prop{1}]: for every $w \in \Sigma^*$, for every state $q$, $r_{\tra}(w,q)$ exists if and only if $r_{\tra}(\rho(w),q)$ does.
\end{description}
Let $w \in \Sigma^*$.
By definition of the transition monoid of $\tra$, the run $r_{\tra}(w,q)$ exists if and only if there exists an initial state $q_0$ such that $(q_0,q) \in \sigma_{\tra}(w)$.
As $\sigma_{\tra}(\rho(w)) = \sigma_{\tra}(w)$ by Corollary \ref{equ1'}.\ref{equ3}, we obtain the desired result.
\end{proof}

%%%%%%%%%%%%%%%%%%%%%%%%%%%%%%%%%%%%%%%%%
\subsubsection{Property 2}\label{subsec:p2}
%%%%%%%%%%%%%%%%%%%%%%%%%%%%%%%%%%%%%%%%%

We begin by exposing general results concerning the behaviour of $\phi(v)$ with respect to the notion of lag.
Then, we show that those results can be extended to $\rho$.
This will ultimately allow us to prove that $\rho$ satisfies \prop{2}.

The following lemma will guarantee that, for all words $v\in\Sigma^*$,
the way of pumping idempotents of $v$ to define
$\phi(v)$ will not decrease the initial lag between any two
runs of $T$ on $v$. 

\begin{lemma}\label{inequ1}
Let $v \in \Sigma^*$, let $S_v = (v_1,y_1, \ldots, v_n,y_n,v_{n+1})$, let $p$ and $q$ be two states such that the runs $r_{\tra}(v,q)$ and $r_{\tra}(v,p)$ exist.
For every $1 \leq i \leq j \leq n+1$, let $u,u' \in \Sigma^*$ such that $\sigma_{\tra}(u) = \sigma_{\tra}(v_{1} \ldots v_{i-1})$ and $\sigma_{\tra}(u') = \sigma_{\tra}(v_{j+1} \ldots v_{n+1})$.
If 
\[
\lag_{p,q}(uv_{i}y_{i}^{12k} \ldots v_{j-1}y_{j-1}^{12k}v_{j}u') \leq 4k,
\]
then
\[
\lag_{p,q}(uv_{i}y_{i}^{12k} \ldots v_{j-1}y_{j-1}^{12k}v_{j}u') \geq \lag_{p,q}(uv_{i} \ldots v_{j}u').
\]
\end{lemma}

\begin{proof}
By Lemma \ref{equ1}, the runs to $p$ and $q$ implicit in the definition of the $lag_{p,q}(\cdot)$ present in the statement of this lemma exist.
We prove the desired result by induction over $j-i$.
If $i = j$, it is immediate.
Now suppose that $j > i$, and that the result is true for $i+1$.

Then
\[
\begin{array}{lll}
\multicolumn{3}{l}{\lag_{p,q}(uv_{i}y_{i}^{12k}v_{i+1} \ldots y_{j-1}^{12k}v_{j}u')}\\
& \stackrel{\mathclap{\scriptscriptstyle\smash{(1)}}}{\geq} & \lag_{p,q}(uv_{i}y_{i}^{12k}v_{i+1}v_{i+2} \ldots v_{j-1} v_{j} u')\\
& \stackrel{\mathclap{\scriptscriptstyle\smash{(2)}}}{\geq} & \lag_{p,q}(uv_{i}v_{i+1}v_{i+2} \ldots v_{j-1} v_{j} u')\\
& = & \lag_{p,q}(uv_{i} \ldots v_{j} u').
\end{array}
\]
where inequality \ensuremath{(1)} follows from the induction hypothesis, which can be applied, as 
\[
\sigma_{\tra}( uv_{i}y_{i}^{12k}) = \sigma_{\tra}( v_1v_2 \ldots v_{i-1}v_{i}y_{i}^{12k}) = \sigma_{\tra}(v_1 \ldots v_{i})
\]
by Lemma \ref{equ1}, and inequality \ensuremath{(2)} follows from Lemma \ref{loop_del}, whose requirements we shall now check.

First, $\lag_{p,q}(uv_{i}y_{i}^{12k}v_{i+1}v_{i+2} \ldots v_{j-1} v_{j} u') \leq 4k$ by inequality \ensuremath{(1)} and the hypothesis.

Moreover, by definition of $S_v$, there exists $w_i,x_{i} \in \Sigma^*$ such that $\sigma_{\tra}(x_{i}) = \sigma_{\tra}(y_{i})$, and  $v_1 \ldots v_{i} = w_{i}x_{i}$.
Hence, for every 
\[
(p,q) \in \sigma_{\tra}(uv_i) = \sigma_{\tra}(v_{1} \ldots v_{i}) = \sigma_{\tra}(w_{i}x_{i}) = \sigma_{\tra}(w_{i}y_{i}),
\]
there exists a state $p'$ such that $(p,p') \in \sigma_{\tra}(w_{i})$ and $(p',q) \in \sigma_{\tra}(y_i)$.

Finally, by definition of $S_v$, there exists $z_{i} \in \Sigma^*$ such that $v_{i+1} \ldots v_{n+1} = y_{i}z_{i}$.
Hence, for every
\[
(p,q) \in \sigma_{\tra}(v_{i+1} \ldots v_j u') = \sigma_{\tra}(v_{i+1} \ldots v_{n+1}) = \sigma_{\tra}(y_{i}z_{i}),
\]
there exists a state $q'$ such that $(p,q') \in \sigma_{\tra}(y_i)$ and $(q',q) \in \sigma_{\tra}(z_{i})$.
\end{proof}

Based on the previous lemma, we show the following corollary which is
crucial in the proof that $\rho$ satisfies property $\prop{2}$. 

\begin{corollary}\label{inequ'}
Let $v \in \Sigma^*$, let $S_v = (v_1,y_1, \ldots, v_n,y_n,v_{n+1})$, and let $p$ and $q$ be two states such that the runs $r_{\tra}(v,q)$ and $r_{\tra}(v,p)$ exist, and $\lag_{p,q}(\phi(v)) \leq 4k$.
Then for every $1 \leq d \leq n$,
\begin{enumerate}
\item\label{lag|||_}
$\lag_{p,q}(\phi(v)) \geq \lag_{p,q}(v_1y_{1}^{12k}v_{2}y_{2}^{12k} \ldots v_{d-1}y_{d-1}^{12k}v_{d}v_{d+1} \ldots v_{n}v_{n+1})$;
\item\label{lag_|_}
$\lag_{p,q}(\phi(v)) \geq \lag_{p,q}(v_1v_2 \ldots v_{d-1}v_{d}y_{d}^{12k}v_{d+1}v_{d+2} \ldots v_{n}v_{n+1})$;
\item\label{lag_p}
$\lag_{p,q}(\phi(v)) \geq \lag_{p,q}(v)$;
\item\label{lag_r_p}
$\lag_{p,q}(\rho(v)) \geq \lag_{p,q}(\phi(v))$;
\item\label{lag_r}
$\lag_{p,q}(\rho(v)) \geq \lag_{p,q}(v)$.
\end{enumerate}
\end{corollary}

\begin{proof}
By Lemma \ref{equ1}, the runs to $p$ and $q$ implicit in the definition of the $lag_{p,q}(\cdot)$ present in the statement of this lemma exist.
\begin{enumerate}
\item
Take $i = d$, $j= n+1$, $u = v_1y_{1}^{12k} \ldots v_{d-1}y_{d-1}^{12k}$ and $u' = \epsilon$ in Lemma \ref{inequ1}.
\item
Two steps are required.
First, take $i = 1$, $j= d-1$, $u = \epsilon$ and $u' = v_{d} y_d^{12k} v_{d+1} v_{d+2} \ldots v_n v_{n+1}$ in Lemma \ref{inequ1}, then apply the first point.
\item
Take $d = 1$ in the first point.
\item
We prove the desired result by induction over the length of $v$.
If $v = \epsilon$, $\rho(v) = \epsilon = \phi(v)$, and the result follows.
Now suppose that $v = wa$, and that the result is true for $w$.
Let 
\[
\begin{array}{lll}
S_v' & = & (v_1 \ldots v_l,y_l,\ldots, v_m, y_m, v_{m+1});\\
S_w & = & (w_1,z_1,\ldots, w_n, z_n, w_{n+1}).
\end{array}
\]
Then
\[
\begin{array}{lll}
\multicolumn{3}{l}{\lag_{p,q}(\rho(v))}\\
& = & \lag_{p,q}(\rho(w)a\phi'(v))\\
& \stackrel{\mathclap{\scriptscriptstyle\smash{(1)}}}{\geq} & \lag_{p,q}(\phi(w)a\phi'(v))\\
& \stackrel{\mathclap{\scriptscriptstyle\smash{(2)}}}{\geq} & \lag_{p,q}(w_1z_1^{12k} \ldots w_{l-1}z_{l-1}^{12k}w_l \ldots w_{n+1}a\phi'(v))\\
& = & \lag_{p,q}(w_1z_1^{12k} \ldots w_{l-1}z_{l-1}^{12k}w_l \ldots w_{n+1}ay_l^{12k-1}v_{l+1}y_{l+1}^{12k} \ldots v_my_m^{12k}v_{m+1})\\
& \stackrel{\mathclap{\scriptscriptstyle\smash{(3)}}}{=} & \lag_{p,q}(v_1y_1^{12k} \ldots v_{l-1}y_{l-1}^{12k}v_ly_ly_l^{12k-1}v_{l+1}y_{l+1}^{12k} \ldots v_my_m^{12k}v_{m+1})\\
& = & \lag_{p,q}(v_1y_1^{12k} \ldots v_my_m^{12k}v_{m+1})\\
& = & \lag_{p,q}(\phi(v)),
\end{array}
\]
where inequality \ensuremath{(1)} follows from the induction hypothesis, inequality \ensuremath{(2)} follows from the first point, and equality \ensuremath{(3)} follows from Lemma \ref{same_seq}.
\item
This follows immediately from the points 3 and 4.
\end{enumerate}
\end{proof}

The following theorem states a property of $\phi$ similar to property
$\prop{2}$. Combined with the previous corollary (item 4), it will allow us to
show that $\rho$ satisfies property $\prop{2}$. 

\begin{theorem}
Let $v \in \Sigma^*$.
If $\lag_{p,q}(v) \geq \oldbound$, then $\lag_{p,q}(\phi(v)) > 4k$.
\end{theorem}

\begin{proof}
As $\lag_{p,q}(v) \geq \oldbound = 2m_{\tra}C_{\tra}$, by Lemma \ref{del_incr} there exists a decomposition of $v = v_1 \ldots v_{C_{\tra}+1}$ into $C_{\tra}+1$ non-epsilon subwords such that for every $1 \leq i \leq C_{\tra}-1$, if $p_i$ and $q_i$ denote the states corresponding to the input $v_1 \ldots v_i$ in the runs $r_{\tra}(v,p)$ and $r_{\tra}(v,q)$, respectively, then $|\del_{p_i,q_i}(v_1 \ldots v_i)| < |\del_{p_{i+1},q_{i+1}}(v_1 \ldots v_{i+1})|$.
Moreover, by Lemma \ref{fac_for}, there exist $0 < i_1 < i_2 < i_3 < i_4 \leq C_{\tra}$ such that 
\[
\sigma_{\tra}(v_{i_1} \ldots v_{i_2-1}) = \sigma_{\tra}(v_{i_2} \ldots v_{i_3-1}) = \sigma_{\tra}(v_{i_3} \ldots v_{i_4-1})
\]
is an idempotent.
Now, consider the runs of $\tra$:
\[
\begin{array}{lll}
p_0 \! \xrightarrow{v_1\ldots v_{i_1{-}1}|g_1} \! p_{i_1{-}1} \! \xrightarrow{v_{i_1}\ldots v_{i_2{-}1}|g_2} \! p_{i_2{-}1} \! \xrightarrow{v_{i_2} \ldots v_{i_3{-}1}|g_{3}} \! p_{i_3{-}1} \! \xrightarrow{v_{i_3} \ldots v_{i_4-1}|g_{4}} \! p_{i_4-1} \! \xrightarrow{v_{i_4} \ldots v_{C_{\tra}+1}|g_{5}} \! p,\\
q_0 \! \xrightarrow{v_1\ldots v_{i_1-1}|h_1} \! p_{i_1-1} \! \xrightarrow{v_{i_1}\ldots v_{i_2-1}|h_2} \! q_{i_2-1} \! \xrightarrow{v_{i_2} \ldots v_{i_3-1}|h_{3}} \! q_{i_3-1} \! \xrightarrow{v_{i_3} \ldots v_{i_4-1}|h_{4}} \! q_{i_4-1} \! \xrightarrow{v_{i_4} \ldots v_{C_{\tra}+1}|h_{5}} \! q.
\end{array}
\]
$p_{i_2-1} = p_{i_3-1}$ and $q_{i_2-1} = q_{i_3-1}$, by Corollary \ref{run_loop}.
Hence, by Lemma \ref{stable_del}, since 
\[
\begin{array}{l}
|\del_{p_{i_2-1},q_{i_2-1}}(v_1 \ldots v_{i_2-1})| < |\del_{p_{i_3-1},q_{i_3-1}}(v_1 \ldots v_{i_3-1})|, \\
\end{array}
\]
either
\[
\begin{array}{lllll}
|\del_{p_{i_2-1},q_{i_2-1}}(v_1 \ldots v_{i_2-1})| > 4k
& \! \! \! \! \textup{or} \! \! \! \! &
|\del_{p_{i_2-1},q_{i_2-1}}(v_1\ldots v_{i_2-1}(v_{i_2} \ldots v_{i_3-1})^{12k})| > 4k.
\end{array}
\]
Therefore $\lag_{p_{i_2-1},q_{i_2-1}}(v_1\ldots v_{i_2-1}(v_{i_2} \ldots v_{i_3-1})^{12k}) > 4k$.
Now, suppose by absurdo that $\lag_{p,q}(\phi(v)) \leq 4k$.
Let $S_v =  (w_1,y_1, \ldots, w_n,y_n,w_{n+1})$.
By definition of $S_v$, there exists $1 \leq d \leq n$ such that $y_d = v_{i_2} \ldots v_{i_3-1}$, $w_1\ldots w_d = v_1\ldots v_{i_2-1}$, and $w_{d+1} \ldots w_{n+1} = v_{i_3} \ldots v_{C_{\tra}+1}$.
However, this causes a contradiction, as
\[
\begin{array}{lll}
\multicolumn{3}{l}{\lag_{p,q}(\phi(v))}\\
& \stackrel{\mathclap{\scriptscriptstyle\smash{(1)}}}{\geq} & \lag_{p,q}(w_1 w_2 \ldots w_{d-1} w_d y_d^{12k} w_{d+1} w_{d+2} \ldots w_nw_{n+1})\\
& = & \lag_{p,q}(v_1\ldots v_{i_2-1}(v_{i_2} \ldots v_{i_3-1})^{12k} v_{i_3}\ldots v_{C_{\tra}+1})\\
& \stackrel{\mathclap{\scriptscriptstyle\smash{(2)}}}{\geq} & \lag_{p_{i_2-1},q_{i_2-1}}(v_1\ldots v_{i_2-1}(v_{i_2} \ldots v_{i_3-1})^{12k})\\
& > & 4k,
\end{array}
\]
where inequality \ensuremath{(1)} follows from Corollary \ref{inequ'}.\ref{lag_|_}, and inequality \ensuremath{(2)} follows from the definition of the lag.
\end{proof}

\begin{corollary}\label{pr_3}
The function $\rho$ satisfies \prop{2}.
\end{corollary}

\begin{proof}
Let us state \prop{2} once again.
\begin{description}
\item[\prop{2}]: for every word $w$, for every pair $p,q \in \choice{2k}{\uni}(\rho(w))$, $\lag_{p,q}(w) \leq \oldbound$.
\end{description}
Let $w$ be a word, and let $p,q \in \choice{2k}{\uni}(\rho(w))$.
Then
\[
\begin{array}{lll}
\multicolumn{3}{l}{\lag_{p,q}(\rho(w))}\\
& \leq & \lag(r_{\tra}(\rho(w),p),r_{\uni}(\rho(w))) + \lag(r_{\uni}(\rho(w)),r_{\tra}(\rho(w),q))\\
& \leq & 2k + 2k = 4k.
\end{array}
\]
Suppose ab absurdo that $\lag_{p,q}(w) > \oldbound$.
Then, by the previous theorem and Corollary \ref{inequ'}.\ref{lag_r_p}, 
\[
\lag_{p,q}(\rho(w)) \geq \lag_{p,q}(\phi(w)) > 4k,
\]
which is a contradiction.
\end{proof}

%%%%%%%%%%%%%%%%%%%%%%%%%%%%%%%%%%%%%%%%%
\subsubsection{Property 3}\label{subsec:p3}
%%%%%%%%%%%%%%%%%%%%%%%%%%%%%%%%%%%%%%%%%

We can finally prove that $\rho$ satisfies \prop{3}.

\begin{lemma}\label{p4}
The function $\rho$ satisfies \prop{3}.
\end{lemma}

\begin{proof}
Let us state \prop{3} once again.
\begin{description}
\item[\prop{3}]: for every $w \in \Sigma^* $, for every $a \in \Sigma$, if $wa$ is a prefix of a word of $\dom(\tra)$, then there exist $q \in \choice{k}{\uni}(\rho(wa))$, $p \in \choice{k}{\uni}(\rho(w))$ and $s \in \Sigma^*$ such that $(p,a,s,q) \in \Delta_{\tra}$.
\end{description}

Let $S'_{wa} = (v_1 \ldots v_l,y_l, \ldots, v_n,y_n,v_{n+1})$.
If $l = n+1$, then $\rho(wa) = \rho(w)a$.
Since $\sigma_{\tra}(\rho(wa)) = \sigma_{\tra}(wa)$ by Corollary \ref{equ1'}, $\rho(wa)$ is also a prefix of some word in $\dom(\rel{\tra})$.
Hence, since $\uni$ sequentially $k$-uniformises $\tra$, there exists $q$ in $\choice{k}{\uni}(\rho(wa))$.
Let $p$ be the state preceding $q$ in the run $r_{\tra}(\rho(wa),q)$.
Then $p \in \choice{k}{\uni}(\rho(w))$, and there exists $s \in \Sigma^*$ such that $(p,a,s,q) \in \Delta_{\tra}$.

Now suppose that $l \leq n$.
In order to find a pair of states satisfying the lemma, we will study the behaviour of $\tra$ over the word $\rho(wa)y_l$.
By definition of $S'_{wa}$, $wa = v_1 \ldots v_ly_l$, and $y_l$ is an idempotent.
Since 
\[
\sigma_{\tra}(\rho(wa)y_l) = \sigma_{\tra}(way_l) = \sigma_{\tra}(v_1 \ldots v_ly_ly_l) = \sigma_{\tra}(v_1 \ldots v_ly_l) =  \sigma_{\tra}(wa),
\]
$\rho(wa)y_l$ is also a prefix of a word of $\dom(\rel{\tra})$, hence, as $\uni$ sequentially k-uniformises $\tra$, there exists $q''$ in $\choice{k}{\uni}(\rho(wa)y_l)$.
Let
\[
r_{\tra}(\rho(wa)y_l,q'') = q_0 \xrightarrow{\rho(w)|u_1} p \xrightarrow{a|u_{2}} q \xrightarrow{\phi'(wa)|u_{3}} q' \xrightarrow{y_l|u_{4}} q''.
\]
We show that Corollary \ref{run_loop} can be applied.
First, by Corollary \ref{equ1'}, $\sigma_{\tra}(\phi'(wa)) = \sigma_{\tra}(y_l)$, hence it is an idempotent.
Moreover, since $\sigma_{\tra}(\rho(w)a) = \sigma_{\tra}(wa)  = \sigma_{\tra}(v_1 \ldots v_ly_l)$, there exists a state $r$ such that $(r,q) \in \sigma_{\tra}(y_l) = \sigma_{\tra}(\phi'(wa))$, and $(q',q'') \in \sigma_{\tra}(y_l) = \sigma_{\tra}(\phi'(wa))$.
This implies that $q = q'$.
Therefore, as $q'' \in \choice{k}{\uni}(\rho(wa)y_l)$, $p \in \choice{k}{\uni}(\rho(w))$, $q = q' \in \choice{k}{\uni}(\rho(wa))$, and $(p,a,u_2,q) \in \Delta_{\tra}$.
\end{proof}

\subsubsection{Proof of Theorem~\ref{f_a_det}}

We can finally prove Theorem~\ref{f_a_det}. It is a direct consequence
of Lemma~\ref{lem:mainlem} that shows how to construct a
$2|Q_{\tra}|\oldbound$-uniformiser from $\rho$, of Corollary~\ref{coro:prop1} showing
that $\rho$ satisfies $\prop{1}$, of Corollary~\ref{pr_3} showing that
$\rho$ satisfies $\prop{2}$, and Theorem~\ref{p4} which finally shows
that $\rho$ satisfies $\prop{3}$.

%%% Local Variables:
%%% mode: latex
%%% TeX-master: t
%%% End:

%%% Local Variables:
%%% mode: latex
%%% TeX-master: t
%%% End:

\section{Details for Section~\ref{sec:drat}}

The largest part of this section covers the proof of the following theorem that was stated in the main part of the paper.

\dratunif*

We define and analyse the notion of transformation
sequence for an input word in Sections~\ref{subsec:transformation-sequences}--\ref{subsec:profiles-and-saturated-sequences}. These transformation sequences are the key information that is stored in the vertices of the game graph. In Section~\ref{subsec:drat-game-construction} we construct the game graph and show how to construct a sequential uniformiser from a winning strategy.
The most involved part of the proof, the construction of a winning strategy from a sequential uniformiser, is presented in Section~\ref{subsec:uniformisers-to-strategies}. The proof of Theorem~\ref{the:drat-unif} is finalised at the end of Section~\ref{subsec:uniformisers-to-strategies}.

In Section~\ref{subsec:drat-k-delay} we then give the proof of the second theorem stated in Section~\ref{sec:drat} on the bounded delay uniformisation.

In the formal definition of deterministic transducers, the endmarker
can occur anywhere in the input and output words. However, for the
semantics we only consider the endmarker at the end of the two
words. We assume from now on that a deterministic transducer rejects
whenever the endmarker occurs anywhere else than at the end of the
input or at the end of the output word. So whenever we consider words
over $\Sigmaend^*$ we rather refer to $\Sigma^* \cup \Sigma^* \cdot
\dashv$ because all other words in $\Sigmaend^*$ are not relevant.

We mainly work with words including the endmarker
(which is not included in $\rel{\tra}$). We therefore denote by
$R_\tra$ the relation $\{(u\dashv,v\dashv) \mid (u,v) \in
\rel{\tra}\}$. A sequential uniformiser for $R_\tra$ can easily be
transformed into a sequential uniformiser for $\rel{\tra}$ by shifting
the output that is produced on reading the endmarker into the output
function for the final states.

The following remark on the composition of applications of $\delta^*$
is a direct consequence of the definitions. However, it is used in
several proofs and therefore we state it explicitly.
\begin{lemma} \label{lem:delta-composition}
If $\delta^*(p,u_1,v_1) = (r, \varepsilon, \varepsilon)$ and 
 $\delta^*(r,u_2,v_2) = (q, \varepsilon, \varepsilon)$, then
 $\delta^*(p,u_1u_2,v_1v_2) = (q, \varepsilon, \varepsilon)$. 
\end{lemma}

%%%%%%%%%%%%%%%%%%%%%%%%%%%%%%%%%%%%%%%%%
\subsection{Transformation Sequences} \label{subsec:transformation-sequences}
%%%%%%%%%%%%%%%%%%%%%%%%%%%%%%%%%%%%%%%%%

Our goal is to build a game for $\tra$
such that a winning strategy for player Output corresponds to a
sequential uniformiser for $R_\tra$.
Wlog, we assume that $\tra$ is complete.
In the game, $\tra$
is simulated on the pairs of input and output that are produced by the
players, i.e., we apply $\delta^*$
to these pairs. The aim of Output is to reach a final state if Input
plays a word in the domain of $R_\tra$.
However, it might happen that the application of $\delta^*$
requires an output symbol but for deciding which output to produce
next, the strategy needs some lookahead on the input (the case that
the output is ahead of the input will not occur). Instead of storing
this lookahead on the input explicitly, we store an abstraction of it
that is sufficient to simulate $\delta^*$
once the next output symbols are provided. This is done by state
transformations induced by an input word.

We consider partial functions $\tau: \Qin \rightarrow Q$,
also called \emph{partial transformations}, and we write $\tau(q) = \bot$
if $\tau$ is not defined for $q \in \Qin$. The set of all partial
transformations (for the transducer $\tra$) is denoted by
$\Theta_\tra$. In the following, we only speak of transformations
instead of partial transformations.

The transformation $\tau_a$ for $a \in \Sigmaend$ is defined by the
transition function, namely, $\tau_a(q) = \delta(q,a)$ for each $q \in
\Qin$. 

We say that $\tau \in \Theta_\tra$
is \emph{consistent} with an input word $u \in \Sigmaend^*$
if for each $q\in \Qin$:
\begin{itemize}
\item  If $\tau(q) \not= \bot$, then $\delta^*(q,u,\varepsilon) = (\tau(q), \varepsilon, \varepsilon)$.
\item If $\delta^*(q,u,\varepsilon) = (p,\varepsilon,\varepsilon)$ with $p \in \Qin$, then $\tau(q) \not= \bot$ (and thus $\tau(q) = p$ by the first condition).
\end{itemize}
This means that the defined part of $\tau$ is consistent with the
state transformation induced by $u$, and that $\tau(q)$ can only be
undefined if reading $u$ from $q$ leads to some output state. Note
that each input letter $a$ is consistent with $\tau_a$, and that it remains
consistent if we set some values with $\tau_a(q) \in \Qout$ to $\bot$
(both conditions of consistency are not affected by this operation).

We say that $\tau$
is \emph{maximal} if there is $q \in \Qin$ with $\tau(q) \in
\Qout$ (intuitively,  $\tau$ cannot be extended without reading an
output symbol). For $\tau_1,\tau_2 \in \Theta_\tra$ such that
$\tau_1$ is not maximal, we define the product $\tau_1 \circ
\tau_2$ as the composition of $\tau_1,\tau_2$, that is
\[
\tau_1 \circ
\tau_2 (q) = 
  \begin{cases}
    \bot \mbox{ if } \tau_1(q) = \bot \mbox{ or }
    \tau_2(\tau_1(q)) = \bot \\
    \tau_2(\tau_1(q)) \mbox{ otherwise}.
  \end{cases}
\]

We extend the above definitions to finite
sequences of transformations.
For $\rho_1,\rho_2 \in \Theta_\tra^*$ define
\[
\rho_1 \circ \rho_2 =
\begin{cases}
\rho_1 \mbox { if $\rho_2 = \varepsilon$} \\  
\rho_2 \mbox { if $\rho_1 = \varepsilon$} \\  
\rho_1 \rho_2 \mbox{ if the last transformation in $\rho_1$ is maximal} \\
\rho_1' (\tau_1 \circ \tau_2) \rho_2' \mbox{ if $\rho_1 =
  \rho_1'\tau_1$, $\rho_2 =
  \tau_2\rho_2'$ with $\tau_1,\tau_2 \in \Theta_T$ and $\tau_1$ not
  maximal} \\
\end{cases}
\]
By a case distinction, one can show that this operation is
associative (it is a mix of concatenation and composition of
functions, which are both associative).
For $u \in \Sigmaend^*$
with $u = a_1 \cdots a_n$,
we define the transformation sequence of $u$
as $\rho_u = \tau_{a_1} \circ \cdots \circ \tau_{a_n}$.

Note that one has to distinguish between the notations
$\rho = \tau_1 \cdots \tau_n$
and $\rho = \tau_1 \circ \cdots \circ \tau_n$.
The first notation means that $\rho$
is the sequence consisting of the transformations
$\tau_1, \ldots, \tau_n$,
and the second notation means that $\rho$
is the sequence obtained from the transformations
$\tau_1, \ldots, \tau_n$
by composing them with $\circ$.
These are the same if, and only if, $\tau_1, \ldots, \tau_{n-1}$
are maximal.

We say that $\rho = \tau_1 \cdots \tau_n$
is \emph{consistent} with an input word $u \in \Sigmaend^*$
if there are words $u_1, \ldots, u_n$
such that $u = u_1 \cdots u_n$
and each $\tau_i$
is consistent with $u_i$.
The empty sequence of transformations is defined to be consistent only
with $\varepsilon$. As for single letters, $u$ is consistent with
$\rho_u$ but also with all sequences obtained from $\rho$ by setting
some values that are output states to $\bot$.

We call $\rho$
\emph{reduced} if $\tau_1, \ldots, \tau_{n-1}$
are maximal (if $\rho$
is not reduced, then a non-maximal $\tau_i$
could be merged with $\tau_{i+1}$
by $\circ$).
Note that $\rho \circ \tau$ is reduced if $\rho$ is reduced.

%Given a transformation sequence $\rho$ and state in $q \in Q$, we can
%apply $\rho$ to $q$ until a state in $\Qout$ is reached. Formally, we
%can view $\rho$ as a mapping $\rho: Q \rightarrow Q \times
%\Theta_\tra^*$ defined by
%\[
%\rho(q) = 
%\begin{cases}
%  (q,\varepsilon) \mbox{ if $\rho$ is the empty sequence} \\
%  (q,\rho) \mbox{ if $q \in \Qout$} \\
%  \rho'(\tau(q)) \mbox{ if $\rho = \tau\rho'$, $q \in \Qin$, and $\tau(q) \not= \bot$} \\
%  \bot \mbox{ otherwise}.
%\end{cases}
%\]
%Since $\rho$ encodes possible behaviours of $\tra$ for the words in $\Lin(\rho)$.

%%%%%%%%%%%%%%%%%%%%%%%%%%%%%%%%%%%%%%%%%
\subsection{Languages and Relations of Transformation Sequences}
\label{subsec:langauges-and-relations-of-transformation-sequences}
%%%%%%%%%%%%%%%%%%%%%%%%%%%%%%%%%%%%%%%%%
Using the consistency notion, a transformation sequence
$\rho \in \Theta_\tra^*$ defines a set of inputs
\[
\Lin(\rho) = \{u \in \Sigmaend^* \mid \rho \mbox{ is consistent with } u\}.
\]
Note that by the consistency definition for the empty sequence, we obtain $\Lin(\varepsilon) = \{\varepsilon\}$.

%The following lemma makes the intuition precise that applying $\rho$ to a state $q$ corresponds to applying $\delta^*$ from $q$ on an input sequence in $\Lin(\rho)$.
%\begin{lemma} \label{lem:simulate-transformation-sequence}
%  If $w \in \Lin(\rho)$ and $\rho(q) = (p,\xi)$ is defined, then $\delta^*(q,w,\varepsilon) = (p,u,\varepsilon)$ with $u \in \Lin(\xi)$.
%\end{lemma}
%\begin{proof}
%We give the proof along the first three cases used in the definition of $\rho(q)$ (the last case being irrelevant since we assume that $\rho(q)$ is defined.
%
%If $\rho = \varepsilon$, then $w = \varepsilon$ and the claim is obviously true. If $q \in \Qout$, then $\delta^*(q,w,\varepsilon) = (q,w,\varepsilon)$ and the claim is again trivially true. 
%
%If $q \in \Qin$ and $\rho = \tau\rho'$ with $\tau(q) = q' \not= \bot$, then $w = xy$ such that $x$ is consistent with $\tau$ and $y$ is consistent with $\rho'$. This means that $\delta^*(q,xy,\varepsilon) = \delta^*(q',y,\varepsilon)$, and $y \in \Lin(\rho')$. Now the claim follows by induction on the length of the transformation sequence. 
%\end{proof}
%
%
% Output language
%

Keeping in mind that $\rho$ encodes information on a lookahead on the
input word, we also define a set of output words that can be used
to ``catch up'' this lookahead. For this definition we also specify a
starting and a target state. For $\tau \in \Theta_\tra$, and $p,q
\in Q$, let
\[
\Lout(p,\tau,q) = \{v \in \Sigmaend^* \mid \exists r \in \Qin:\;
\delta^*(p,\varepsilon,v) = (r,\varepsilon,\varepsilon) \mbox{ and }
\tau(r) = q \}.
\]
We inductively extend this to sequences
of transformations:
\[
\Lout(p,\rho_1 \circ \rho_2,q) = \bigcup_{r \in Q} \Lout(p,\rho_1,r) \cdot \Lout(r,\rho_2,q).
\]
For the empty transformation sequence we let
\[
\Lout(p,\varepsilon,q) =
\begin{cases}
 \{\varepsilon\} & \mbox{if } p=q \\
 \emptyset & \mbox{otherwise}. 
\end{cases}
\]
The language without a specific target state is
\[
\Lout(p,\rho) = \bigcup_{q \in Q} \Lout(p,\rho,q).
\]
We refer to $\Lout(p,\rho)$
as the \emph{output language of $\rho$},
and for a word $v \in \Lout(p,\rho)$,
we say that it \emph{traverses} $\rho$.

The following properties are direct consequences of the definition of the output language, and the determinism of
$\tra$ (including its partition into input and output states).
%------------------------------
\begin{lemma}\label{lem:properties-lout}
  \begin{enumerate}
%  \item If $\varepsilon \in \Lout(p,\rho)$, then $\rho=\varepsilon$ or $p \in \Qin$. 
  \item If $v \in \Lout(p,\rho,q)$
    and $v \in \Lout(p,\rho,q')$, then $q=q'$.
  \item If $v \in \Lout(p,\rho)$,
    then no proper prefix of $v$ is in $\Lout(p,\rho)$.
  \end{enumerate}
\end{lemma}
%------------------------------

The next lemma states that $\rho$
encodes enough information to simulate $\tra$
on output words in the output language of $\rho$ from $p$.
\begin{lemma} \label{lem:Lin-Lout}
If $u \in \Lin(\rho)$ and $v \in \Lout(p,\rho,q)$, then
$\delta^*(p,u,v) = (q,\varepsilon,\varepsilon)$.
\end{lemma}
\begin{proof}
  We show the claim by induction on the length of $\rho$. If $\rho = \varepsilon$, then $u = v = \varepsilon$ and $p = q$. 
 
  If $\rho = \tau \in \Theta_\tra$,
  then $v \in \Lout(p,\rho,q)$
  means that
  $\delta^*(p,\varepsilon, v) = (r, \varepsilon, \varepsilon)$
  and $\tau(r) = q$.
  Furthermore, $\tau$
  is consistent with $u$,
  that is,
  $\delta^*(r,u,\varepsilon) = (q, \varepsilon, \varepsilon)$.
  In combination (using Lemma~\ref{lem:delta-composition}) we obtain
  $\delta^*(p,u,v) = (q, \varepsilon, \varepsilon)$.

  If $\rho = \tau \rho'$,
  then $u = u_1u'$
  such that $\tau$
  is consistent with $u_1$,
  and $v = v_1v'$ such that $v_1 \in \Lout(p,\tau,r)$ for some state
  $r$. From the base case for sequences of length $1$ we obtain
  $\delta^*(p,u_1,v_1) = (r, \varepsilon, \varepsilon)$. We conclude by
  induction since $u' \in \Lin(\rho')$ and $v' \in \Lout(r,\rho,q)$. 
\end{proof}

The nodes of the game graph that we construct later, consists of a
state together with a sequence of transformations, encoding the
lookahead on the input. For a state $p$
and a transformation sequence $\rho \in \Theta_\tra^*$,
we define the relation $R_p^\rho$
as those pairs of words whose output starts with a prefix that
traverses $\rho$
from $p$ to some state $q$,
and the remaining pair is accepted by $\tra$
from $q$. Formally,
\[
  R_p^\rho = \{(x,yz) \in \Sigmaend^* \times \Sigmaend^* \mid \exists q
  \in Q:\; y \in \Lout(p,\rho,q) \mbox{ and } \delta^*(q,x,z) = (r,\varepsilon,\varepsilon) \mbox{ with } r \in F\}.
\]
Note that, while in $R_\tra$ all words have to end with the endmarker, it is possible that $\varepsilon$ is in the domain or image of $R_p^\rho$. For example, if $p \in F$, and $\rho = \varepsilon$, then $R_p^\rho = \{(\varepsilon,\varepsilon)\}$.

The following Lemma states the connection between $R_p^\rho$ and $R_\tra$.
%------------------------------
\begin{lemma} \label{lem:lookahead-relation}
If $\delta^*(q_0,u,v) = (p,w,\varepsilon)$, $w \in \Lin(\rho)$, and
$(u',v') \in R_p^\rho$, then $(uu',vv') \in R_\tra$. 
\end{lemma}
%------------------------------
\begin{proof}
 First of all, $\delta^*(q_0,u,v) = (p,w,\varepsilon)$
  implies that $u = u_1w$
  and $\delta^*(q_0,u_1,v) = (p,\varepsilon,\varepsilon)$.

  Furthermore, since $(u',v') \in R_p^\rho$,
  we know that $v' = yz$
  with $y \in \Lout(p,\rho,q)$
  and 
  $\delta^*(q,u',z) = (r,\varepsilon,\varepsilon)$
  with $r \in F$.
  From $y \in \Lout(p,\rho,q)$
  and $w \in \Lin(\rho)$
  with Lemma~\ref{lem:Lin-Lout} we obtain
  $\delta^*(p,w,y) = (q,\varepsilon,\varepsilon)$.

  Combining  $\delta^*(q_0,u_1,v) = (p,\varepsilon,\varepsilon)$,
  $\delta^*(p,w,y) = (q,\varepsilon,\varepsilon)$, and
  $\delta^*(q,u',z) = (r,\varepsilon,\varepsilon)$ using
  Lemma~\ref{lem:delta-composition}, we obtain 
  $\delta^*(q_0,u_1wu',vyz) = (r,\varepsilon,\varepsilon)$ with $r \in
  F$, hence $(uu',vv') =  (u_1wu',vyz) \in R_\tra$.
\end{proof}
%------------------------------

The idea of our game construction can be illustrated using the
statement of Lemma~\ref{lem:lookahead-relation}. Assume that the
players Input and Output have already played the pair $(u,v)$
of words with $\delta^*(q_0,u,v) = (p,w,\varepsilon)$
as in Lemma~\ref{lem:lookahead-relation}. Then the current node of the
game is of the form $(p,\rho)$,
where $w \in \Lin(\rho)$.
The statement of Lemma~\ref{lem:lookahead-relation} now means that if
player Output ensures that the pair $(u',v')$
of words from the remaining play is in $R_p^\rho$, then she wins 
because $(uu',vv') \in R_\tra$.

%%%%%%%%%%%%%%%%%%%%%%%%%%%%%%%%%%%%%%%%%
\subsection{Reduction of Transformation Sequences} \label{subsec:reduction-transformation-sequences}
%%%%%%%%%%%%%%%%%%%%%%%%%%%%%%%%%%%%%%%%%

Transformation sequences encode information about the lookahead on the
input.  To keep this information bounded, we apply an operation to
specific sequences that corresponds to removing all paths that
require output. Formally, for $\rho \in \Theta_\tra^+$
with $\rho = \tau_1 \cdots \tau_n$,
we define its \emph{reduction to input paths}
$\reduce{\rho} \in \Theta_\tra$
by removing all intermediate output states, which formally is
$\reduce{\rho} = \tau_1' \circ \cdots \circ \tau_{n-1}' \circ \tau_n$,
where for $i \in \{1,\ldots,n-1\}$:
\[
\tau_i'(q) =
\begin{cases}
 \bot \mbox{ if } \tau_i(q) \in \Qout \\
 \tau_i(q) \mbox{ otherwise}
\end{cases}
\]
Note that $\reduce{\rho}$
is a single transformation because all output states have been removed
in the first $n-1$
transformations, and then all are merged into one transformation by
$\circ$.

From the definition of consistency it is clear that $\Lin(\rho) \subseteq \Lin(\reduce{\rho})$, while the output languages can only decrease by setting some values to undefined, that is, $\Lout(p,\reduce{\rho},q) \subseteq \Lout(p,\rho,q)$. However, since $\reduce{\rho}$ only removes intermediate output states, the empty word cannot be removed from the output languages, as stated in the lemma below. 
%------------------------------
\begin{lemma} \label{lem:empty-word-traversal}
  Let $\rho \in \Theta_\tra^+$
  and $p,q \in Q$.
  Then $\varepsilon \in \Lout(p,\rho,q)$
  iff $\varepsilon \in \Lout(p,\reduce{\rho},q)$.
\end{lemma}
%------------------------------
\begin{proof}
  First note that if the output language of $\rho$ from $p$ contains
  $\varepsilon$, then $p \in \Qin$. We show the claim by induction on the length of $\rho$.

  If  $\rho = \varepsilon$, the claim is obviously true.

  If $\rho = \tau\rho'$,
  then $\varepsilon \in \Lout(p,\rho,q)$
  iff there is a state $r \in \Qin$
  such that $\varepsilon \in \Lout(p,\tau,r)$ and $\varepsilon \in \Lout(r,\rho',q)$.
  Then $\tau(p) = r$ is not changed in $\reduce{\rho}$, and hence the claim follows by induction.
\end{proof}
%------------------------------

%We use the operation $\reduce{\rho}$
%on parts of sequences to keep the number of vertices in the game graph
%(which is to be defined later) bounded. If a sequence becomes too
%long, we replace a suffix $\rho$
%of it by the single transformation $\reduce{\rho}$.
%The following lemma is technical tool for this operation. It basically
%states that an output word that traverses $\rho_1$
%and $\rho_1\rho_2$ from $p$ also traverses $\rho_1\reduce{\rho_2}$.
%%------------------------------
%\begin{lemma}\label{lem:output-suffix-epsilon}
%  Let $\rho_1,\rho_2 \in \Theta_\tra^+$,
%  $p,q \in Q$,
%  and $v \in \Sigma^*$.
%  If $v \in \Lout(p,\rho_1)$
%  and $v \in \Lout(p,\rho_1\rho_2,q)$,
%  then $v \in \Lout(p,\rho_1\reduce{\rho_2},q)$.
%\end{lemma}
%%------------------------------
%\begin{proof}
%  Since $v \in \Lout(p,\rho_1)$,
%  there is a unique state $r$
%  with $v \in \Lout(p,\rho_1,r)$
%  (see Lemma~\ref{lem:properties-lout}). Because
%  $v \in \Lout(p,\rho_1\rho_2,q)$,
%  we have $v \in \Lout(p,\rho_1,r) \Lout(r,\rho_2,q)$.
%  Again by Lemma~\ref{lem:properties-lout}, there is no prefix of
%  $v$
%  in $\Lout(p,\rho_1,r)$.
%  Thus, $\varepsilon \in \Lout(r,\rho_2,q)$
%  and therefore also $\varepsilon \in \Lout(r,\reduce{\rho_2},q)$
%  (see Lemma~\ref{lem:empty-word-traversal}).  This implies
%  $v \in \Lout(p,\rho_1\reduce{\rho_2},q)$. 
%\end{proof}
%%------------------------------

%%%%%%%%%%%%%%%%%%%%%%%%%%%%%%%%%%%%%%%%%
\subsection{Profiles and Saturated Sequences}
\label{subsec:profiles-and-saturated-sequences}
%%%%%%%%%%%%%%%%%%%%%%%%%%%%%%%%%%%%%%%%%

An important tool in this analysis is an abstraction of
transformation sequences $\rho$ into their profiles $P_\rho$. This
abstraction basically contains the information which of the languages
$\Lout(p,\rho,q)$ are nonempty, and which of them contain
$\varepsilon$. Intuitively, this abstraction is useful in a
uniformisation setting because a uniformiser does not need to know the
exact language $\Lout(p,\rho,q)$, it just needs to know some word
in $\Lout(p,\rho,q)$ that it can produce. The special case of
$\varepsilon$ is interesting because it means that a uniformiser does
not have to produce any output for traversing $\rho$ from $p$ to $q$.

Formally, the profile of $\rho$ is of the form $P_\rho \subseteq (Q
\times \{\varepsilon,+\} \times Q)$ with 
\begin{itemize}
\item $(p,\varepsilon,q) \in P_\rho$ if $\varepsilon \in
  \Lout(p,\rho,q)$,
\item $(p,+,q) \in P_\rho$
  if $\Lout(p,\rho,q) \not= \emptyset$
  and $\varepsilon \notin \Lout(p,\rho,q)$.
\end{itemize}

Since
$\Lout(p,\rho_1\rho_2,q) = \bigcup_{r \in Q} \Lout(p,\rho_1,r)
\Lout(r,\rho_2,q)$, the profile
$P_{\rho_1\rho_2}$
is uniquely determined from $P_{\rho_1}$
and $P_{\rho_2}$.
We thus can define the multiplication
$P_{\rho_1}P_{\rho_2} = P_{\rho_1\rho_2}$. As usual, a profile $P$ is
called \emph{idempotent} if $PP = P$, and 
 a transformation sequence $\rho$ is
\emph{idempotent} if its profile $P_{\rho}$ is idempotent.

We state a few simple properties of profiles that are useful in later
proofs.
\begin{lemma} \label{lem:profile-properties}
 Let $P =  P_\rho$ be the profile of some transformation sequence
 $\rho$. 
  \begin{enumerate}
  \item If $(q_1,\varepsilon,q_2),  (q_1,\varepsilon,q_3) \in P$, then
    $q_2 = q_3$.
  \item  If $P$ is idempotent and
    $(q_1,\varepsilon,q_2),(q_2,\varepsilon,q_3) \in P$, then $q_2 =
    q_3$. 
  \end{enumerate}
\end{lemma}
\begin{proof}
  The first claim follows from Lemma~\ref{lem:properties-lout}(1)
  for $v = \varepsilon$.
  
  For the second claim, $P$ being idempotent implies
  $(q_1,\varepsilon,q_3) \in P$. Now the first claim implies $q_2 = q_3$.
\end{proof}

The notion of profile of a transformation sequence can be extended to
a profile of an input word $u$ by letting $P_u = P_{\rho_u}$, that is,
$P_u$ is the profile of the full transformation sequence corresponding
to $u$.

Later, we show how to obtain a strategy in the game that we construct
from a sequential uniformiser of the relation. In this
strategy we apply the operation $\reduce$
to infixes of transformation sequences whose profiles satisfy certain
properties. To ensure the existence of such an infix, we define the
notion of saturated sequence. Basically, saturated means that the
transformation sequence contains a nontrivial, idempotent infix whose
profile is furthermore absorbed by the profile of the prefix (see
below for the formal definition). However, when constructing the
strategy, we do not just work with transformation sequences, but the
strategy keeps in memory the input word that led to the transformation
sequence. The idempotent infix of the transformation sequence then
corresponds to an infix of this input word. The definition of saturated
sequence below requires that the profile of this infix of the input word
should also be idempotent and absorbed by the prefix. The profile of
the input word is captured in the definition below by the additional
sequence of profiles.

A transformation sequence $\rho$
is called \emph{saturated} if for each sequence of profiles
$\bar{P} = P_1,\cdots, P_{|\rho|}$ (of the same length as $\rho$),
it is possible to split $\rho$
as $\rho = \rho_1\rho_2\rho_3$,
such that the following properties are satisfied.
\begin{itemize}
\item $\rho_1 \not= \varepsilon$
\item $\rho_2$ is non-trivial and idempotent, and $P_{\rho_1} =
  P_{\rho_1}P_{\rho_2}$
\item Define the profiles $\hat{P}_1 = P_1 \cdots P_{|\rho_1|}$ and
  $\hat{P}_2 = P_{|\rho_1|+1} \cdots P_{|\rho_2|}$ as the products of
  the profiles corresponding to $\rho_1$ and $\rho_2$,
  respectively. Then $\hat{P}_1 \hat{P}_2 = \hat{P}_1$, and
  $\hat{P}_2\hat{P}_2 = \hat{P}_2$. 
\end{itemize}
We refer to the splitting $\rho_1\rho_2\rho_3$
as saturation witness for $\rho$ and $\bar{P}$.

The bound from the following lemma is used to bound the length of
transformation sequences used in the game graph.

%------------------------------
\begin{lemma} \label{lem:length-saturated}
  There is a number $K$
  such that each reduced transformation sequence $\rho$
  of length at least $K$ is saturated. This number $K$ is computable from $\tra$.
\end{lemma}
%------------------------------
\begin{proof}
  Let $\rho = \tau_1 \cdots \tau_n$. We take the set $\{1,\ldots,n+1\}$ as nodes of a finite complete graph, and colour the edges $\{i,j\}$ with $i < j$ by the pair of profiles $(P_{\tau_i \cdots \tau_{j-1}},P_i \cdots P_{j-1})$. Ramsey's theorem yields that there is a number $K$ (that is computable) such that for $n \ge K$, there are three positions $h < i < j$ such that all edges are assigned the same pair. We can further assume that $i+1 < j$.  Choosing $\rho_2 = \tau_i \cdots \tau_{j-1}$ yields the saturation witness. 
\end{proof}
%------------------------------

%%%%%%%%%%%%%%%%%%%%%%%%%%%%%%%%%%%%%%%%%
\subsection{Game Construction} \label{subsec:drat-game-construction}
%%%%%%%%%%%%%%%%%%%%%%%%%%%%%%%%%%%%%%%%%

We now have all the ingredients for defining the uniformisation game. 
As mentioned earlier, the basic idea is that player Input plays an
input sequence and player Output plays an output sequence, such that
if Input plays a word $u$ in the domain of $R_\tra$, then Output has to
produce a word $v$ such that $(u,v) \in R_\tra$. For checking the
condition on the domain of $R_\tra$, we use a DFA $A_\dom =
(Q_\dom,\Sigma,q_0^\dom,\delta_\dom,F_\dom)$ that recognises the
domain of $R_\tra$, and which is simulated on the played input sequence in the game.
We can furthermore safely assume that two input words $u_1,u_2$ with the same profile $P_{u_1} = P_{u_2}$ also induce the same state transformation on $A_\dom$. This can always be ensured by taking the product of $A_\dom$ with the transducer $\tra$, and considering the profiles w.r.t.\ this product transducer. These profiles then also encode the state transformations of $A_\dom$.

The other components of the vertices are of the form $(p,\rho)$
where $p$
is a state of $\tra$,
and $\rho$
is a transformation sequence of length at most $2K+1$
with $K$
as in Lemma~\ref{lem:length-saturated} (the reason for choosing $2K+1$ becomes clear when we construct a winning strategy for Output from a uniformising sequential
transducer; see Lemma~\ref{lem:drat-uniformiser-to-strategy}).
The vertices of player Output additionally encode the last input
letter played by Input.

Formally, the uniformisation game $\game_\tra$ for $\tra$ has the
following components (with $K$ as in Lemma~\ref{lem:length-saturated}):
\begin{itemize}
\item $\Vin = Q \times \{\rho \in \Theta_\tra^* \mid
  |\rho| \le 2K+1\}  \times Q_\dom$
\item $\Vout = \Vin \times \Sigmaend$
\item The initial vertex is $(q_0,\varepsilon,q_0^\dom)$
\item The edges of the game graph are annotated with input and output words, respectively, which are for later reference when transforming strategies into sequential uniformisers and vice versa. For this purpose, let $v_{p,\rho,q}$ be a shortest word in $\Lout(p,\rho,q)$ for all $p,q \in Q$ and all transformation sequences $\rho$ such that $\Lout(p,\rho,q) \not= \emptyset$ (if $\Lout(p,\rho,q) = \emptyset$, then $v_{p,\rho,q}$ is undefined and cannot be used for the moves defined below).

We then have the following edges (the names of the moves are for later reference in the proofs):
  
  \begin{itemize}

  \item[(In)] $(p,\rho,d) \xrightarrow{a} (p,\rho,d,a)$ (Input chooses the next letter)
  \item[(Out0)] $(p,\varepsilon,d,a) \xrightarrow{v_{p,\tau_a,q}}
    (q,\varepsilon,\delta_\dom(d,a))$  (produce output matching the next input symbol and simulate $\tra$; only possible if there is no lookahead)
  \item[(Out1)] $(p,\rho,d,a) \xrightarrow{\varepsilon}
    (p,\rho \circ \tau_a,\delta_\dom(d,a))$ if $|\rho \circ \tau_a|
    \le 2K+1$ (increase the lookahead on the input)
  \item[(Out2)] $(p,\rho,d,a) \xrightarrow{v_{p,\rho_1,q}} (q,\rho_2,d,a)$ if $\rho= \rho_1\rho_2$ with $|\rho_1| \ge 1|$ (produce output that consumes a prefix $\rho_1$ of the current lookahead)
  \item[(Out3)]
    $(p,\rho,d,a) \xrightarrow{\varepsilon}
    (p,\rho_1\reduce{\rho_2}\rho_3,d,a)$ if
    $\rho = \rho_1\rho_2\rho_3$
    with $|\rho_2| \ge 2$ (reduce the information in the lookahead)
  \end{itemize}
  Note that the target of the last two moves is again a vertex of
  player Output. However, each of these moves strictly reduces the
  length of the transformation sequence, which means that there are
  only finitely many such moves before a vertex of player Input is
  reached. 

  Also note that Output can always move because at least one of (Out1) or (Out3) is possible.
\item The winning condition of Output is a safety condition. The set
  of bad vertices (to be avoided by Output) are those of the form $(p,\rho,d)$ such that $d \in F_\dom$ but ($p \notin F$ or $\rho \not= \varepsilon$).

Note that $d \in F_\dom$ is only reached after the endmarker on the
input. The condition says that then there is no more lookahead and the
pair $(u,v)$ of played input and output words is accepted by $\tra$.  
\end{itemize}

The following lemma formally states that a play simulates $A_\dom$ on the played input word and $\tra$ on the pair of played input and output words, and that the transformation sequence encodes the lookahead on the input.
\begin{lemma} \label{lem:game-property}
Assume that the players have reached a vertex $(p,\rho,d) \in \Vin$ by  moves corresponding to words $u$ as input and $v$ as output. Then $A_\dom:  q_0^\dom \xrightarrow{u} d$, and furthermore $u=xw$ with $\delta^*(q_0,x,v) = (p,\varepsilon,\varepsilon)$ and $w \in \Lin(\rho)$.
\end{lemma}
\begin{proof}
The property $A_\dom:  q_0^\dom \xrightarrow{u} d$ is obvious since $A_\dom$ is simulated on the input symbols in the game construction.

The property $u=xw$ with $\delta^*(q_0,x,v) = (p,\varepsilon,\varepsilon)$ and $w \in \Lin(\rho)$ can be shown inductively by the number of moves. It certainly holds at the initial vertex with $u=v=\varepsilon$. Assume that the property is true for the words $u$ and $v$ at the vertex $(p,\rho,d) \in \Vin$. Consider the next input symbol $a$ played by Input. The play moves to $(p,\rho,d,a)$, and we consider the different types of moves that are available.

For a move (Out0) with word $v_{p,\tau_a,q}$, we have $\rho = \varepsilon$ and thus $w = \varepsilon$. The play moves to $(q,\varepsilon,\delta_\dom(d,a))$. We have $\delta^*(q_0,ua,vv_{p,\tau_a,q}) = \delta^*(p,a,v_{p,\tau_a,q}) = (q, \varepsilon, \varepsilon)$. The first equality is by assumption, and the second equality follows from Lemma~\ref{lem:Lin-Lout} with the fact that $a \in \Lin(\tau_a)$.

The move (Out1) takes the play into $(p,\rho \circ \tau_a,\delta_\dom(d,a))$. Since $w \in \Lin(\rho)$, we obtain that $wa \in \Lin(\rho \circ \tau_a)$.

For (Out2) and (Out3) the moves lead to a vertex in $\Vout$. We show that the claimed property is preserved on the first three components (ignoring the $a$ until an (Out0) or (Out1) move is played).

The move (Out2) produces some output $v_{p,\rho_1,q}$ and leads to $(q,\rho_2,d,a)$ for
$\rho = \rho_1\rho_2$. We can split $w = w_1w_2$ with $w_1\in \Lin(\rho_1)$ and $w_2 \in \Lin(\rho_2)$. Since $v_{p,\rho_1,q} \in \Lout(p,\rho_1,q)$, we obtain by Lemma~\ref{lem:Lin-Lout} that $\delta^*(p,w_1,v_{p,\rho_1,q}) = (q,\varepsilon,\varepsilon)$. Hence, we have $u = xw_1w_2$ with $\delta^*(q_0,xw_1,vv_{p,\rho_1,q}) = (q,\varepsilon,\varepsilon)$ and $w_2 \in \Lin(\rho_2)$.

The move (Out3) leads to $(p,\rho',d,a)$ with $\rho' = \rho_1\reduce{\rho_2}\rho_3$. Since $\Lin(\rho_2) \subseteq \Lin(\reduce{\rho_2})$, we conclude that $w \in \Lin(\rho')$.
\end{proof}

\begin{lemma} \label{lem:drat-strategy-to-uniformiser}
If Output has a winning strategy in $\game_\tra$, then $R_\tra$ can be
uniformised by a sequential transducer. 
\end{lemma}
\begin{proof}
 Since $\game_\tra$ is a safety game, there is a positional winning
 strategy for Output. We build a sequential transducer $S$ with $\Vin$ as state set. The transition function is derived from the winning
 strategy: Let $s=(p,\rho,d) \in \Vin$ be a state of $S$ and $a \in \Sigmaend$.
 
The successor vertex in $\game_\tra$ is $(p,\rho,d,a)$,
 and the winning strategy describes a finite sequence of Output moves that
 ends up in a vertex of the form $(p',\rho',d')$ with $d' = \delta_\dom(d,a)$. The
 sequential transducer moves to this state
 $(p',\rho',d')$ and outputs the word obtained along
 this finite sequence of moves. %If $d' \in F$, then by the winning condition, there is a final state $q$ of $\tra$ such that $\Lout(p',\rho',q) \not= \emptyset$. Then $S$ appends a word in $\Lout(p',\rho',q)$ to its output. 

 From Lemma~\ref{lem:game-property} follows that this defines a
 uniformiser of $R_\tra$: Consider a word $u$ in the domain of $R_\tra$, which uniquely determines the moves of Player Input. The strategy for Player Output generates moves inducing a word $v$, such that the play ends up in a vertex $(p,\rho,d)$ with $d \in F_\dom$. Since the strategy is winning, $p \in F$ and $\rho = \varepsilon$.
 Lemma~\ref{lem:game-property} yields that $u =xw$ with $\delta^*(q_0,u,v) = (p,w,\varepsilon)$ and $w \in \Lin(\rho)$. Since $\rho = \varepsilon$, we get $w = \varepsilon$ and conclude that $(u,v) \in R_\tra$.
\end{proof}

%%%%%%%%%%%%%%%%%%%%%%%%%%%%%%%%%%%%%%%%%
\subsection{From Uniformisers to Strategies}
\label{subsec:uniformisers-to-strategies}
%%%%%%%%%%%%%%%%%%%%%%%%%%%%%%%%%%%%%%%%%

We prove the existence of a winning strategy for Output in case that $R_\tra$
is uniformised by a sequential transducer. A note for the reader who
is not very familiar with this kind of game-theoretic reasoning: The
strategy that we construct below is not positional, it uses a memory of unbounded size. Furthermore, we do not even need to care whether we can compute the individual moves. Proving the existence of some strategy is sufficient and implies the existence of a positional strategy, as assumed in the proof of Lemma~\ref{lem:drat-strategy-to-uniformiser}.

\begin{lemma}\label{lem:drat-uniformiser-to-strategy}
  If $R_\tra$
  can be uniformised by a sequential transducer, then Output has a
  winning strategy in $\game_\tra$.
\end{lemma}

The rest of this section is devoted to the proof of
Lemma~\ref{lem:drat-uniformiser-to-strategy}, which requires some
definitions and other lemmas. For the construction of the strategy, we
split the transformation sequence in a game position in two parts. So
we consider game positions of the form $(p,\rho\xi,d)$ with $|\rho|
\le K+1$ and $|\xi| \le K$ (where $K$ is the number from
Lemma~\ref{lem:length-saturated} used in the game construction). The moves of type (Out3), which apply the operation $\reduce{\cdot}$ to an infix of the transformation sequence, are only applied to the second part $\xi$ of the transformation sequence. The first part $\rho$ remains fixed until it is consumed by an output move of type (Out2). 

The input word $u$ that generates the part $\xi$ of the transformation sequence is stored in the memory of the strategy, together with the information to which parts the operation $\reduce{\cdot}$ was applied (the operation $\reduce{\cdot}$ is applied to parts of the transformation sequence, but each such part corresponds to a part of $u$).
This leads to a structure of the form $(u,M)$ with a set $M$ of edges $(i,j)$, where $i < j$ are positions in $u$. These edges are well-nested (they do not cross). In the construction of the strategy, we need to refer to words that are obtained from $u$ by pumping the parts of $u$ enclosed by an edge.  Hence, we refer to $(u,M)$ as a \emph{pump-word}. 

Formally, pump-words $(u,M)$ and their corresponding transformation sequences $\rho_{(u,M)}$ are defined inductively as detailed below. 
Along the inductive definition we also define how to concatenate two pump-words.

%For these pumping operations we require that the edges $(i,j)$ are such that the infix $u[i,j]$ of $u$, and the infix of $\xi$ corresponding to $u[i,j]$ are idempotent and absorbed. This ensures that pumping the infixes corresponding to the edges does not change the profile of the word.
%
%Before we give the precise conditions that pump-words $(u,M)$ have to satisfy, we  define 

\begin{definition}\label{def:pump-word}
\begin{enumerate}
\item For each $u \in \Sigmaend^*$, $(u,\emptyset)$ is a pump-word. The transformation sequence of $(u,\emptyset)$ is $\rho_{(u,\emptyset)} := \rho_u$, the transformation sequence induced by $u$.
\item The concatenation of two pump-words $(u_1,M_1)$ and $(u_2,M_2)$ is defined in the expected way as $(u_1,M_1)(u_2,M_2) = (u_1u_2,M_1 \cup (M_2 + |u_1|))$, where $M_2 + |u_1|$ denotes the set of edges in $M_2$ shifted by $|u_1|$ to the right. For example $(aba,\{(1,2)\})(bbc,\{(1,3)\}) = (ababbc,\{(1,2),(4,6)\}$. 

%If $(u,M) = (u_1,M_1)(u_2,M_2)$, then the transformaiton sequence is defined as $\rho_{(u,M)} = \rho_{(u_1,M_1)} \circ  \rho_{(u_2,M_2)}$.

\item Let $(u,M)$ be a pump-word and let $\xi = \rho_{(u,M)}$ to simplify notation. Let $\xi = \xi_1\xi_2\xi_3$, where $\xi_1,\xi_2,\xi_3$ are transformation sequences and $|\xi_2| \ge 2$. Let $(u,M) = (u_1,M_1)(u_2,M_2)(u_3,M_3)$ be the corresponding split of $(u,M)$, that is, $\rho_{(u_i,M_i)} = \xi_i$ for all $i \in \{1,2,3\}$. 
If $P_{\xi_2}$ is idempotent and $P_{u_2}$ is idempotent, then $(u,M') = (u_1,M_1)(u_2,M_2 \cup \{(1,|u_2|)\})(u_3,M_3)$ is a pump-word, and $\rho_{(u,M')} := \xi_1\reduce{\xi_2}\xi_3$ is the transformation sequence of $(u,M')$.

Note that $M' = M \cup \{(|u_1|+1,|u_1|+|u_2|)\})$ according to the definition of concatenation of pump-words. \qed
\end{enumerate}
\end{definition}
In point~3 of Definition~\ref{def:pump-word}, we implicitly assume that the split $(u,M)) = (u_1,M_1)(u_2,M_2)(u_3,M_3)$ that is induced by the split $\xi = \xi_1\xi_2\xi_3$ does not ``cut'' any edge of $M$ (all edges of $u$ are inside $u_1$, $u_2$, or $u_3$). This follows from the fact that an infix $v$ of $u$ that is enclosed by an edge only contributes a single transformation $\reduce{\rho_v}$ to  $\rho_{(u,M)}$.

%Point 2 is only needed to define the concatenation of pump-words, which is used in point 3. 

As mentioned earlier, the edges in a pump-word mark factors that we want to pump. In the following, when we speak of a pumping factor, then we refer to an infix of $u$ that corresponds to an edge. For a pump-word $(u,M)$ and a number $k$, we define the set $\Pump{k}{(u,M)}$, which is obtained by repeating each pumping factor at least $k$ times:

\begin{itemize}
\item If $M = \emptyset$, then $\Pump{k}{(u,M)} = \{u\}$.
%\item If $(u,M) = (u_1,M_1)(u_2,M_2)$, then $\pump{k}{(u,M)} = \pump{k}{(u_1,M_1)}\pump{k}{(u_2,M_2)}$ and $\Pump{k}{(u,M)} = \Pump{k}{(u_1,M_1)}\Pump{k}{(u_2,M_2)}$.
\item If $(u,M) = (u_1,M_1)(u_2,M_2 \cup \{(1,|u_2|)\})(u_3,M_3)$ as in point 3 of Definition~\ref{def:pump-word}, then
 $\Pump{k}{(u,M)} = \bigcup_{\ell \ge k} \Pump{k}{(u_1,M_1)}(\Pump{k}{(u_2,M_2)})^\ell\Pump{k}{(u_3,M_3)}$ 
\end{itemize}
Clearly, $\Pump{k}{(u,M)}$ is a regular set for each $k$ and each $(u,M)$.

The condition on the profiles of pumping factors being idempotent, ensures that pumpings do not change profiles, as expressed in the following lemma.
\begin{lemma} \label{lem:pump-word-profile}
If $(u,M)$ is a pump-word, then $P_w = P_u$ for each word $w \in \Pump{1}{(u,M)}$.  
\end{lemma}
\begin{proof}
  The proof is by induction on the complexity of $(u,M)$, which is a mapping $\varphi: \nat \rightarrow \nat$ where $\varphi(n)$ is the number of edges in $M$ that are of nesting depth $n$. The nesting depth of an edge $(i,j)$ is defined as $1$ if there are no other edges inside $(i,j)$, and otherwise it is $n+1$ if the maximal nesting depth of an edge inside $(i,j)$ is $n$.
Note that $\varphi(n) > 0$ for only finitely many $n$.

We use a lexicographic ordering for comparing these mappings, letting $\varphi_1 < \varphi_2$ if $\varphi_1(n) < \varphi_2(n)$ for the biggest $n$ with $\varphi_1(n) \not= \varphi_2(n)$ (note that such an $n$ exists if $\varphi_1 \not= \varphi_2$, as there are only finitely many non-zero entries). This is a well-ordering \cite{Dickson13} and thus can be used for an induction.

For the induction base, if $\varphi$ maps everything to $0$, then $M = \emptyset$ and $w = u$. 

So consider the last edge that has been added to $M$ according to Definition~\ref{def:pump-word}. This means that $(u,M) = (u_1,M_1)(u_2,M_2 \cup \{1,|u_2|\})(u_3,M_3)$, where the edge enclosing $u_2$ is the last one that was added. Then $w = w_1w_{2,1}\cdots w_{2,\ell}w_3$ with $w_1 \in \Pump{k}{(u_1,M_1)}$, $w_{2,1},\cdots, w_{2,\ell} \in \Pump{k}{(u_2,M_2)}$, and $w_3 \in \Pump{k}{(u_3,M_3)}$. Note that $(u_1,M_1)(u_2,M_2)^\ell(u_3,M_3)$ is a pump-word, and its complexity is smaller than the one of $(u,M)$ because the edge enclosing $u_2$ is removed and the iteration of $(u_2,M_2)$ only increases the number of edges of smaller nesting depth.

Then by induction, $P_{w_1w_{2,1}\cdots w_{2,\ell}w_3} = P_{u_1(u_2)^\ell u_3}$, and since $P_{u_2}$ is idempotent, we obtain $P_{u_1(u_2)^\ell u_3} = P_{u_1u_2u_3} = P_u$ for $\ell \ge 1$.
\end{proof}

The pump-words that we build during the strategy construction apply point 3 of Definition~\ref{def:pump-word} to saturation witnesses (see Section~\ref{subsec:profiles-and-saturated-sequences}). So they satisfy some further properties that, intuitively, ensure that outputs matching pumped words can be replaced by outputs matching the unpumped words. We refer to such pump-words as \emph{safe}:
 
\begin{definition}\label{def:safe-pump-word}
A pump-word is called \emph{safe} if in Definition~\ref{def:pump-word}(3) the following additional conditions are satisfied (using the same notations as in Definition~\ref{def:pump-word}(3)): 
\begin{itemize}
\item $M_3 = \emptyset$, so there are no edges after the one that is newly introduced.
\item $P_{\xi_1} = P_{\xi_1}P_{\xi_2}$ ($\xi_2$ is absorbed by $\xi_1$).
\item $P_{u_1} = P_{u_1}P_{u_2}$ ($u_2$ is absorbed by $u_1$).
\end{itemize}
The connection of saturation witnesses (see Section~\ref{subsec:profiles-and-saturated-sequences}) and the condition for safe pump-words is as follows. Let $\xi = \tau_1 \cdots \tau_{|\xi|}$, and  let $(u,M) = (u_1',M_1') \cdots (u_{|\xi|}',M_{|\xi|}')$ be the corresponding split of $(u,M)$ (that is, $\rho_{(u_i',M_i')} = \tau_i$). Let $\bar{P} = P_{u_1},\cdots, P_{u_{|\xi|}}$. Then the conditions for safe pump-words implies that $\xi_1\xi_2\xi_3$ is a saturation witness for $\xi$ and $\bar{P}$. Since $(u,M)$ determines $\xi = \rho_{(u,M)}$, and $\bar{P}$, we call this a \emph{saturation witness for $(u,M)$}. \qed
\end{definition}
The reason for considering pump-words and safe pump-words (and not just defining the latter) is that in a decomposition of a safe pump-word $(u,M) = (u_1,M_1)(u_2,M_2)(u_3,M_3)$, the part $(u_2,M_2)$ needs not to be safe (the condition of safe pump-word refers to the prefixes before the pumping factors). However, $(u_2,M_2)$ is a pump-word because being idempotent is a local property that does not depend on the prefix.

The following lemma is essential for the strategy construction. It basically states that if an output word matches a large pumping of $(u,M)$, then there is also an output word matching $\rho_{(u,M)}$, that is, an output that is consumed outside the pumping factors of $(u,M)$.
\begin{lemma} \label{lem:output-large-pumping}
  Let $(u,M)$ be a safe pump-word, $p,q \in Q$, and $v \in \Sigmaend^*$ be an output word such that there is $w \in \Pump{3|v|}{(u,M)}$ with $v \in \Lout(p,\rho_w,q)$. Then $\Lout(p,\rho_{(u,M)},q) \not= \emptyset$.
\end{lemma}
\begin{proof}
  The proof is by induction on the complexity of $(u,M)$, as in the proof of Lemma~\ref{lem:pump-word-profile}. If $M = \emptyset$, and the claim obviously holds because $w = u$ (there are no pumping factors) and $\rho_u = \rho_{(u,\emptyset)}$.

So consider the last edge that has been added to $M$ according to Definition~\ref{def:safe-pump-word}. This means that $(u,M) = (u_1,M_1)(u_2,M_2 \cup \{1,|u_2|\})(u_3,M_3)$, where the edge enclosing $u_2$ is the last one that was added. And as in the definition of pump-words, let $\xi_i = \rho_{(u_i,M_i)}$. We want to show that $\Lout(p,\xi_1\reduce{\xi_2}\xi_3,q) \not = \emptyset$.

By definition of $\Pump{3|v|}{(u,M)}$, we can write $w = w_1 w_{2,1}\cdots w_{2,\ell} w_3$ with $\ell \ge 3|v|$, $w_1 \in \Pump{3|v|}{(u_1,M_1)}$, $w_3 \in \Pump{3|v|}{(u_3,M_3)}$, and $w_{2,i} \in \Pump{3|v|}{(u_2,M_2)}$. Note that $P_{w_1} = P_{u_1}$, $P_{w_{2,i}} = P_{u_2}$, and $P_{w_3} = P_{u_3}$ by Lemma~\ref{lem:pump-word-profile}.

According to the definition of $\Lout(p,\rho_w,q)$, we can write $v = v_1v_{2,1}\cdots  v_{2,\ell} v_3$ with $v_1 \in \Lout(p,\rho_{w_1},q_2^0)$, each $v_{2,i} \in \Lout(q_2^{i-1},\rho_{w_{2,i}},q_2^i)$, and $v_3 \in \Lout(q_2^\ell,\rho_{w_3},q)$ for some states $q_2^0, \ldots, q_2^\ell$.

Since $\ell \ge 3|v|$, there must be an $i < 2|v|$ with $v_{2,i} = v_{2,{i+1}} = \varepsilon$. By definition of profiles we obtain $(q_2^{i-1},\varepsilon, q_2^i) \in P_{u_2}$ and $(q_2^i,\varepsilon, q_2^{i+1}) \in P_{u_2}$.  Since $P_{u_2}$ is idempotent by definition of pump-words, we obtain $(q_2^{i-1},\varepsilon, q_2^{i+1}) \in P_{u_2}$, and thus $q_2^i = q_2^{i+1}$ (see Lemma~\ref{lem:profile-properties}). We conclude that $(q_2^i, \varepsilon, q_2^i) \in P_{u_2}$, and $q_2^j = q_2^i$ and $v_{2,j} = \varepsilon$ for each $i \le j \le \ell$. 

Let $q_2 := q_2^i$, $w' =  w_1w_{2,1} \cdots w_{2,{2|v|}}$ and $v' = v_1v_{2,1} \cdots v_{2,{2|v|}}$. As a consequence of the above considerations, $v' \in \Lout(p,\rho_{w'},q_2)$.
Note that $(u_1,M_1)(u_2,M_2)^{2|v|}$ is a safe pump-word of smaller complexity than the one of $(u,M)$ (the edge enclosing $u_2$ in $(u,M)$ was removed, and the iteration of $(u_2,M_2)$ only adds edges of smaller nesting depth). Furthermore $w' \in \Pump{3|v|}{(u_1,M_1)(u_2,M_2)^{2|v|}}$, and $3|v| \ge 3|v'|$. We can thus apply the induction and obtain that $\Lout(p,\xi_1(\xi_2)^{2|v|},q_2) \not= \emptyset$. Since  $P_{\xi_2} = P_{\xi_2}P_{\xi_2}$ and $P_{\xi_1} = P_{\xi_1}P_{\xi_2}$, we conclude that $\Lout(p,\xi_1,q_2) \not= \emptyset$. 

From $(q_2,\varepsilon,q_2) \in P_{u_2}$ (as deduced above), we obtain $\varepsilon \in \Lout(q_2,\rho_{u_2},q_2)$, which implies $\varepsilon \in \Lout(q_2,\reduce{\rho_{u_2}},q_2)$ (see Lemma~\ref{lem:empty-word-traversal}). Since $\reduce{\rho_{u_2}} = \reduce{\rho_{(u_2,M_2)}} = \reduce{\xi_2}$, we conclude that $\Lout(q_2,\reduce{\xi_2},q_2) \not= \emptyset$.

Finally, note that $M_3 = \emptyset$ in the definition of safe pump-word. Hence $w_3 = u_3$, and $\xi_3 = \rho_{u_3}$. Thus $v_3 \in \Lout(q_2^\ell,\xi_3,q)$ because $P_{w_3} = P_{u_3}$.

Putting all these together, we obtain that $\emptyset \not= \Lout(p,\xi_1,q_2)\Lout(q_2,\reduce{\xi_2},q_2)\Lout(q_2,\xi_3,q) \subseteq \Lout(p,\xi_1\reduce{\xi_2}\xi_3)$ as desired. 
\end{proof}

\subsubsection*{An invariant for the strategy}

Before we start the construction of the strategy, we introduce one more terminology concerning sequential uniformisers that  simplifies the presentation below. Recall the definition of the relation $R_p^\rho$ from Section~\ref{subsec:langauges-and-relations-of-transformation-sequences}. Every output in this relation has to start with a prefix in $\Lout(p,\rho)$. We say that a sequential transducer $S$ is a \emph{$(p,\rho)$-uniformiser} if it is a uniformiser of of $R_p^\rho$ such that
\begin{itemize}
\item in every run, the first non-empty output that $S$ produces has a prefix in $\Lout(p,\rho)$, and
\item $S$ outputs the endmarker only if it reads the endmarker on the input.
\end{itemize}
So with $(p,\rho)$-uniformisers we exclude the case that the prefix from $\Lout(p,\rho)$ is built up incrementally along several transitions, and the case that an endmarker is produced on the output before the endmarker has appeared on the input. It is quite easy to see that working with $(p,\rho)$-uniformisers is not a restriction.

\begin{remark}
If there is a sequential uniformiser for $R_p^\rho$, then there is also a $(p,\rho)$-uniformiser. 
\end{remark}
\begin{proof}
First note that for each $q \in Q$ the language $\Lout(p,\rho,q)$ is regular. 
Given a sequential uniformiser $S$ for $R_p^\rho$, one can construct a $(p,\rho)$-uniformiser $S'$ as follows. It mimics the transitions of $S$ but instead of producing output, it simulates 
DFAs for each $\Lout(p,\rho,q)$ on the output that $S$ would have produced. If the DFA for $\Lout(p,\rho,q)$ reaches an accepting state, then $S'$ produces some output word in $\Lout(p,\rho,q)$, appending the possibly remaining part $v$ of the output from the last transition. 

The property concerning the endmarker is easily achieved by delaying the output of the endmarker to the transition that reads $\dashv$ on the input.
\end{proof}

%A memory state of the strategy $f_\tra$ that we construct from $\tra$ is of the form $[(d,p,\rho,\xi),(u,M)]$ where  $(d,p,\rho\xi)$ is the current vertex of the play, $|\rho|,|\xi|\le K$, and $(u,M)$ is a pump-word with $\rho_{(u,M)} = \xi$. 
We construct a strategy $f_\tra$ from $\tra$, such that an invariant is maintained in each move, from which follows that $f_T$ is winning. The invariant is stated below. Condition (Inv0) only describes some properties on the structure of the memory that is used for $f_\tra$. Condition (Inv1) implies that $f_\tra$ is a winning strategy. The conditions (Inv2) and (Inv3) state intuitively that it is still possible to uniformise the remaining relation w.r.t.\ the input and output moves that have already been played.

For a sequential transducer $S$ and an input word $w$, we write $S(w)=v$ if $S$ produces output $v$ when reading input $w$.

\begin{itemize}
\item[(Inv0)] The memory states of $f_\tra$ are of the form $[(p,\rho,d',\xi,d),(u,M)]$, where 
  \begin{enumerate}[(i)]
  \item the current game position is of the form $(p,\rho\xi,d)$ or $(p,\rho\xi,d,a)$ (the memory maintains two parts of the transformation sequence in the game position;  $f_\tra$ only applies moves of type (Out3) to the second part $\xi$);
  \item  $|\rho| \le K+1$, $|\xi| \le K$;
  \item $(u,M)$ is a safe pump-word with $\rho_{(u,M)} = \xi$ ($u$ is the input sequence that induced $\xi$, and $M$ corresponds to the applications of moves of type (Out3));
  \item if $\xi$ is saturated, then a saturation witness for $(u,M)$ is of the form $\xi_1\xi_2\xi_3$ with $\xi_3 = \varepsilon$ (new edges to $u$ will be added only at the end of the word);
  \item $A_\dom: d' \xrightarrow{u} d$ (before reading $u$, the automaton $A_\dom$ was in state $d'$). 
  \end{enumerate}
\item[(Inv1)] If $d \in F_\dom$, then $\rho = \xi = \varepsilon$ and $p \in F$.
\item[(Inv2)] If $d \notin F_\dom$, then there is a $(p,\rho)$-uniformiser, and the domain of $R_p^{\rho}$ contains all words from $L(A_\dom,d')$, which is the set of word accepted by  $A_\dom$ from state $d'$. 
\item[(Inv3)] For every $(p,\rho)$-uniformiser $S$ there is $w \in \Pump{|S|!}{(u,M)}$ such that $S(w) = \varepsilon$.
\end{itemize}

The last condition (Inv3) is used to build a $(q,\xi)$-uniformiser (ensuring (Inv2) in the new vertex), once the strategy is ready to play an output in $\Lout(p,\rho,q)$ that consumes $\rho$ from the lookahead. The choice of $|S|!$ for the number of repetitions in the pumping ensures that $S$ enters a loop on each of the pumped factors. This is useful in the proofs below.

\begin{lemma} \label{lem:invariant-implies-winning}
  If the strategy $f_\tra$ satisfies the properties (Inv0)--(Inv3), then it is a winning strategy. 
\end{lemma}
\begin{proof}
This is a direct consequence of  (Inv1).
\end{proof}

The next lemma is the key lemma for preserving the invariant in a move that produces output. It covers the case that (Inv3) fails after appending the next input letter $a$. Then the goal is to produce output that consumes $\rho$ and to transfer the property (Inv2) to $\xi \circ \tau_a$.
\begin{lemma} \label{lem:propagate-invariant}
Let $[(p,\rho,d',\xi,d),(u,M)]$ be a memory state that satisfies (Inv0)--(Inv3). Let $a \in \Sigmaend$ be such that $S(wa) \not= \varepsilon$ for all $w \in \Pump{|S|!}{(u,M)}$ for some $(p,\rho)$-uniformiser $S$. Then there exists $q \in Q$ such that 
\begin{itemize}
\item $\Lout(p,\rho,q) \not= \emptyset$, 
\item The domain of $R_q^{\xi \circ \tau_a}$ contains all words in $L(A_\dom,d_a)$ with $d_a = \delta_\dom(d,a)$, 
\item and if $a \not= \dashv$, there is a $(q,\xi \circ \tau_a)$-uniformiser.  
\end{itemize}
\end{lemma}
\begin{proof}
By (Inv3), there is  $w \in \Pump{|S|!}{(u,M)}$ such that $S(w) = \varepsilon$. Let $s_0$ be the initial state of $S$, and $s',s$ be the states of $S$ such that $S: s_0 \xrightarrow{w/\varepsilon} s' \xrightarrow{a/yz} s$, where $y \in \Lout(p,\rho,q)$ for some state $q$ (the first non-empty output contains a prefix in $\Lout(p,\rho)$). This proves the existence of $q \in Q$ such that $\Lout(p,\rho,q) \not= \emptyset$.

Since $S$ does not see the difference between a factor repeated $\ell \ge |S|!$ times or $\ell + k|S|!$ times, one can even find $w_k \in  \Pump{k|S|!}{(u,M)}$ for each $k\ge 1$ such that $S(w_k) = \varepsilon$ and furthermore, $S$ reaches the same state $s'$ after reading $w_k$ for all $k$. So $S: s_0 \xrightarrow{w_k/\varepsilon} s' \xrightarrow{a/yz} s$ for all $k$. 

Let $x \in L(A_\dom,d_a)$. Then $uax \in L(A_\dom,d')$ because:
\[
  A_\dom: d' \xrightarrow{u} d \xrightarrow{a} d_a \xrightarrow{x} F_\dom.
\]
Since the pumped words of $(u,M)$ induce the same state transformation on $A_\dom$ as $u$, we obtain that $w_kax \in L(A_\dom,d')$ for all $k$.
Let $z_x$ be the output produced by $S$ on input $x$ from state $s$. Then
\[
S: s_0 \xrightarrow{w_k/\varepsilon} s' \xrightarrow{a/yz} s \xrightarrow{x/z_x} F_S
\]
for each $k$, where $F_S$ is the set of final states of $S$.
This implies that $\delta^*(q,w_kax,zz_x) = (r_k,\varepsilon,\varepsilon)$ with $r_k \in F$ for each $k$ because $S$ is a $(p,\rho)$-uniformiser. So for each $k$, $zz_x$ has a prefix $v_k$ such that $\delta^*(q,w_ka,v_k) = \delta^*(q_k',a,\varepsilon) = (q_k,\varepsilon,\varepsilon)$ for some state $q_k$. In other words, $v_k \in \Lout(q,\rho_{w_ka},q_k)$.

Since there are only finitely many possible prefixes $v_k$ and states $q_k$,  there are infinitely many $k$ with the same $v_k$ and the same $q_k$. By choosing the corresponding subsequence of $w_1,w_2,\ldots$, we can assume that all $v_k$ and $q_k$ are the same, so we just denote them by $v_1$ and $q_1$.

We obtain that $v_1 \in \Lout(q,\rho_{w_ka},q_1)$ for all $k$. By choosing $k = 3|v_1|$, we can apply Lemma~\ref{lem:output-large-pumping}, and obtain that $\Lout(q,\xi \circ \tau_a,q_1) \not= \emptyset$. 

We can conclude that  $x \in \dom(R_q^{\xi \circ \tau_a})$ as follows. The output word $zz_x$ was shown to be the form $v_1v'$. We can replace the prefix $v_1$ by some word $v'' \in \Lout(q,\xi \circ \tau_a,q_1)$. Then $(x,v''v') \in R_q^{\xi \circ \tau_a}$.

It remains to prove that there is a $(q,\xi \circ \tau_a)$-uniformiser $S'$ in case $a \not= \dashv$. The reason for excluding $\dashv$ is that in this case there are no further transitions of $S$ from $s$.

We already proved that for each input $x \in L(A_\dom,d_a)$, the output word $zz_x$ (as above) is of the form $v_1v'$ with $v_1 \in \Lout(q,\rho_{w'a},q_1)$ for some $w' \in  \Pump{|S|!}{(u,M)}$ with $S:s_0 \xrightarrow{w'/\varepsilon} s'$, and some state $q_1$, such that $\Lout(q,\xi \circ \tau_a,q_1) \not= \emptyset$. And then $(x,v''v') \in R_q^{\xi \circ \tau_a}$ for $v'' \in \Lout(q,\xi \circ \tau_a,q_1)$.

The idea for constructing a $(q,\xi \circ \tau_a)$-uniformiser $S'$ is that we modify $S$ to detect a prefix of the output with the properties of $v_1$ as above, and then replace it by $v''$. To implement this operation, we show that the set of output words $v_1$ with the above property is regular, by using the following observations:
\begin{itemize}
\item $L_1 := \Pump{|S|!}{(u,M)}$ is regular.
\item $L_2 := \{w' \in \Sigmaend^* \mid S:s_0 \xrightarrow{w'/\varepsilon} s'\}$ is regular.
\item For each $q_1$, the set $L_{q_1}$ of output words $v_1$ such that there exists $w' \in L_1 \cap L_2$ and $v_1 \in \Lout(q,\rho_{w'a},q_1)$ is regular. An automaton for this set reads $v_1$, guesses $w'$, simulates $\tra$ from $q$ on the pair $(w',v_1)$, and simulates an automaton for $L_1 \cap L_2$ on the guessed $w'$. 
\end{itemize}
So the $(q,\xi \circ \tau_a)$-uniformiser $S'$ simulates  the transitions of $S$ starting from state $s$, and in parallel simulates for each $q_1 \in Q$ such that $\Lout(q,\xi \circ \tau_a,q_1) \not= \emptyset$ an automaton for $L_{q_1}$ on $z$ followed by the output produced in the simulation of $S$. During this simulation, $S'$ does not produce any output. As observed earlier, for some $q_1$,  the automaton for $L_{q_1}$ will eventually reach an accepting state. In this transition, $S'$ produces an output word $v''$ in $\Lout(q,\xi \circ \tau_a,q_1)$ and then continues by simply copying the remaining output of $S$.
\end{proof}

\subsubsection*{The strategy construction}

Finally, we explain how to construct $f_\tra$ such that (Inv0)--(Inv3) are ensured. 

The initial memory state is $[(q_0, \varepsilon, q_0^\dom, \varepsilon, q_0^\dom), (\varepsilon,\emptyset)]$. Then (Inv0)--(Inv3) are all satisfied:
\begin{itemize}
\item (Inv0) obviously holds.
\item (Inv1) holds because $q_0^\dom$ is not in $F_\dom$ because each word accepted by $A_\dom$ has to end with $\dashv$.
\item Since $R_{q_0}^\varepsilon = R_\tra$ and by assumption there is a sequential uniformiser for $R_\tra$, (Inv2) is satisfied.
\item (Inv3) is obviously satisfied by choosing $w = \varepsilon$ (because pumping the empty word results again in the empty word). 
\end{itemize}

Let $[(p,\rho,d',\xi,d);(u,M)]$ be the current memory state, $(p,\rho\xi,d)$ be the current position of the play, and assume that (Inv0)--(Inv3) are satisfied. Let $a \in \Sigmaend$ be the next move of Input leading to $(p,\rho\xi,d,a)$.

We now describe how $f_\tra$ selects the next moves of Output leading again to an Input vertex, and updates the memory. Later we verify that the invariant holds at the new Input vertex reached by the moves. For all cases, let $d_a = \delta_\dom(d,a)$.

\begin{enumerate}[(a)]
%\item If $p \in \Qin$, then $\rho\xi = \varepsilon$, $d' = d$, and $u=\varepsilon$. The only available move is (Out0), leading to $(p_1,\varepsilon,d_1)$ with $p_1 = \delta(p,a)$ and $d_1 = \delta_\dom(d,a)$. We update the memory to $[(p_1,\varepsilon,d_1,\varepsilon,d_1);(\varepsilon,\emptyset)]$. 
%
%for $p \in \Qout$ we show how to select the next move of Output such that the invariant is preserved. 
%\item If $|\rho| < K$ (and thus $\xi = u = x\varepsilon$, $d=d'$), then $f_T$ plays move (Out1) to $(p,\rho \circ \tau_a,d_1)$ with $d_1 = \delta_\dom(d,a)$, and updates the memory to $[(p,\rho,d_1,\varepsilon,d_1),(\varepsilon,\emptyset)]$.
\item Assume there is a $(p,\rho)$-uniformiser $S$ with $S(wa) \not= \varepsilon$ for all $w \in \Pump{|S|!}{(u,M)}$. Then the conditions of Lemma~\ref{lem:propagate-invariant} are satisfied. Let $q \in Q$ be the state with the properties from Lemma~\ref{lem:propagate-invariant}.

Let $w \in \Pump{|S|!}{(u,M)}$ such that $S(w) = \varepsilon$. The existence of $w$ is ensured by (Inv3) for the current node. 

We distinguish two cases depending on whether $a = \dashv$ or not.
  \begin{enumerate}[(1)]
  \item If $a \not= \dashv$, then play (Out2) with $v_{p,\rho,q}$, leading to vertex $(q,\xi,d,a)$, followed by the move (Out1) leading to $(q,\xi \circ \tau_a,d_a)$. 

Update the memory to $[(q,\xi \circ \tau_a,d_a,\varepsilon,d_a);(\varepsilon,\emptyset)]$.
  \item Assume $a = \dashv$ and $d_\dashv \in F_\dom$ (meaning that Input has played a word in the domain of $R_\tra$). In this case, the properties of $q$ are 
$\Lout(p,\rho,q) \not= \emptyset$ and $\dom(R_q^{\xi \circ \tau_\dashv}) = L(A_\dom,d_\dashv)$. Since  $L(A_\dom,d_\dashv) = \{\varepsilon\}$, the latter property implies that there is $r \in F$ with $\Lout(q,\xi \circ \tau_\dashv,r) \not= \emptyset$. So there must be $q'$ such that $\Lout(q,\xi,q') \not= \emptyset$ and $\Lout(q',\tau_\dashv,r) \not= \emptyset$.

Then $f_\tra$ plays (Out2) with $v_{p,\rho\xi,q'}$, leading to vertex $(q',\varepsilon,d,\dashv)$, followed by (Out0) with $v_{q',\tau_\dashv,r}$, leading to $(r,\varepsilon,d_\dashv)$. 

Update the memory to $[(r,\varepsilon,d_\dashv,\varepsilon,d_\dashv));(\varepsilon,\emptyset)]$.
  \end{enumerate}
\item Not case (a), and $\xi$ is not saturated. Then make the move (Out1) to $(p,\rho\xi \circ \tau_a,d_a)$.

Update the memory to $[(p,\rho,d',\xi \circ \tau_a,d_a);(ua,M)]$

This move is possible because $\xi = K$ would imply that $\xi$ is saturated  (Lemma~\ref{lem:length-saturated}).

\item Not case (a), and $\xi$ is saturated. Then there is a saturation witness $\xi = \xi_1\xi_2\xi_3$ with $\xi_3 = \varepsilon$. In particular, $M_3 = \emptyset$ for the  corresponding decomposition $(u,M) = (u_1,M_1)(u_2,M_2)(u_3,M_3)$ of $(u,M)$.
This means that $\xi_1\xi_2\xi_3 = \xi_1\xi_2$ is a saturation witness for $(u,M)$  as in Definition~\ref{def:safe-pump-word}. Then make the move (Out3) to $(p,\rho\xi_1\reduce{\xi_2},d,a)$ and update the memory to $[(p,\rho,d',\xi_1\reduce{\xi_2},d);(u,M')]$ with $M' = M \cup \{(|u_1|+1,|u_1|+|u_2|)\}$.
\end{enumerate}

\begin{lemma} \label{lem:invariant-preserved}
The strategy $f_T$ satisfies (Inv0)--(Inv3).
\end{lemma}
\begin{proof}
Most of the properties are obvious from the construction of the move:
\begin{itemize}
\item (Inv0) describes simple properties that directly follow from the construction of the new memory states. As only case we mention (Inv0)(iv). This follows from the fact that the move of type (Out3) is always applied when $\xi$ is saturated. This means if a new edge can be added to $(u,M)$ then it must involve the last letter that was added to $u$. Otherwise, the edge could have been added in a previous move.

\item (Inv1) clearly holds after (a)(2) by construction of the move. Furthermore, this is the only move after which the state of $A_\dom$ is in $F_\dom$.

\item (Inv2) holds after (b) and (c) because $p,\rho$, and $d'$ do not change. It holds after (a)(1) by Lemma~\ref{lem:propagate-invariant}. After (a)(2) the state of $A_\dom$ is in $F_\dom$ and hence (Inv2) is trivially satisfied.

\item (Inv3) is trivially satisfied after (a)(1) and (a)(2) because the new pump-word is $(\varepsilon,\emptyset)$. (Inv3) also holds after (b) because the new pump-word is $(ua,M)$ and the negation of (a) is exactly (Inv3) for $(ua,M)$. (Inv3) after (c) requires some more work, and is shown below.
\end{itemize}
We need to show that (Inv3) is preserved after an application of case (c) in the strategy description. So we need to show that for every $(p,\rho)$-uniformiser $S$ there is $w \in \Pump{|S|!}{(u,M')}$ such that $S(w) = \varepsilon$, where $M'$ extends $M$ by the new edge as defined in case (c).

Toward a contradiction, assume that there is $(p,\rho)$-uniformiser $S$ such that $S(w) \not= \varepsilon$ for all $w \in \Pump{|S|!}{(u,M')}$. We show that then there also exists a $(p,\rho)$-uniformiser $S'$ such that $S(w) \not= \varepsilon$ for all $w \in \Pump{|S'|!}{(u,M)}$, contradicting the assumption that (Inv3) holds for $(u,M)$. 

We construct $S'$ such that $|S'| \ge |S|$ and $S(w) \not= \varepsilon$ for all $w \in \Pump{|S|!}{(u,M)}$, which then implies $S(w) \not= \varepsilon$ for all $w \in \Pump{|S'|!}{(u,M)}$.

Let $\xi = \xi_1\xi_2\xi_3$ and $(u,M) = (u_1,M_1)(u_2,M_2)(u_3,M_3)$ be as in case (c). Recall that $\xi_3 = \varepsilon$ and thus $u_3 = \varepsilon$,  $\xi = \xi_1\xi_2$, and $(u,M) = (u_1,M_1)(u_2,M_2)$
%Words in $\Pump{|S|!}{(u,M')}$ are thus of the form $w_1(w_2)^\ell$ with  $w_1 \in \Pump{|S|!}{(u_1,M_1)}$ and $w_2 \in \Pump{|S|!}{(u_2,M_2)}$.

We describe how $S'$ works in two phases. We do not provide a formal definition of $S'$ because it should be clear that the description below can be implemented by a sequential finite state transducer.
\begin{itemize}
\item In the first phase, $S'$ mimics  $S$ on the input, while at the same time simulating an automaton for $\Pump{|S|!}{(u_1,M_1)}$. If during this phase $S$ produces output, $S'$ simply continues to mimic $S$ and cancels all other activities (recall that the goal is construct a transducer $S'$ that produces output on all words from $\Pump{|S|!}{(u,M)}$; so if $S$ produces output we are done for this input word).

If $S'$ detects that $w_1 \in \Pump{|S|!}{(u_1,M_1)}$ has been read (and no output has been produced), jump to the next phase.

\item Let $s_1$ be the state of $S$ with $S:s_0 \xrightarrow{w_1/\varepsilon} s$ (which is reached during the simulation of $S$ in the first phase). Fix some $w_2 \in \Pump{|S|!}{(u_2,M_2)}$, e.g., the shortest word in this set. Let $h \ge 0$ be maximal such that starting from $s_1$, $S$ does not produce any output on $(w_2)^h$. Note that there exists a maximal such $h$ because $S$ will produce output for $|S|!$ repetitions of $w_2$. So the situation is now as follows:
\[
S: s_0 \xrightarrow{w_1 / \varepsilon} s_1 \xrightarrow{(w_2)^h / \varepsilon} s_1'
\]

In the second phase, $S'$ continues the simulation of $S$ but now from state $s_1'$. If the first output in this simulation of $S$ is produced, it is of the form $yz$ with $y \in \Lout(p,\rho,q)$ for some state $q$. Then, $S'$ outputs $y$ and jumps to the third phase.
\item The current situation can be depicted as follows, where the second line only indicates the part of $S'$ that simulates $S$:
\[
  \begin{array}{l}
S: s_0 \xrightarrow{w_1 / \varepsilon} s_1 \xrightarrow{(w_2)^h / \varepsilon} s_1' \xrightarrow{w'/yz} s_2 \\    
S': s_0 \xrightarrow{w_1 / \varepsilon} s_1 \leadsto s_1' \xrightarrow{w'/y} s_2
  \end{array}
\] 
Note that the actual input is $w_1w'$, and that $(w_2)^h$ has only virtually be inserted in the computation of $S$.

For the third phase, let $L_{\text{Pump},{s_1}'}$ be the set of all $\hat{w} \in \Pump{|S|!}{(u_1,M_1),(u_2,M_2)^h}$ with $S: s_0 \xrightarrow{\hat{w}/\varepsilon} s_1'$. Note that $w_1(w_2)^h$ is such a word, and that $L_{\text{Pump},{s_1}'}$ is regular.

In the third phase, $S'$ continues the simulation of $S$ from $s_2$ but without producing output. Instead, it waits until the output that $S$ would have produced, including the pending $z$ from the previous transition, contains a prefix $\hat{v} \in \Lout(q,\rho_{\hat{w}},r)$ for some $r$ and some $\hat{w} \in L_{\text{Pump},{s_1}'}$ such that $\Lout(q,\xi_1(\xi_2)^h,r) \not= \emptyset$. If such a prefix $\hat{v}$ is reached, then $S'$ outputs instead a word $v_1 \in \Lout(q,\xi_1,r)$, which exists because $P_{\xi_1} = P_{\xi_1}P_{\xi_2} \cdots P_{\xi_2}$. And after that, $S'$ just copies the output of $S$.
\end{itemize}

We need to argue that $S'$ is a $(p,\rho)$-uniformiser with the desired property, and  for the third phase, that the output of $S$ contains such a prefix $\hat{v}$ that can be substituted by $v_1 \in \Lout(q,\xi_1(\xi_2)^h,r)$.

For this purpose, let $x$ be an input that takes $S$ from $s_2$ into a final state $s_3$, and let $v$ be the corresponding output. The whole computation of $S$ that is used in the description above looks as follows:
\[
S: s_0 \xrightarrow{w_1 / \varepsilon} s_1 \xrightarrow{(w_2)^h / \varepsilon} s_1' \xrightarrow{w'/yz} s_2 \xrightarrow{x/v} s_3    
\]
We note that by definition of $R_\rho^p$, the pair $(w_1(w_2)^hw'x,zv)$ is accepted by $\tra$ from state $q$. Now assume that $zv = \hat{v}v'$ has a prefix $\hat{v}$ as described in the third phase. Then we can replace $w_1(w_2)^h$ in the above computations of $S$ by $\hat{w}$:
\[
S: s_0 \xrightarrow{\hat{w} / \varepsilon} s_1' \xrightarrow{w'/yz} s_2 \xrightarrow{x/v} s_3    
\]
and we thus know that $(\hat{w}w'x,\hat{v}v')$ is accepted by $T$ from $q$, and furthermore $\delta^*(q,\hat{w},\hat{v}) = (r,\varepsilon,\varepsilon)$ (for $r$ as in the third phase). By the choice of $v_1 \in \Lout(q,\xi_1,r)$ and the fact that $w_1$ is compatible with $\xi_1$, we obtain that $\delta^*(q,w_1,v_1) = (r,\varepsilon,\varepsilon)$.

The computation of $S'$ can be sketched  as follows (again only showing the states of the simulation of $S$ inside $S'$):
\[
S': s_0 \xrightarrow{w_1 / \varepsilon} s_1 \leadsto s_1' \xrightarrow{w'/y} s_2  \xrightarrow{x/v_1v'} s_3 
\] 
Since $\delta^*(q,w_1,v_1) = (r,\varepsilon,\varepsilon)$, we obtain that $(w_1w'x,v_1v')$ is accepted by $\tra$ from $q$. Hence $S'$ is indeed a $(p,\rho)$-uniformiser.

It remains to prove the existence of the prefix $v_1$ in the output $zv$.
Let us verify by applying  Lemma~\ref{lem:output-large-pumping} that the output must contain such a prefix.

Because the pumping factors for $w_1$ and $w_2$ are at least $|S|!$, one can safely pump the factors further by multiples of $|S|!$ without $S$ noticing the difference. Hence, for each $k$ one can find some $\hat{w} \in \Pump{k}{(u_1,M_1)(u_2,M_2)^h}$ with $S:s_0 \xrightarrow{\hat{w}/\varepsilon} s_1'$.
Thus, there is a prefix $v_1$ of $zv$ and some  $\hat{w} \in \Pump{3|v_1|}{(u_1,M_1)(u_2,M_2)^h}$ with $v_1 \in \Lout(q,\rho_{\hat{w}},r)$ for a state $r$. Then Lemma~\ref{lem:output-large-pumping} yields that $\Lout(q,\xi_1(\xi_2)^h,r) \not= \emptyset$.
\end{proof}

Lemma~\ref{lem:drat-uniformiser-to-strategy} now follows from
Lemma~\ref{lem:invariant-implies-winning} and
Lemma~\ref{lem:invariant-preserved}.

We can now finish the proof of Theorem~\ref{the:drat-unif}: The game $\game_\tra$ can be effectively constructed (the number $K$ from Lemma~\ref{lem:length-saturated} is computable from $\tra$). According to Lemma~\ref{lem:drat-uniformiser-to-strategy} and Lemma~\ref{lem:drat-strategy-to-uniformiser}, there is a sequential uniformiser of $\tra$ if, and only if, player Output has a winning strategy in $\game_\tra$. Determining the player that has a winning strategy in the safety game $\game_\tra$ can be done in polynomial time in the size of the game graph (see, e.g., \cite{GradelTW02}). 
%%%%%%%%%%%%%%%%%%%%%%%%%%%%%%%%%%%%%%%%%
\subsection{Bounded Delay Uniformisation}
\label{subsec:drat-k-delay}
%%%%%%%%%%%%%%%%%%%%%%%%%%%%%%%%%%%%%%%%%

Based on the the decidability proof we can now prove the following theorem.

\dratkdelay*

\begin{proof}
  Let $T$ be deterministic transducer. If there is a sequential uniformiser for $T$, then there is a winning strategy in $\game_\tra$ by Lemma~\ref{lem:drat-uniformiser-to-strategy}. Consider the sequential uniformiser $S$ that is constructed in the proof of Lemma~\ref{lem:drat-strategy-to-uniformiser} from a positional winning strategy in $\game_\tra$. 

For each input word $u$ in the domain of $R_\tra$, there is a unique computation of $S$ on $u$. Let $v$ be the output produced by this computation. This induces a unique computation of $\tra$ on $(u,v)$. 

Let $u'$ be a prefix of $u$ and $S$ produces the prefix $v'$ of $v$ as output while reading $u'$, and reaches the state $(p,\rho,d)$. Then Lemma~\ref{lem:game-property} yields that $u'=xw$ with $\delta^*(q_0,u',v') = (p,w,\varepsilon)$ and $w \in \Lin(\rho)$. This implies that the output produced by $S$ on $u'$ is a prefix of the output that $T$ reads in its computation for $u$ and $v$. So the delay for the prefix $u'$ is the part of the output word that $T$ is ahead: If $v''$ is such that $\delta^*(q_0,u',v'v'') = (p,\varepsilon,\varepsilon)$, then $v''$ is the delay. 

By the rules of the game, $v''$ is a word in $\Lout(p,\rho)$ that is produced by moves of type (Out2). The number of these moves that produce $v''$ is at most $|\rho|$, and in each move a word of bounded length is produced. Therefore, the length of $v''$ is bounded by $(2K+1)M$, where $M$ is the maximal length of a word annotating an edge of the game graph, and $2K+1$ is the maximal length of $\rho$. This number can be computed. 
\end{proof}

%%% Local Variables:
%%% TeX-master: "main"
%%% End:

\end{document}